\documentclass[twoside,leqno]{article}

\usepackage[letterpaper]{geometry}

\usepackage{siamproceedings}

\usepackage[T1]{fontenc}
\usepackage{amsfonts}
\usepackage{graphicx}
\usepackage{epstopdf}
\usepackage{enumitem}

\newsiamremark{remark}{Remark}
\newsiamremark{hypothesis}{Hypothesis}
\crefname{hypothesis}{Hypothesis}{Hypotheses}
\newsiamthm{claim}{Claim}

\usepackage{amsopn}

\usepackage{xcolor,xspace,nicefrac,physics,tcolorbox,enumitem}
\usepackage{amssymb,mathtools}

\usepackage[suppress]{color-edits}
\definecolor{@gray}{HTML}{edc3c5}
\addauthor[Eva]{et}{cyan}
\addauthor[Sid]{sb}{orange}
\addauthor[Giannis]{gf}{teal}
\addauthor[David]{dl}{violet}
\addauthor[\textbf{TODO}]{todo}{purple}

\usepackage{tikz}
\usetikzlibrary{fadings,patterns,shadows.blur,shapes}

\newcommand{\citet}[1]{\cite{#1}}

\usepackage{algorithmic}

\crefname{algorithm}{Mechanism}{Mechanisms} 
\Crefname{algorithm}{Mechanism}{Mechanisms} 

\DeclareMathOperator*{\E}{\mathbb{E}}
\DeclareMathOperator*{\argmin}{arg\,min}

\newcommand{\lnash}{\lambda_{\mathtt{NASH}}}
\newcommand{\lrob}{\lambda_{\mathtt{ROB}}}
\newcommand{\mechanism}{\hyperref[alg:mechanism]{\color{black}Budgeted Robust Border}\xspace}
\newcommand{\BRB}{\hyperref[alg:mechanism]{\color{black}BRB}\xspace}
\newcommand{\generalmechanism}{\hyperref[alg:general_mechanism]{\color{black}Generalized Budgeted Border}\xspace}
\newcommand{\anytimemechanism}{\hyperref[alg:anytime_mechanism]{\color{black}Any Time Budgeted Border}\xspace}

\DeclareMathOperator{\poly}{poly}

\let\epsilon\varepsilon

\begin{document}

\newcommand{\acknowledgements}{\thanks{
David X. Lin was supported by NSF grant ECCS-1847393.
Siddhartha Banerjee was supported in part by AFOSR grant FA9550-23-1-0068, ARO MURI grant W911NF-19-1-0217, and NSF grants ECCS-1847393 and CNS-195599.
Giannis Fikioris was supported in part by the Google PhD Fellowship, the Onassis Foundation – Scholarship ID: F ZS 068-1/2022-2023, and ONR MURI grant N000142412742.
\'Eva Tardos was supported in part by AFOSR grant FA9550-23-1-0410, AFOSR grant FA9550-231-0068, and ONR MURI grant N000142412742.
}}

\title{\Large Robust Equilibria in Shared Resource Allocation via Strengthening Border's Theorem\acknowledgements}
    \author{David X. Lin\thanks{Cornell University (\email{dxl2@cornell.edu}, \email{sbanerjee@cornell.edu}, \email{gfikioris@cs.cornell.edu}, \email{eva.tardos@cornell.edu}).}
    \and Siddhartha Banerjee\footnotemark[2]    
    \and Giannis Fikioris\footnotemark[2]
    \and \'Eva Tardos\footnotemark[2]}

\date{}

\maketitle

\begin{abstract}
We consider repeated allocation of a shared resource via a non-monetary mechanism, wherein a single item must be allocated to one of multiple agents in each round. We assume that each agent has i.i.d. values for the item across rounds, and additive utilities. Past work on this problem has proposed mechanisms where agents can get one of two kinds of guarantees: $(i)$ (approximate) Bayes-Nash equilibria via linkage-based mechanisms which need extensive knowledge of the value distributions, and $(ii)$ simple distribution-agnostic mechanisms with robust utility guarantees for each individual agent, which are worse than the Nash outcome, but hold irrespective of how others behave (including possibly collusive behavior). 
Recent work has hinted at barriers to achieving both \emph{simultaneously}. 
Our work however establishes this is not the case, by proposing the first mechanism in which each agent has a natural strategy that is both a Bayes-Nash equilibrium and also comes with strong robust guarantees for individual agent utilities.
Our mechanism comes out of a surprising connection between the online shared resource allocation problem and implementation theory, and uses a surprising strengthening of Border’s theorem. In particular, we show that establishing robust equilibria in this setting reduces to showing that a particular subset of the Border polytope is non-empty. We establish this via a novel joint Schur-convexity argument. This strengthening of Border's criterion for obtaining a stronger conclusion is of independent technical interest, as it may prove useful in other settings.

\end{abstract}

\section{Introduction} 
\label{sec:intro}

Consider a single indivisible public resource being repeatedly allocated between multiple agents -- for example, a scientific instrument shared by multiple university labs.
Since the resource is public, its allocation should be determined ideally without using monetary transfers.
Moreover, the principal wants to allocate the resource in a way that is both efficient (so labs get the instrument only when their need is great) and fair (so that the amount of time each lab gets to is roughly proportional to some pre-determined share). Each individual agent is of course self-serving, and so some mechanism is required to encourage agents to request for the resource only when they need it the most. And in the absence of money, the principal may not be able to verify agents' reports, or have any strong beliefs about their valuations and actions. 

The above problem was first studied under a simple model by~\cite{guo2009competitive}, with a single indivisible item per round, and agents with random valuations across rounds. Early work on this setting~\cite{guo2009competitive,cavallo2014incentive,balseiro2019multiagent,gorokh2021monetary} focused on the question of how efficient a non-monetary mechanism could be in such a setting; as we discuss in~\cref{ssec:related}, this led to several mechanisms with near-efficient Bayes-Nash equilibria. All these mechanisms, however, critically depend on knowing the exact value distributions.

More recently,~\citet{gorokh2019remarkable} pointed out that in non-monetary settings, the inability to make interpersonal comparisons makes the knowledge of exact value distributions unlikely. This led to a line of work~\cite{gorokh2019remarkable, banerjee2023robust, fikioris2023online} focusing on mechanisms that, under minimal assumptions, have \emph{robust} individual-level guarantees, which hold irrespective of how other agents behave. On the positive side, such guarantees adhere closer to Wilson's doctrine that mechanisms should be as `detail-free' (i.e., distribution agnostic) as possible~\cite{wilson1985game}. On the other hand, it is unclear how predictive such results are of agent behavior; these works do not show if their mechanisms admit any simple equilibria, and in general, this is a difficult task in any repeated mechanism.  
More surprisingly, it was recently reported~\cite{chido} that in one of the simplest such mechanisms called Dynamic Max-Min Fair sharing (DMMF), there is in fact \emph{no equilibrium} under a natural class of strategies (roughly speaking, where each agent's actions only depend on their value in each round). 
While this result involves a technical characterization of a certain high-dimensional Markov chain induced under DMMF, it appears to support a natural critique of robustness results: that they involve strategies that are too pessimistic, and hence not supported in an equilibrium.

So is it possible to design mechanisms that unite the above two streams for the setting of~\cite{guo2009competitive}? We define this formally in \cref{ssec:ideal_benchmarks}, but at a high level, we seek mechanisms that support a \emph{good robust equilibrium}: a strategy that simultaneously satisfies the following: 
\begin{tcolorbox}[standard jigsaw,opacityback=0]
    \begin{enumerate}[leftmargin=*]
        \item \textbf{Equilibrium Performance}: Assuming all agents follow the strategy, no one has an incentive to deviate. Moreover, under this equilibrium, every agent enjoys high utility.
        
        \item \textbf{Robust Performance}: If an agent follows the strategy, but others act arbitrarily (maybe even collusively), then the agent still enjoys some (relatively high) utility guarantee.
    \end{enumerate}
\end{tcolorbox}
This definition is of course under-specified in terms of what we mean by `good': note that never allocating, or allocating uniformly at random, supports any strategy both as an equilibrium, and also gives every agent a (weak) robust guarantee. Our main result however is a \emph{new mechanism for repeated non-monetary allocation} which admits a simple strategy that is a Bayes-Nash equilibrium in the infinite-horizon limit (and close to Bayes-Nash equilibrium in the finite case) with high equilibrium and robust performance. Moreover, we achieve this via a surprising strengthening of Border's theorem, which may prove useful in other settings.


\subsection{Overview of our Mechanism and Main Result}

We consider the setting of~\citet{guo2009competitive}, with a horizon of $T$ rounds, a single indivisible item to allocate per round, and $n$ agents with random valuations. Agent $i$'s values $V_i[t]$ are i.i.d. across rounds and independent across agents.
Building on the ideas of~\cite{gorokh2019remarkable}, we eschew efficiency to focus on \emph{share-based} guarantees: each agent $i$ has a pre-determined fair share $\alpha_i$ (with $\sum_i\alpha_i=1$), which then allows us to define a per-user utility benchmark, rather than compare utilities across users. 
With an $\alpha_i$ share, agent $i$'s best hope is to get her favorite $\alpha_iT$ rounds (i.e., rounds $t\in[T]$ where $V_i[t]$ is in the top $\alpha_i$ quantile of her value distribution). Formally (\cref{def:alpha_ideal_utility}), an agent's \emph{ideal utility} $v_i^\star$ with fair share $\alpha_i$ is the maximum utility she can get if awarded the resource with probability at most $\alpha_i$. 


Our aim is to develop a mechanism that admits a \emph{robust equilibrium}: a joint profile of simple strategies for the agents $(\pi_1,\pi_2,\ldots,\pi_n)$ that is simultaneously a Bayes-Nash equilibrium where every player receives per-round utility at least $\lnash v^{\star}_i$ 
(which we refer to as a $\lnash$-good Bayes-Nash equilibrium) 
and at the same time the same strategy ${\pi}_i$ guarantees that agent $i$ receives per-round utility at least $\lrob v^{\star}_i$, irrespective of how other agents act, including possibly collusive actions (which we call $\lrob$-robust). 
We want strategies to be simple so that they are more meaningful in practice. 
Note also that since we benchmark against the ideal utility $v_i^{\star}$ instead of the first-best outcome, we may not be able to get $\lnash=1$.
For example, in the case where all agents have equal shares and i.i.d values $V_i[t]\sim \textrm{Bernoulli}(1/n)$, the best possible is $\lnash \leq 1- (1-1/n)^n \approx 1 - 1/e$ \cite{banerjee2023robust}.

Within this backdrop, we propose a new mechanism we call \textbf{Budgeted Robust Border (\BRB)} (see~\cref{alg:mechanism}). At a high level, the mechanism has two simple components:
\begin{enumerate}[nosep]
    \item \emph{Budget-regulated bidding}: Each agent has a budget of $B_i\approx \alpha_i T + o(T)$ tokens, which regulates the number of times she can bid for the item over $T$ rounds.
    \item \emph{Probabilistic Allocation}: At time $t$, given the set of bidders $S[t]=S\subseteq[n]$, the principal allocates the item to some agent $i\in S$ according to some pre-decided probability $p_i^S$.
\end{enumerate}
For the exact expression for the budget $B_i$ see~\cref{prop:approximate_nash}; the additional $o(T)$ budget is needed to get high probability guarantees. The allocation probabilities $\{p_i^S\}$ are more complex (and indeed, that is where the main novelty of our work lies, as we discuss next). 
However, at this point, we can already summarize our main result, as follows (see~\cref{thm:nash}):
\vspace{-0.1cm}
\begin{tcolorbox}[standard jigsaw,opacityback=0]
Under the \BRB Mechanism, the policy where player $i$ bids in round $t$ whenever her value $V_i[t]$ is in the top $\alpha_i$-quantile (and subject to her budget), is simultaneously:
\begin{itemize}[nosep]
    \item Robust with factor $\lrob \geq 1/2$
    \item A Bayes-Nash equilibrium with $\lnash=1-\prod_{j=1}^n(1-\alpha_j)\geq 1-1/e\approx 0.63$
\end{itemize}
\end{tcolorbox}
\noindent In~\cref{lem:centralallocationworstcase}, we prove the above $\lnash$ is minimax optimal for any vector of fair shares $\{\alpha_i\}_{[n]}$. $\lrob=1/2$ matches the one in~\cite{gorokh2019remarkable,banerjee2023robust,fikioris2023online}, which was the best-known until quite recently, when \citet{lin2025online} provided a more complicated strategy in the mechanism of \cite{gorokh2019remarkable} that obtains $\lrob=2-\sqrt2\approx 0.59$. However, everyone playing their robust strategy is not an equilibrium. Moreover, their mechanism is more complicated, so equilibrium behavior is hard to analyze, and only robustness guarantees are known.

\subsection{Overview of our Techniques}

Our main conceptual idea is a connection between repeated non-monetary allocation and implementation theory (particularly, Border's theorem). 
Our main technical novelty is modifying the Border flow network to select $1/2$-robust equilibria.
This strengthening of Border's theorem to show that this restricted flow polytope is also guaranteed to be non-empty can be of independent interest.
Our mechanism uses three main ideas: $(i)$ a repeated all-pay mechanism to control bidding rates, $(ii)$ guaranteeing good interim allocations via the regular Border condition, and $(iii)$ strengthening the Border condition to guarantee robustness.
We now briefly describe these.

We start by first trying to get good equilibria. In the infinite-horizon setting, our idea is as follows: We restrict agents to only bid or not in each round, and further restrict them to bid at most an $\alpha_i$ fraction of the time. 
Our mechanism aims to guarantee agent $i$ a good \emph{interim allocation} conditioned on them bidding (the probability that they are allocated), assuming every other agent $j$ bids independently with probability at most $\alpha_j$.
Finally, to pass to the finite horizon, we use an idea from~\cite{gorokh2021monetary} and employ an all-pay mechanism with budgets $\{\alpha_iT+o(T)\}$, to ensure that agents are not incentivized to bid at a rate higher than $\alpha_i$, while also not running out of budget with high probability.

To ensure good interim allocations, we use probabilistic allocation with preset probabilities $p_i^S$ for assigning the item to agent $i$ when the set $S$ of agents bid for each subset $S\subseteq[n]$, and agents $i\in S$. 
These are chosen to ensure each agent's bids are accepted with probability at least $\lnash$, assuming each agent $i$ bids independently with probability~$\alpha_i$. 
The bid budgets and probabilistic allocation serve to somewhat insulate each agent from the actions of others.
Consequently, in equilibrium, agent $i$'s best response is to bid when her value is in the highest $\alpha_i$ quantile, and which makes bids independent across agents (since valuations are independent). Border's theorem (see~\cref{thm:alpha_border}) provides necessary and sufficient conditions for which interim allocation probabilities are feasible;
in particular, we use this to show we can always guarantee $\lnash\geq 1-\prod_j(1-\alpha_j) \ge (1-1/e)\approx 0.63.$ (see \cref{thm:nash,thm:worst_case_interim_allocation_probabilities}). 
Our guarantee is in fact minimax optimal for any vector of fair shares $\{\alpha_i\}_{[n]}$: in \cref{lem:centralallocationworstcase} we show a valuation profile that gives a matching upper bound even if the principal knows agents' realized values.

The challenge however is that many interim allocation probabilities admitted by Border's theorem are not robust: if all other agents collude against agent $i$, they can severely limit agent $i$'s allocation. In fact, different allocation rules inducing the same interim allocation can have vastly different robustness guarantees (see~\cref{ssec:interim_allocation_probabilities_do_not_guarantee_robustness}).
Even computing $\lrob$ for a given allocation rule appears difficult, as we need to search over all potentially correlated bidding schemes by the other agents.

Our main technical novelty is in identifying the subset of interim allocations allowed by Border's theorem which are $1/2$-robust.
To do this, we consider any chosen agent $i$, and characterize the \emph{bang-per-buck} (in terms of blocking agent $i$ per bid token spent) of each possible set $S\subseteq [n]\setminus\{i\}$ of colluding agents bidding in a round.
Using this, we show that the critical case is allocating when only 2 agents bid (colluding agents can always make this the dominant case). We show that it is possible to make the mechanism $1/2$-robust (matching \cite{gorokh2019remarkable,banerjee2023robust,fikioris2023online}), by modifying the Border flow network by adding appropriate bounds to edges going out of each \emph{doubleton} subset $\{i,j\}\subset [n]$ (\cref{fig:modifiedborderscriterionflownetwork}). This leads to a modification of Border's criterion for ensuring the resultant network still supports the same maximum amount of flow (\cref{lem:specific_border_modification}). 
Finally, in~\cref{lem:key_lemma}, we demonstrate the resultant polytope is non-empty, via establishing the Schur-convexity/concavity of a certain bivariate function that comes out of our new criterion.

Finally, in \cref{sec:computation}, we demonstrate how to efficiently sample from a probabilistic allocation rule in the above polytope using the techniques of \cite{bhalgat2013optimal}.
Specifically, we design a linear program whose feasible region corresponds to the above polytope.
However, a solution of this LP is of exponential size (one probabilistic allocation rule for every $S \subseteq [n]$), which means that solving it is hopeless.
Instead, we calculate a dual solution of polynomial size, which we can use to calculate the desired allocation probabilities for a given set $S$.
Specifically, we calculate dual solutions by running an instance of the Hedge algorithm. These dual solutions can be used to sample from allocation probabilities $p_i^S$ that result in the $1/2$-robustness and equilibrium utility guarantees that are only $O(\sqrt{\log n/K)}$ away from the optimal, where $K$ is the number of steps of the Hedge algorithm.

Taken together, our arguments demonstrate allocation probabilities that simultaneously guarantee for all agents at least $(1-1/e)$ fraction of their ideal utility at Nash equilibrium and, at the same time, a $1/2$ fraction of their ideal utility even if other agents collude. While the allocation rules in the general case are somewhat involved, we illustrate them explicitly in two simple corner cases -- $2$ agents with arbitrary shares in \cref{sec:twoagents}, and $n$ agents with equal shares in \cref{sec:robustness}.

\subsection{Related Work}
\label{ssec:related}

Allocating shared resources without money is a foundational problem in economics. The intersection with worst-case analysis and approximation was first explored by \citet{procaccia2009approximate}, sparking renewed interest in this area. Repeated allocation of an indivisible item with random valuations was introduced by \citet{guo2009competitive}, and has since been successful in applications like course allocation \cite{budish2017course}, food banks \cite{prendergast2022allocation}, and cloud computing \cite{vasudevan2016customizable}.

Prior work in this setting falls into two main streams. The first focuses on emulating monetary mechanisms using non-monetary tools \cite{guo2010,cavallo2014incentive,gorokh2021monetary,balseiro2019multiagent}. This culminates in the black-box reduction of \citet{gorokh2021monetary}, which uses a repeated all-pay auction to emulate any monetary mechanism with vanishing efficiency loss. Their approach builds on the idea of “linking decisions” from \citet{jackson}, where simultaneous execution of multiple mechanisms induces equilibrium behavior. However, these works all require full knowledge of value distributions and offer no guarantees under off-equilibrium behavior.

The second stream, initiated by \citet{gorokh2019remarkable}, shifts focus to robust, individual-level guarantees without assuming distributional knowledge. They introduce the notion of an agent’s \emph{ideal utility} and show that a simple threshold strategy in a repeated first-price auction with artificial credits guarantees each agent at least half of their ideal utility. This $1/2$-robustness guarantee has since been matched and extended: \citet{banerjee2023robust} provide a simpler reserve-price mechanism, and \citet{fikioris2023online} match the guarantee via Dynamic Max-Min Fairness (DMMF), a non-market mechanism that allocates to the agent with the least normalized wins. DMMF even achieves stronger guarantees under mild distributional assumptions.
However, recent work by~\citet{chido} shows that DMMF lacks equilibrium support under natural threshold strategies, even with just two agents, thereby raising questions about its practical relevance. More recently, \cite{lin2025online} improved the $1/2$ factor to $2-\sqrt2\approx0.59$ by randomizing the amount of artificial credits bid when above the threshold (but again, with no equilibrium guarantees).

Our approach draws on two key tools: repeated all-pay auctions \cite{gorokh2021monetary} and the use of Border’s theorem \cite{border1991implementation,border2007reduced} to implement interim allocations. While Border’s theorem has been instrumental in algorithmic mechanism design \cite{cai2012optimal,cai2012algorithmic,alaei2012bayesian,bhalgat2013optimal}, our use is existential rather than algorithmic—we strengthen the theorem to establish the existence of robust equilibria. This is closest in spirit to \citet{che2013generalized}, who modify allocation mechanisms to enforce lower and upper bounds on allocation probabilities. However, our setting is more complex: we must modify the Border polytope to penalize arbitrary agent behavior while preserving feasibility, a challenge not present in their capacity-based formulation.


\section{Preliminaries}

We consider the following simple setting of repeatedly allocating a single, indivisible item among $n$ agents; this was first introduced in the work of~\cite{guo2009competitive}, and has since been widely studied, as we discuss in~\cref{ssec:related}.
At each discrete time-step $t=1,2,\dots,T$, a principal receives a new item and must select an agent $i\in[n]$ to allocate the item to (or not allocate the item).

We assume that each agent $i$ has a private value $V_i[t]$ for the item at time $t$, where $V_i[t]\sim\mathcal F_i$ for some value distribution $\mathcal F_i$ not depending on $t$.
We assume that the $V_i[t]$ are nonnegative, bounded, and independent across both agents and time.
Let $X_i[t]$ be the indicator that agent $i$ was allocated the item at time $t$.
At time $t$, the agent gets utility $U_i[t] = V_i[t]X_i[t]$.
We assume that utilities are additive over time so that an agent's total utility after $T$ time periods is $\sum_{t=1}^T U_i[t]$.

\subsection{Ideal Utility and Benchmarks} \label{ssec:ideal_benchmarks}

As in previous work in this setting \cite{gorokh2019remarkable,banerjee2023robust,fikioris2023online}, we assume each agent $i$ has an exogenously defined fair share $\alpha_i$, where $\alpha_i \ge 0$ and $\sum_{i=1}^n \alpha_i=1$.
The fair share $\alpha_i$ of agent $i$ represents the fraction of rounds we want agent $i$ to win the item under ideal circumstances.
As done by previous work, we use the benchmark of ideal utility. The ideal utility of an agent $i$ is the maximum expected utility they can obtain from a single round if they can obtain the item simply by requesting it, but they are only allowed to request it with probability at most their fair share $\alpha_i$.
Formally, the (per-round) ideal utility is the following.
\begin{definition}[Ideal Utility]
    \label{def:alpha_ideal_utility}
    Agent $i$'s ideal utility is the value of the following maximization problem over measurable $\rho_i:[0,\infty)\to[0,1]$:
    \begin{equation*}
        \begin{split}
            v_i^\star = \max\,\E_{V_i\sim\mathcal F_i}[V_i\rho_i(V_i)] \quad\text{subject to}\quad \E_{V_i\sim\mathcal F_i}[\rho_i(V_i)]\leq \alpha_i.
        \end{split}
    \end{equation*}
\end{definition}
We seek to define a mechanism $M$ where agent $i$ can achieve high average expected utility $\frac1T\sum_{t=1}^T \mathbb E[U_i[t]]$. Specifically, we want that agent $i$ to achieve a high fraction of her ideal utility, i.e., $\frac1T\sum_{t=1}^T \mathbb E[U_i[t]] \geq \lambda_i v_i^\star$ for some $\lambda_i$ as large as possible.
As discussed in~\cite{gorokh2019remarkable}, with additional knowledge of valuation distributions, such guarantees naturally translate into efficiency guarantees (by setting fair shares to maximize the ex-ante welfare relaxation). However these guarantees are equally compelling as individual-level guarantees in non-monetary settings, characterizing what fraction of her top rounds each agent $i$ can realize while leaving the resource on at least $(1-\alpha_i)T$ days for the others to use.

We study what fraction of her ideal utility an agent $i$ can guarantee in two settings:
First, we examine the performance of Bayes-Nash equilibria, where each agent is trying to maximize their own utility.
Second, we examine what fraction of her ideal utility agent $i$ can guarantee robustly, i.e., even if the other agents $j\neq i$ are playing adversarially and collude to harm agent $i$ without regard for their own utilities.
Critically, we want an agent to achieve both the above with the same (ideally simple) strategy.
As in~\cref{sec:intro}, we want a mechanism that achieves the following.
\begin{tcolorbox}[label=tcbox:goal,standard jigsaw,opacityback=0]
    \textbf{$\lnash$-Nash $\lrob$-robust equilibria}: A mechanism with policy profile $(\pi_1, \pi_2, \ldots, \pi_n)$ such that
    \begin{itemize}[leftmargin=*]
        \item $(\pi_1, \pi_2, \ldots, \pi_n)$ forms a Nash equilibrium as $T \to \infty$ in which each agent $i$ gets a $\lnash$ fraction of their ideal utility.
        \item Policy $\pi_i$ is \textit{$\lrob$-robust} with $\lrob$ not much smaller than $\lnash$: it guarantees agent $i$ a $\lrob$ fraction of her ideal utility even if agents $j \ne i$ collude and act adversarially.
    \end{itemize}
\end{tcolorbox}


To understand how the ideal utility benchmark behaves under Bayes-Nash equilibria, we first present a bound on what fraction $\lnash$ of ideal utility agents can get under ideal conditions. This will help later certify the minimax optimality of our mechanism in the equilibrium setting. Note that this bound is mechanism-agnostic, as in our example below, the mechanism has access to the agents' realized values.

\begin{lemma}
    \label{lem:centralallocationworstcase}
    If each agent $i$ has value distribution $\mathcal F_i = \mathrm{Bernoulli}(\alpha_i)$, it is impossible to guarantee every agent $i$ a $\lambda$ fraction of their ideal utility in expectation for $\lambda > 1 - \prod_{j\in[n]}(1-\alpha_j)$, even if the mechanism knows $V_i[t]$ for all $i,t$ before round $1$. Note that $\inf_{n, \alpha_1, \ldots, \alpha_n}1 - \prod_{j\in[n]}(1-\alpha_j) = 1 - 1/e$.
\end{lemma}

We prove the lemma in \cref{ssec:appendix_central_allocation_worst_case_proof}.

\section{Mechanism and Proposed Strategy}
\label{sec:mech_and_strategy}

We now give our mechanism.
Every agent $i$ bids for the item or not in every round and can make (approximately) at most $\alpha_i T$ bids in total.
In every round where agents bid for the resource, our mechanism randomly allocates the resources to one of those agents.
For the simpler randomized allocation rule in the case where there are only $2$ agents, see \cref{sec:twoagents}\footnote{When $n=2$, in \cref{sec:twoagents}, we manage to get optimal performance $\lnash=\lrob=1 - (1-\alpha_1)(1-\alpha_2)$, as in \cref{lem:centralallocationworstcase}.}.
In general, if agents $S \subseteq [n]$ bid for the resource, we allocate the resource to agent $i \in S$ with probability $p_i^S$.
We call the $p_i^S$'s \textit{allocation probabilities} and describe their values later (\cref{thm:worst_case_interim_allocation_probabilities}).
Our formal mechanism can be found in \cref{alg:mechanism}, where we also provide a little extra budget for each agent, to ensure they do not run out with high probability when following our proposed equilibrium strategy, which we introduce next.

\floatname{algorithm}{Mechanism}
\begin{algorithm}
    \caption{Budgeted Robust Border (BRB) Mechanism}
    \label{alg:mechanism}
    \begin{algorithmic}
        \REQUIRE Fair shares $\alpha_i$, time horizon $T$, allocation probabilities $(p_i^S)_{i \in S}\; \forall S \subseteq [n]$, budget slack $\delta_i^T$.
        \STATE Endow each agent $i$ with a budget $B_i[1] = \alpha_i(1+\delta_i^T)T$ of bid tokens.
        \FOR{$t=1,2,\dots,T$}
            \STATE Agents submit bids $b_i^t\in\{0,1\}$ and budgets are enforced: $b_i^t\gets 0$ for each $i$ such that $B_i[t]\leq 0$.
            \STATE Let $S[t] = \{i:b_i^t=1\}$ be the set of \textit{bidding agents}.
            \STATE A winner $i^t$ is randomly selected from $S[t]$ according to the probability distribution $(p_i^{S[t]})_{i\in S[t]}$.
            \STATE Budgets get updated: $B_i[t+1] = B_i[t] - b_i^t$ for every agent $i$.
        \ENDFOR
    \end{algorithmic}
\end{algorithm}

The strategy class we propose is a simple and natural strategy wherein each agent bids whenever their value is above a certain threshold. As in~\cite{fikioris2023online}, we refer to this as a $\beta$-aggressive strategy. \begin{definition}[$\beta$-aggressive strategy]
\label{def:beta_aggressive}
    Agent $i$ follows a \textit{$\beta$-aggressive strategy} if she bids at time $t$ if and only if her value $V_i[t]$ is the top $\beta$-quantile of her value distribution\footnote{If the top $\beta$-quantile is not well-defined, as when the value distribution has atoms, the agent can bid with probability $\rho(V_i[t])$ where $\rho:[0,\infty)\to[0,1]$ maximizes $\E_{V_i\sim\mathcal Fi}[V_i\rho(V_i)]$ subject to $\E_{V_i\sim\mathcal F_i}[\rho(V_i)]\leq \beta$.}.
\end{definition}
Intuitively, for an agent with fair share $\alpha$, following a $\alpha$-aggressive strategy is akin to truthfully reporting her $\alpha T$ favorite rounds. Building on this, we propose that each agent $i$ follows an $\alpha_i$-aggressive strategy. Note that with no competition, this would realize the agent's ideal utility. Our first result is that everyone playing an $\alpha_i$-aggressive strategy is an approximate Nash equilibrium for any allocation probabilities.
\begin{definition}[$\epsilon$-approximate Nash equilibrium]
    For a profile of strategies $\pi = (\pi_1, \dots, \pi_n)$, let $\mathcal U_i(\pi)$ denote the per-rounded expected utility $\frac1T\sum_{t=1}^T \E[U_i[t]]$ of agent $i$ when players play according to $\pi$. We say that a profile of strategies $\pi = (\pi_1, \dots, \pi_n)$ forms an \textit{$\epsilon$-approximate Nash equilibrium} if for each agent $i$, for every profile of strategies $\pi' = (\pi_1', \dots, \pi_n')$ where $\pi_j = \pi_j'$ for each $j\neq i$, we have $\mathcal U_i(\pi') \leq \mathcal U_i(\pi) + \epsilon$.
\end{definition}
\begin{proposition}
\label{prop:approximate_nash}
    By setting $\delta_i^T = \sqrt{\nicefrac{6\ln T}{\alpha_i T}}$, for any fixed allocation probabilities $p_i^S$, each agent playing an $\alpha_i$-aggressive strategy is an $O(\sqrt{\nicefrac{\log T}{T}})$-approximate Nash equilibrium.
\end{proposition}
We defer the proof of this to~\cref{sec:nash}.
Note that the utility guarantees at this equilibrium depend on the choice of allocation probabilities $p_i^S$.
Since the agents' bids are i.i.d. in every round, we can calculate the fixed probability $p_i$ that agent $i$ will win the item conditioned on bidding (assuming everyone has budget remaining).
This $p_i$ is called agent $i$'s \textit{interim allocation probability}.

\begin{definition}[Interim allocation probability]
\label{def:alpha_interim_allocation_probability}
Allocation probabilities $p_i^S$ induce \textit{interim allocation probabilities} $p_i$ if $p_i$ is the probability that agent $i$ wins the item in a given round conditioned on agent $i$ bidding, and agents $j\neq i$ bidding independently with probability $\alpha_j$ each. Formally:
\begin{equation*}
		p_i = \sum_{S\subseteq[n]:i\in S}p_i^S\left(\prod_{j\in S\setminus\{i\}}\alpha_j\right)\left(\prod_{j\in[n]\setminus S}(1-\alpha_j)\right).
	\end{equation*}
\end{definition}

Our goal is to maximize the fraction of ideal utility every agent is guaranteed in the above equilibrium.
By definition of $p_i$, every agent in the equilibrium is guaranteed a $p_i$ fraction of their ideal utility.
The following theorem shows that we can set $p_i^S$ to achieve
the upper bound of \cref{lem:centralallocationworstcase}; we explain the theorem in more detail in terms of Border's theorem in \cref{sec:border}.

\begin{theorem}
    \label{thm:worst_case_interim_allocation_probabilities}
    There exist allocation probabilities $p_i^S$ such that the induced interim allocation probabilities $p_i$ are the same for all agents $i$. Specifically, we can make
    $ 
        p_i = 1 - \prod_{j=1}^n (1-\alpha_j)
    $ 
    for all $i$, matching \cref{lem:centralallocationworstcase}.
\end{theorem}

Having achieved optimal equilibrium performance, we turn our attention to robustness.
Unlike the equilibrium analysis, where we can use any allocation probabilities $p_i^S$ that induce the desired interim allocation probabilities $p_i$, under arbitrary competition, we must be more careful with our selection of $p_i^S$.
In \cref{ssec:interim_allocation_probabilities_do_not_guarantee_robustness}, we show that not carefully picking the $p_i^S$'s can lead to very poor robustness.
Our main technical contribution is in \cref{ssec:half_robustness}, where we show that under careful upper bounds on the $p_i^S$'s we can guarantee strong robust performance without changing the $p_i$'s and thus the equilibrium performance.
The following theorem summarizes our overall guarantees, subsuming \cref{thm:worst_case_interim_allocation_probabilities}.

\begin{theorem}
\label{thm:nash}
\label{thm:robustness_guarantee}
Consider the \BRB Mechanism (\cref{alg:mechanism}) with $\delta_i^T = \sqrt{\nicefrac{6\ln T}{\alpha_i T}}$. With a careful choice of allocation probabilities $p_i^S$ (that satisfy~\cref{thm:worst_case_interim_allocation_probabilities}), we have the following:
\begin{enumerate}
\item Each agent $i$ playing an $\alpha_i$-aggressive strategy is an $O(\sqrt{\nicefrac{\log T}{T}})$-approximate Bayes-Nash equilibrium where, with probability $1-O(1/T^2)$, agent $i$ realizes utility
\begin{equation*}
 \frac1T\sum_{t=1}^T U_i[t] \geq \left(1-\prod_{j=1}^n (1-\alpha_j)\right)v_i^\star - O\left(\sqrt{\frac{\log T}{T}}\right).
\end{equation*}

\item Regardless of behavior of others, playing an $\alpha_i$-aggressive strategy, with probability $1-O(1/T^2)$, gives agent $i$ utility
    \begin{equation*}
        \frac1T\sum_{t=1}^T U_i[t] \geq \left(\frac12 + \frac12\alpha_i^2\right)v_i^\star - O\left(\sqrt{\frac{\log T}{T}}\right).
    \end{equation*}
\end{enumerate}
\end{theorem}

Note that this implies that with high probability, each agent gets a $\lnash =  1-\prod_{j=1}^n (1-\alpha_j) - O(\sqrt{\nicefrac{\log T}{T}})$ fraction of their ideal utility at the approximate equilibrium, and a $\lrob = (\frac12 - O(\sqrt{\nicefrac{\log T}{T}}))$ fraction robustly.
We present in more detail the equilibrium claim in \cref{sec:nash} and the robustness claim in \cref{sec:robustness}.

In \cref{sec:nash,sec:robustness} we only show the existence of allocation probabilities $p_i^S$ and not how to efficiently compute them.
However, using techniques from \cite{bhalgat2013optimal}, we give an algorithm that, given a bidding set $S$, can efficiently sample from a distribution $(p_i^S)_{i\in S}$, where $((p_i^S)_{i\in S})_{S\subseteq[n]}$ are allocation probabilities that will give us the equilibrium and robustness guarantees of \cref{thm:nash}. We detail this algorithm and how it can be used with \BRB in \cref{sec:computation}.

By our choice of $\delta_i^T$ in \cref{thm:robustness_guarantee}, the probability of any agent running out of budget is $O(1/T^2)$ (similar to the probability of our bound not holding).
Due to this, we could modify \BRB's budget constraint to allow agent $i$ to bid at most $\alpha_i (1 + \delta_i^t) t$ times by round $t \in [T]$ to get any time guarantees.
Specifically, our utility bounds would hold for every round $t \in [T]$ with probability $1 - O(1/\sqrt t)$.
Enforcing the budget constraint at each time also allows us to obtain an exact Nash equilibrium in the infinite time horizon.
See \cref{sec:anytime} for details.

\section{Equilibrium and Good Allocation Probabilities}
\label{sec:nash}
\label{sec:border}

In this section, we prove our claim of \cref{sec:mech_and_strategy} that each agent $i$ following the $\alpha_i$-aggressive strategy in the \BRB Mechanism forms an approximate Nash Equilibrium.
For simplicity, we consider a simpler game where budgets are enforced in expectation.
While bounds in expectation are not enforceable, we focus on this case for simplicity of presentation, and only require that $\sum_{t=1}^T \mathbb E[b_i^t] \leq \alpha_i T$.
Roughly speaking, we prove that in this game, the best response of agent $i$ when agents $j \ne i$ are following an $\alpha_j$-aggressive strategy is to follow an $\alpha_i$-aggressive strategy also.
This proves that these strategies form an equilibrium.

\begin{lemma}
    \label{lem:knapsackproblemovertime}
    Fix an agent $i$ and assume all other agents $j\neq i$ are bidding i.i.d. $\mathrm{Bernoulli}(\alpha_j)$.
    Suppose agent $i$ is trying to maximize her expected utility subject to the constraint that she does not spend more than $\alpha_iT$ tokens in expectation; that is, she is choosing $(b_i^t)_t$ to solve
    \begin{equation*}
        \max\,\frac1T\sum_{t=1}^T\mathbb E[U_i[t]] \quad\text{subject to}\quad \sum_{t=1}^T \mathbb E[b_i^t]\leq \alpha_iT
    \end{equation*}
    Her optimal strategy is to choose $(b_i^t)_t$ that corresponds to bidding whenever her value is in the top $\alpha_i$-quantile of her value distribution, which yields agent $i$ expected utility
    $
        \frac1T\sum_{t=1}^T \mathbb E[U_i[t]] = p_i v_i^\star.
    $
\end{lemma}

To see why $\alpha_i$-aggressiveness is the best response, note that in each round, the probability that agent $i$ wins conditioned on bidding is exactly their interim allocation probability $p_i$.
The best solution for the agent is to bid in the $\alpha_iT$ rounds where the agent has the highest value.
Therefore, agent $i$ should follow an $\alpha_i$-aggressive strategy, bidding whenever her value is in the top $\alpha_i$-quantile of her value distribution. 
The full proof is in \cref{ssec:appendix_equilibrium_proofs}, where we also translate the result from this simplified game with expected budget constraints to the actual game where budgets are strictly enforced (but larger).

\subsection{Inducing Optimal Interim Allocation Probabilities}
\label{ssec:border}

Next we show how to set the allocation probabilities $p_i^S$
to achieve optimal equilibrium performance in our \BRB Mechanism.
Specifically, in \cref{thm:worst_case_interim_allocation_probabilities} we claimed that we can set each $p_i^S$ so that all interim allocation probabilities are equal: $p_i = 1 - \prod_j (1 - \alpha_j)$ for all $i$.
This result follows from a special case of Border’s theorem on the feasibility of reduced-form auctions; we present this special case here and defer the full statement of the theorem to \cref{sec:border_details}.

\begin{theorem}
\label{thm:alpha_border}
    Given arbitrary numbers $p_i\in[0,1]$, there exist allocation probabilities $p_i^S$ such that the $p_i$ are interim allocation probabilities induced by the $p_i^S$'s if and only if for every $I\subseteq[n]$,
	\begin{equation}
		\label{eq:introductory_border_new}
		\sum_{i\in I}p_i\alpha_i \leq 1 - \prod_{i\in I}(1-\alpha_i)
	\end{equation}
        with equality for $I = [n]$.
\end{theorem}

For completeness, we include a flow-based proof of Border’s theorem from \citet{che2013generalized} in \cref{sec:border_details}. Our special case (\cref{thm:alpha_border}) reduces the feasibility of interim allocations to a flow problem, illustrated in \cref{fig:introductory_border_flow_network}, where flows represent the probabilities of agent groups bidding and being allocated (see the figure caption for details). A feasible flow of value $\Pr(S' \neq \emptyset) = 1 - \prod_{i \in [n]} (1 - \alpha_i)$ exists if and only if there are allocation probabilities $p_i^S$ that induce the given $p_i$. The necessity of condition~(\ref{eq:introductory_border_new}) follows from interpreting the left-hand side as the total interim allocation to agents in set $I$, and the right-hand side as the probability that some agent in $I$ bids. Sufficiency follows by showing that these inequalities ensure every $(s,t)$ cut with finite capacity has at least the required flow value.


\begin{figure}
\centering
    \resizebox{.55\textwidth}{!}{\tikzset{every picture/.style={line width=0.75pt}} 

\begin{tikzpicture}[x=0.75pt,y=0.75pt,yscale=-1,xscale=1]

\draw   (99,105.5) .. controls (99,97.49) and (105.49,91) .. (113.5,91) .. controls (121.51,91) and (128,97.49) .. (128,105.5) .. controls (128,113.51) and (121.51,120) .. (113.5,120) .. controls (105.49,120) and (99,113.51) .. (99,105.5) -- cycle ;
\draw   (255,23) .. controls (255,14.16) and (262.16,7) .. (271,7) .. controls (279.84,7) and (287,14.16) .. (287,23) .. controls (287,31.84) and (279.84,39) .. (271,39) .. controls (262.16,39) and (255,31.84) .. (255,23) -- cycle ;
\draw   (252,83.5) .. controls (252,73.84) and (259.84,66) .. (269.5,66) .. controls (279.16,66) and (287,73.84) .. (287,83.5) .. controls (287,93.16) and (279.16,101) .. (269.5,101) .. controls (259.84,101) and (252,93.16) .. (252,83.5) -- cycle ;
\draw   (253,158) .. controls (253,148.61) and (260.61,141) .. (270,141) .. controls (279.39,141) and (287,148.61) .. (287,158) .. controls (287,167.39) and (279.39,175) .. (270,175) .. controls (260.61,175) and (253,167.39) .. (253,158) -- cycle ;
\draw   (412,21.5) .. controls (412,13.49) and (418.49,7) .. (426.5,7) .. controls (434.51,7) and (441,13.49) .. (441,21.5) .. controls (441,29.51) and (434.51,36) .. (426.5,36) .. controls (418.49,36) and (412,29.51) .. (412,21.5) -- cycle ;
\draw   (413,77.5) .. controls (413,69.49) and (419.49,63) .. (427.5,63) .. controls (435.51,63) and (442,69.49) .. (442,77.5) .. controls (442,85.51) and (435.51,92) .. (427.5,92) .. controls (419.49,92) and (413,85.51) .. (413,77.5) -- cycle ;
\draw   (412,158.5) .. controls (412,150.49) and (418.49,144) .. (426.5,144) .. controls (434.51,144) and (441,150.49) .. (441,158.5) .. controls (441,166.51) and (434.51,173) .. (426.5,173) .. controls (418.49,173) and (412,166.51) .. (412,158.5) -- cycle ;
\draw   (528,100.5) .. controls (528,92.49) and (534.49,86) .. (542.5,86) .. controls (550.51,86) and (557,92.49) .. (557,100.5) .. controls (557,108.51) and (550.51,115) .. (542.5,115) .. controls (534.49,115) and (528,108.51) .. (528,100.5) -- cycle ;
\draw    (113.5,91) -- (253.2,23.87) ;
\draw [shift={(255,23)}, rotate = 154.33] [color={rgb, 255:red, 0; green, 0; blue, 0 }  ][line width=0.75]    (10.93,-3.29) .. controls (6.95,-1.4) and (3.31,-0.3) .. (0,0) .. controls (3.31,0.3) and (6.95,1.4) .. (10.93,3.29)   ;
\draw    (124,95) -- (250.01,82.69) ;
\draw [shift={(252,82.5)}, rotate = 174.42] [color={rgb, 255:red, 0; green, 0; blue, 0 }  ][line width=0.75]    (10.93,-3.29) .. controls (6.95,-1.4) and (3.31,-0.3) .. (0,0) .. controls (3.31,0.3) and (6.95,1.4) .. (10.93,3.29)   ;
\draw    (125,114.5) -- (251.11,157.36) ;
\draw [shift={(253,158)}, rotate = 198.77] [color={rgb, 255:red, 0; green, 0; blue, 0 }  ][line width=0.75]    (10.93,-3.29) .. controls (6.95,-1.4) and (3.31,-0.3) .. (0,0) .. controls (3.31,0.3) and (6.95,1.4) .. (10.93,3.29)   ;
\draw    (287,157.5) -- (410,158.48) ;
\draw [shift={(412,158.5)}, rotate = 180.46] [color={rgb, 255:red, 0; green, 0; blue, 0 }  ][line width=0.75]    (10.93,-3.29) .. controls (6.95,-1.4) and (3.31,-0.3) .. (0,0) .. controls (3.31,0.3) and (6.95,1.4) .. (10.93,3.29)   ;
\draw    (287,83) -- (411,77.59) ;
\draw [shift={(413,77.5)}, rotate = 177.5] [color={rgb, 255:red, 0; green, 0; blue, 0 }  ][line width=0.75]    (10.93,-3.29) .. controls (6.95,-1.4) and (3.31,-0.3) .. (0,0) .. controls (3.31,0.3) and (6.95,1.4) .. (10.93,3.29)   ;
\draw    (289,23.5) -- (410,21.53) ;
\draw [shift={(412,21.5)}, rotate = 179.07] [color={rgb, 255:red, 0; green, 0; blue, 0 }  ][line width=0.75]    (10.93,-3.29) .. controls (6.95,-1.4) and (3.31,-0.3) .. (0,0) .. controls (3.31,0.3) and (6.95,1.4) .. (10.93,3.29)   ;
\draw    (284,72) -- (411.1,30.62) ;
\draw [shift={(413,30)}, rotate = 161.97] [color={rgb, 255:red, 0; green, 0; blue, 0 }  ][line width=0.75]    (10.93,-3.29) .. controls (6.95,-1.4) and (3.31,-0.3) .. (0,0) .. controls (3.31,0.3) and (6.95,1.4) .. (10.93,3.29)   ;
\draw    (441,21.5) -- (534.36,86.85) ;
\draw [shift={(536,88)}, rotate = 214.99] [color={rgb, 255:red, 0; green, 0; blue, 0 }  ][line width=0.75]    (10.93,-3.29) .. controls (6.95,-1.4) and (3.31,-0.3) .. (0,0) .. controls (3.31,0.3) and (6.95,1.4) .. (10.93,3.29)   ;
\draw    (441,158.5) -- (534.2,113.86) ;
\draw [shift={(536,113)}, rotate = 154.41] [color={rgb, 255:red, 0; green, 0; blue, 0 }  ][line width=0.75]    (10.93,-3.29) .. controls (6.95,-1.4) and (3.31,-0.3) .. (0,0) .. controls (3.31,0.3) and (6.95,1.4) .. (10.93,3.29)   ;
\draw    (442,77.5) -- (526.07,99.98) ;
\draw [shift={(528,100.5)}, rotate = 194.97] [color={rgb, 255:red, 0; green, 0; blue, 0 }  ][line width=0.75]    (10.93,-3.29) .. controls (6.95,-1.4) and (3.31,-0.3) .. (0,0) .. controls (3.31,0.3) and (6.95,1.4) .. (10.93,3.29)   ;

\draw (109,101) node [anchor=north west][inner sep=0.75pt]   [align=left] {$\displaystyle s$};
\draw (259,104) node [anchor=north west][inner sep=0.75pt]   [align=left] {$\displaystyle \vdots $};
\draw (421,102) node [anchor=north west][inner sep=0.75pt]   [align=left] {$\displaystyle \vdots $};
\draw (538,97.4) node [anchor=north west][inner sep=0.75pt]    {$t$};
\draw (334,119.4) node [anchor=north west][inner sep=0.75pt]    {$\cdots $};
\draw (260,18.4) node [anchor=north west][inner sep=0.75pt]    {$u_{\{1\}}$};
\draw (252,75.4) node [anchor=north west][inner sep=0.75pt]    {$u_{\{1,2\}}$};
\draw (257,151.4) node [anchor=north west][inner sep=0.75pt]    {$u_{\{n\}}$};
\draw (420,18.4) node [anchor=north west][inner sep=0.75pt]    {$v_{1}$};
\draw (420,73.4) node [anchor=north west][inner sep=0.75pt]    {$v_{2}$};
\draw (418,153.4) node [anchor=north west][inner sep=0.75pt]    {$v_{n}$};
\draw (134.88,54.56) node [anchor=north west][inner sep=0.75pt]  [rotate=-334.91]  {$\Pr( S'=\{1\})$};
\draw (151.46,73.83) node [anchor=north west][inner sep=0.75pt]  [rotate=-352.96]  {$\Pr( S'=\{1,2\})$};
\draw (150.27,102.45) node [anchor=north west][inner sep=0.75pt]  [rotate=-18.18]  {$\Pr( S'=\{n\})$};
\draw (341,8.4) node [anchor=north west][inner sep=0.75pt]    {$\infty $};
\draw (324.91,44.76) node [anchor=north west][inner sep=0.75pt]  [rotate=-340.02]  {$\infty $};
\draw (336.56,65.92) node [anchor=north west][inner sep=0.75pt]  [rotate=-356.57]  {$\infty $};
\draw (336.25,142.13) node [anchor=north west][inner sep=0.75pt]  [rotate=-1.83]  {$\infty $};
\draw (459.8,9.03) node [anchor=north west][inner sep=0.75pt]  [rotate=-35.31]  {$p_{1}\Pr( 1\in S')$};
\draw (446.93,57.89) node [anchor=north west][inner sep=0.75pt]  [rotate=-17.56]  {$p_{2}\Pr( 2\in S')$};
\draw (437.83,137.84) node [anchor=north west][inner sep=0.75pt]  [rotate=-333.31]  {$p_{n}\Pr( n\in S')$};

\end{tikzpicture}}
	\caption{\small\emph{Flow network that can be used to prove \cref{thm:alpha_border}.
            We let $S'$ be the random set of bidding agents where each agent $i$ lies in $S'$ independently with probability $\alpha_i$.
            Then, there are three kind of edges: edges whose flow corresponds to the probability of observing a specific $S'$ (left), edges whose flow corresponds to how we randomly allocate the item condition on observing a specific $S'$ (middle), and edges whose flow represent the probability that a specific agent gets the item (right).
    		There is a flow of value $\Pr(S'\neq\emptyset)$ if and only if there exist allocation probabilities $p_i^S$ inducing the interim allocation probabilities $p_i$.
            In other words, the flows $p_i^S\Pr(S' = S)$ in the middle transform the probabilities that agents in a certain set $S$ bid to an agent $i\in S$ being allocated.
            We obtain the conditions in \cref{thm:alpha_border} by analyzing every minimum-cut of this network.
	}}
	\label{fig:introductory_border_flow_network}
\end{figure}
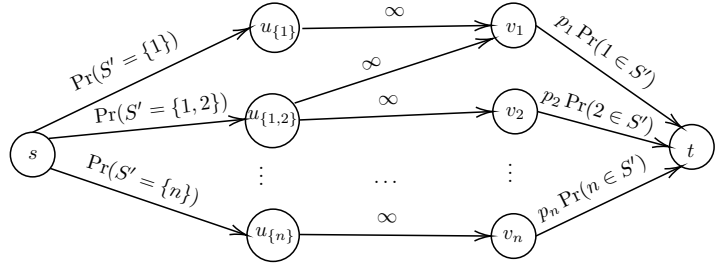

In \cref{ssec:half_robustness}, we prove a generalization of \cref{thm:alpha_border} and use it to prove a stronger result than \cref{thm:worst_case_interim_allocation_probabilities}. That proof is quite involved, so for simplicity, we prove \cref{thm:worst_case_interim_allocation_probabilities} directly from \cref{thm:alpha_border} in \cref{ssec:interim_allocation_probability_feasibility_proofs}.

 {}

\section{Robustness}
\label{sec:robustness}
In \cref{sec:nash}, we prove that each player playing an $\alpha_i$-aggressive strategy is an approximate equilibrium. To complete the proof of \cref{thm:robustness_guarantee}, we show how to choose the allocation probabilities $p_i^S$ to guarantee robustness while maintaining the same interim allocations.
To do this, in \cref{ssec:half_robustness}, we strengthen Border's Theorem to also satisfy some additional properties that guarantee $1/2$-robustness.
In \cref{ssec:interim_allocation_probabilities_do_not_guarantee_robustness}, we then show that this strengthening is necessary, as many allocations obtained by the standard Border's Theorem are not robust.
Finally, in \cref{ssec:half_robustness_hardness}, we show our robustness result is tight in that no matter the choice of $(p_i^S)$, in the \BRB mechanism, it is not possible to guarantee each agent a $\lambda$-robust strategy for a constant $\lambda$ greater than $1/2$ not depending on the fair shares $(\alpha_i)$ or the time horizon $T$.

\subsection{Achieving a \texorpdfstring{$1/2$}{1/2} Robustness Factor}
\label{ssec:half_robustness}
    
To prove \cref{thm:robustness_guarantee}, we show how to choose probabilities $p_i^S$ to make the $\alpha_i$-aggressive strategy $(\frac12 + \frac12\alpha_i^2)$-robust while still inducing interim allocation probabilities $p_i = 1 - \prod_{k=1}^n (1-\alpha_k)$ (as in \cref{thm:worst_case_interim_allocation_probabilities}).

We consider (collusive) strategies that agents $j \ne i$ can employ to minimize agent $i$'s utility.
Say that agent $i$ is \textit{blocked} when she bids but does not receive the item.
We will focus on how many bids the other agents have to make each time they block agent $i$.
Let us first consider the case where other agents never bid two at a time.
Conditioned on only agent $j$ bidding, agent $i$ gets blocked with probability $\alpha_i p_j^{\{i, j\}}$, with $\alpha_i$ being the probability that agent $i$ bids and $p_j^{\{i,j\}}$ being the probability that agent $j$ wins when they both bid. We have to ensure that this probability is not too large.
Specifically, assume that for some $\bar p$ it holds $p_j^{\{i, j\}} \le \bar p$ for all $j \ne i$.
Thus, when one other agent bids at a time, the other agents must spend $(\bar p \alpha_i)^{-1}$ tokens each time they block agent $i$ in expectation.
When 2 or more agents bid, we still have that agent $i$ bids only with probability $\alpha_i$, so the expected number of tokens spent for blocking agent $i$ is at least $2\alpha_i^{-1}$. As long as $\bar p\ge 1/2$, this is less effective from a bang-per-buck perspective.

Since the number of times agents $j \ne i$ can bid is at most $\sum_{j \ne i}\alpha_j = (1 - \alpha_i) T$, the expected number of times agent $i$ can get blocked is at most $\max(\bar p, 1/2) \alpha_i (1 - \alpha_i) T$. 
Combining this with the fact that agent $i$ bids $\alpha_i T$ times in expectation, we get that the expected number of rounds in which she wins the item is at least
\begin{equation*}
    \alpha_i T - \max(\bar p, 1/2) \alpha_i (1 - \alpha_i) T
    =
    \qty\big(1 - \max(\bar p, 1/2) (1 - \alpha_i) ) \alpha_i T
\end{equation*}
To ensure that agent $i$ does not get too low utility, we thus need that $\bar p$ is not too large.

The formal statement of the above argument is the lemma below. Its proof, which we defer to \cref{ssec:appendix_bang_for_buck_proof}, uses martingale concentration arguments to ensure that agent $i$ obtains at least $(1 - (1-\alpha_i)\bar p)$ fraction of her ideal utility with high probability.

\begin{lemma}
    \label{lem:bangforbuck}
    Fix an agent $i$.
    Given allocation probabilities $(p_k^S)$, if $p_j^{\{i,j\}} \leq \bar p$ for every other agent $j$ where $\bar p \geq 1/2$, then when we run \BRB with slack parameters $\delta_i^T = \sqrt{\nicefrac{6\ln T}{\alpha_i T}}$, an $\alpha_i$-aggressive strategy guarantees agent $i$ utility
    \begin{equation*}
        \frac1T\sum_{t=1}^T U_i[t] \geq \left(1 - \bar p(1-\alpha_i)\right)v_i^\star - O\left(\sqrt{\frac{\log T}{T}}\right)
    \end{equation*}
    with probability at least $1- O(1/T^2)$ regardless of the behavior of agents $j\neq i$.
\end{lemma}

We want to give allocation probabilities $p_i^S$ that induce interim probabilities $p_i = 1 - \prod_k (1 - \alpha_k)$, while keeping the upper bounds $\bar{p}$ as in \cref{lem:bangforbuck} small. This enables us to achieve both equilibrium and robustness guarantees, as stated in \cref{thm:nash}. 

When all agents have equal fair share $\alpha_i = 1/n$, assigning the item uniformly at random among bidders (i.e., $p_i^S = 1/|S|$) yields the desired interim allocation $p_i = 1 - (1 - 1/n)^n$ by symmetry, and satisfies $p_j^{\{i,j\}} = 1/2$. Applying \cref{lem:bangforbuck} with $\bar{p} = 1/2$ then gives a robustness guarantee of $1/2 + 1/(2n)$.

It would be natural to try to generalize the use of $\bar p = 1/2$ to the asymmetric case. However, this would be too low, and the desired interim allocation of \cref{thm:nash} may not be feasible with this restriction.
In fact, in \cref{sec:twoagents}, we show that if there are only $2$ agents, to induce the desired interim allocation probabilities, the unique solution is $p_1^{\{1,2\}} = \alpha_2$, which can be more than $\bar p = 1/2$.
We end up using $\bar p = \frac{1+\alpha_i}{2}$ (satisfying $\alpha_i \le \bar p$), which proves the robustness guarantee of \cref{thm:robustness_guarantee}.
The remaining technically challenging part of the proof is to show that there is a set of allocation probabilities $p_i^S$ satisfying these conditions.

To prove that the desired allocation probabilities exist, we consider the network of \cref{fig:introductory_border_flow_network} used to prove Border's theorem. Border's theorem does not cover upper bound constraints on $p_j^{\{i, j\}}$ of \cref{lem:bangforbuck}. To ensure this limit, we change the capacity of the edges from node $u_{\{i,j\}}$ to node $v_j$ for all agents $j$ and all two-element sets $\{i,j\}$ to $\frac{1+\alpha_i}{2}$ as illustrated in \cref{fig:modifiedborderscriterionflownetwork}. 
In \cref{lem:specific_border_modification}, we establish conditions under which a flow of value $1 - \prod_k (1 - \alpha_k)$ exists in this modified network. While the condition is more complex, the proof mirrors the original Border argument: we verify that every cut has sufficient capacity to support the desired flow. The most technically challenging and surprising part is establishing that this condition actually holds, which we do in \cref{lem:key_lemma}, thereby completing the proof of \cref{thm:robustness_guarantee}.

Finally, while our condition uses specific bounds $p_j^{\{i,j\}} \leq (1 + \alpha_i)/2$, we also prove a more general version allowing arbitrary upper bounds on any $p_i^S$. This generalization is presented in \cref{sec:border_details}.

\begin{figure}[t]
\centering
    \resizebox{.55\textwidth}{!}{\tikzset{every picture/.style={line width=0.75pt}} 

\begin{tikzpicture}[x=0.75pt,y=0.75pt,yscale=-1,xscale=1]

\draw   (90,113.5) .. controls (90,105.49) and (96.49,99) .. (104.5,99) .. controls (112.51,99) and (119,105.49) .. (119,113.5) .. controls (119,121.51) and (112.51,128) .. (104.5,128) .. controls (96.49,128) and (90,121.51) .. (90,113.5) -- cycle ;
\draw   (234,25.5) .. controls (234,16.39) and (241.39,9) .. (250.5,9) .. controls (259.61,9) and (267,16.39) .. (267,25.5) .. controls (267,34.61) and (259.61,42) .. (250.5,42) .. controls (241.39,42) and (234,34.61) .. (234,25.5) -- cycle ;
\draw   (235,94.5) .. controls (235,84.84) and (242.84,77) .. (252.5,77) .. controls (262.16,77) and (270,84.84) .. (270,94.5) .. controls (270,104.16) and (262.16,112) .. (252.5,112) .. controls (242.84,112) and (235,104.16) .. (235,94.5) -- cycle ;
\draw   (236,168.5) .. controls (236,159.39) and (243.39,152) .. (252.5,152) .. controls (261.61,152) and (269,159.39) .. (269,168.5) .. controls (269,177.61) and (261.61,185) .. (252.5,185) .. controls (243.39,185) and (236,177.61) .. (236,168.5) -- cycle ;
\draw   (405,29.5) .. controls (405,21.49) and (411.49,15) .. (419.5,15) .. controls (427.51,15) and (434,21.49) .. (434,29.5) .. controls (434,37.51) and (427.51,44) .. (419.5,44) .. controls (411.49,44) and (405,37.51) .. (405,29.5) -- cycle ;
\draw   (406,85.5) .. controls (406,77.49) and (412.49,71) .. (420.5,71) .. controls (428.51,71) and (435,77.49) .. (435,85.5) .. controls (435,93.51) and (428.51,100) .. (420.5,100) .. controls (412.49,100) and (406,93.51) .. (406,85.5) -- cycle ;
\draw   (405,166.5) .. controls (405,158.49) and (411.49,152) .. (419.5,152) .. controls (427.51,152) and (434,158.49) .. (434,166.5) .. controls (434,174.51) and (427.51,181) .. (419.5,181) .. controls (411.49,181) and (405,174.51) .. (405,166.5) -- cycle ;
\draw   (523,111.5) .. controls (523,103.49) and (529.49,97) .. (537.5,97) .. controls (545.51,97) and (552,103.49) .. (552,111.5) .. controls (552,119.51) and (545.51,126) .. (537.5,126) .. controls (529.49,126) and (523,119.51) .. (523,111.5) -- cycle ;
\draw    (104.5,99) -- (236.2,35.37) ;
\draw [shift={(238,34.5)}, rotate = 154.21] [color={rgb, 255:red, 0; green, 0; blue, 0 }  ][line width=0.75]    (10.93,-3.29) .. controls (6.95,-1.4) and (3.31,-0.3) .. (0,0) .. controls (3.31,0.3) and (6.95,1.4) .. (10.93,3.29)   ;
\draw    (115,103) -- (233,94.64) ;
\draw [shift={(235,94.5)}, rotate = 175.95] [color={rgb, 255:red, 0; green, 0; blue, 0 }  ][line width=0.75]    (10.93,-3.29) .. controls (6.95,-1.4) and (3.31,-0.3) .. (0,0) .. controls (3.31,0.3) and (6.95,1.4) .. (10.93,3.29)   ;
\draw    (116,122.5) -- (234.13,167.78) ;
\draw [shift={(236,168.5)}, rotate = 200.97] [color={rgb, 255:red, 0; green, 0; blue, 0 }  ][line width=0.75]    (10.93,-3.29) .. controls (6.95,-1.4) and (3.31,-0.3) .. (0,0) .. controls (3.31,0.3) and (6.95,1.4) .. (10.93,3.29)   ;
\draw    (269,168.5) -- (403,166.53) ;
\draw [shift={(405,166.5)}, rotate = 179.16] [color={rgb, 255:red, 0; green, 0; blue, 0 }  ][line width=0.75]    (10.93,-3.29) .. controls (6.95,-1.4) and (3.31,-0.3) .. (0,0) .. controls (3.31,0.3) and (6.95,1.4) .. (10.93,3.29)   ;
\draw [color={rgb, 255:red, 208; green, 2; blue, 27 }  ,draw opacity=1 ]   (270,94.5) -- (404,85.63) ;
\draw [shift={(406,85.5)}, rotate = 176.21] [color={rgb, 255:red, 208; green, 2; blue, 27 }  ,draw opacity=1 ][line width=0.75]    (10.93,-3.29) .. controls (6.95,-1.4) and (3.31,-0.3) .. (0,0) .. controls (3.31,0.3) and (6.95,1.4) .. (10.93,3.29)   ;
\draw    (267,25.5) -- (406,21.06) ;
\draw [shift={(408,21)}, rotate = 178.17] [color={rgb, 255:red, 0; green, 0; blue, 0 }  ][line width=0.75]    (10.93,-3.29) .. controls (6.95,-1.4) and (3.31,-0.3) .. (0,0) .. controls (3.31,0.3) and (6.95,1.4) .. (10.93,3.29)   ;
\draw [color={rgb, 255:red, 208; green, 2; blue, 27 }  ,draw opacity=1 ]   (261,80) -- (405.09,35.59) ;
\draw [shift={(407,35)}, rotate = 162.87] [color={rgb, 255:red, 208; green, 2; blue, 27 }  ,draw opacity=1 ][line width=0.75]    (10.93,-3.29) .. controls (6.95,-1.4) and (3.31,-0.3) .. (0,0) .. controls (3.31,0.3) and (6.95,1.4) .. (10.93,3.29)   ;
\draw    (434,29.5) -- (529.37,97.84) ;
\draw [shift={(531,99)}, rotate = 215.62] [color={rgb, 255:red, 0; green, 0; blue, 0 }  ][line width=0.75]    (10.93,-3.29) .. controls (6.95,-1.4) and (3.31,-0.3) .. (0,0) .. controls (3.31,0.3) and (6.95,1.4) .. (10.93,3.29)   ;
\draw    (434,166.5) -- (529.17,124.8) ;
\draw [shift={(531,124)}, rotate = 156.34] [color={rgb, 255:red, 0; green, 0; blue, 0 }  ][line width=0.75]    (10.93,-3.29) .. controls (6.95,-1.4) and (3.31,-0.3) .. (0,0) .. controls (3.31,0.3) and (6.95,1.4) .. (10.93,3.29)   ;
\draw    (435,85.5) -- (521.08,110.93) ;
\draw [shift={(523,111.5)}, rotate = 196.46] [color={rgb, 255:red, 0; green, 0; blue, 0 }  ][line width=0.75]    (10.93,-3.29) .. controls (6.95,-1.4) and (3.31,-0.3) .. (0,0) .. controls (3.31,0.3) and (6.95,1.4) .. (10.93,3.29)   ;

\draw (100,109) node [anchor=north west][inner sep=0.75pt]   [align=left] {$\displaystyle s$};
\draw (242,115) node [anchor=north west][inner sep=0.75pt]   [align=left] {$\displaystyle \vdots $};
\draw (414,110) node [anchor=north west][inner sep=0.75pt]   [align=left] {$\displaystyle \vdots $};
\draw (533,108.4) node [anchor=north west][inner sep=0.75pt]    {$t$};
\draw (329,126.4) node [anchor=north west][inner sep=0.75pt]    {$\cdots $};
\draw (241,21.4) node [anchor=north west][inner sep=0.75pt]    {$u_{\{1\}}$};
\draw (235,89.4) node [anchor=north west][inner sep=0.75pt]    {$u_{\{1,2\}}$};
\draw (240,162.4) node [anchor=north west][inner sep=0.75pt]    {$u_{\{n\}}$};
\draw (413,26.4) node [anchor=north west][inner sep=0.75pt]    {$v_{1}$};
\draw (413,81.4) node [anchor=north west][inner sep=0.75pt]    {$v_{2}$};
\draw (411,161.4) node [anchor=north west][inner sep=0.75pt]    {$v_{n}$};
\draw (129.35,63.17) node [anchor=north west][inner sep=0.75pt]  [rotate=-332.74]  {$\Pr( S'=\{1\})$};
\draw (134.52,83.23) node [anchor=north west][inner sep=0.75pt]  [rotate=-354.71]  {$\Pr( S'=\{1,2\})$};
\draw (141.41,109.45) node [anchor=north west][inner sep=0.75pt]  [rotate=-21.13]  {$\Pr( S'=\{n\})$};
\draw (330,7.4) node [anchor=north west][inner sep=0.75pt]    {$\infty $};
\draw (256.46,59.19) node [anchor=north west][inner sep=0.75pt]  [color={rgb, 255:red, 208; green, 2; blue, 27 }  ,opacity=1 ,rotate=-342.09]  {$\frac{1+\alpha _{2}}{2}\Pr( S'=\{1,2\})$};
\draw (286.2,72.86) node [anchor=north west][inner sep=0.75pt]  [color={rgb, 255:red, 208; green, 2; blue, 27 }  ,opacity=1 ,rotate=-355.16]  {$\frac{1+\alpha _{1}}{2}\Pr( S'=\{1,2\})$};
\draw (325.98,151.42) node [anchor=north west][inner sep=0.75pt]  [rotate=-359.85]  {$\infty $};
\draw (452.85,17.99) node [anchor=north west][inner sep=0.75pt]  [rotate=-35.39]  {$p_{1}\Pr( 1\in S')$};
\draw (441.5,66.71) node [anchor=north west][inner sep=0.75pt]  [rotate=-16.43]  {$p_{2}\Pr( 2\in S')$};
\draw (433.98,145.46) node [anchor=north west][inner sep=0.75pt]  [rotate=-337.49]  {$p_{n}\Pr( n\in S')$};

\end{tikzpicture}}
    \caption{\small\emph{Flow network for the proof of \cref{lem:specific_border_modification}, similar to the one in \cref{fig:introductory_border_flow_network} used in \cref{thm:alpha_border}.
    However, the red edges $(u_{\{i,j\}}, v_j)$ have explicit capacities as opposed to infinite capacity to enforce the additional bounds $\bar p_j^{\{i,j\}}\leq \frac{1+\alpha_i}{2}$.
        }
	}
	\label{fig:modifiedborderscriterionflownetwork}
\end{figure}
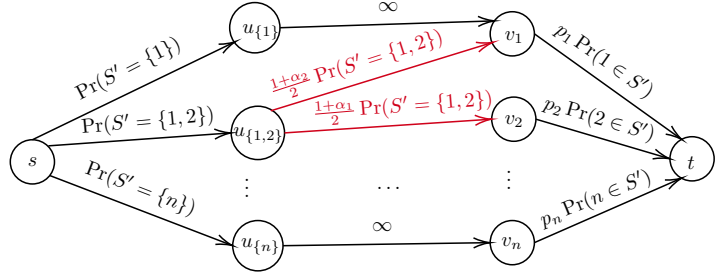

\begin{lemma}
\label{lem:specific_border_modification}
    There exists $p_i^S$'s that induce interim allocation probabilities $p_i = 1 - \prod_{j=1}^n (1-\alpha_j)$ and satisfy the upper bounds $p_j^{\{i,j\}}\leq \frac{1+\alpha_i}{2}$ if and only if for every $I\subseteq[n]$,
    \begin{equation}
    \label{eq:specific_generalized_border}
		\left(1 - \prod_{k=1}^n (1-\alpha_k)\right)\sum_{i\in I}\alpha_i + \prod_{i\in I}(1-\alpha_i) + \frac12\prod_{k=1}^n (1-\alpha_k)\sum_{i\in I}\frac{\alpha_i}{1-\alpha_i}\sum_{j\notin I}\alpha_j\leq 1.
    \end{equation}
\end{lemma}

\begin{proof}

Create an $s$-$t$ flow network as follows.
Let $\mu$ be the probability distribution over subsets $S'\subseteq[n]$ where each $i\in S'$ independently with probability $\alpha_i$.
For each nonempty $S\subseteq[n]$ create a node $u_S$ and connect it with an edge $(s, u_S)$ to the source node with capacity
\begin{equation*}
c(s, u_S) = \Pr_{S'\sim\mu}(S' = S).
\end{equation*}
For each $i$, create a node $v_i$.
For each $S\neq\emptyset$ such that $i\in S$, add an edge $(u_S, v_i)$ with capacity
\begin{equation*}
   c(u_S, v_i) = \begin{cases}\infty & \text{if $|S|\neq 2$}\\\frac{1+\alpha_j}{2}\Pr_{S'\sim\mu}(S'=S) & \text{if $S = \{i,j\}$}\end{cases}.
\end{equation*}
Also, add an edge $(v_i, t)$ with capacity
\begin{equation*}
c(v_i, t) = p_i\Pr_{S'\sim\mu}(i\in S').
\end{equation*}
The flow network is depicted in \cref{fig:modifiedborderscriterionflownetwork}.

Note that with $p_i = 1 - \prod_{j=1}^n(1-\alpha_j)$,
\begin{equation}
\label{eq:specific_no_waste}
    \sum_{i\in[n]}p_i\alpha_i  = 1 - \prod_{i\in[n]}(1-\alpha_i).
\end{equation}
The cut with $s$ on one side and everything else on the other has capacity
\begin{equation*}
\sum_{S\subseteq[n]:S\neq\emptyset}c(s,u_S) = \sum_{S\subseteq[n]:S\neq\emptyset}\Pr_{S'\sim\mu}(S'=S) = \Pr_{S'\sim\mu}(S'\neq\emptyset) = 1-\prod_{i\in [n]}(1-\alpha_i).
\end{equation*}
The cut with $t$ on one side and everything else on the other has capacity
\begin{equation*}
\sum_{i\in[n]}c(v_i, t) = \sum_{i\in[n]}p_i\Pr_{S'\sim\mu}(i\in S') = \sum_{i\in[n]}p_i\alpha_i = 1 - \prod_{i\in[n]}(1-\alpha_i),
\end{equation*}
using \eqref{eq:specific_no_waste} for the last equality.
Observe that in this flow network, allocation probabilities $(p_i^S)$ satisfying upper bounds $p_i^{\{i,j\}}\leq \frac{1+\alpha_i}{2}$ induce the interim allocation probabilities $(p_i)$ if and only if the flow $f$ is feasible where $f(s, u_S) = c(s, u_S)$, $f(v_i,t) = c(v_i,t)$, and $f(u_S, v_i) = p_i^S\Pr_{S'\sim\mu}(S' = S)$.
Since both the $s$-$t$ cuts with $s$ on one side and everything else on the other and the cut with $t$ on one side and everything else on the other both have cut capacity $1 - \prod_{i\in[n]}(1-\alpha_i)$, it suffices to show that \eqref{eq:specific_generalized_border} holds only if there is a feasible flow of flow value equal to this cut capacity.

Take any minimum-capacity $s$-$t$ cut $(A,B)$.
Let $I = \{i\in [n]:v_i\in B\}$.
We now argue that we can determine which side of the minimum-capacity cut $(A,B)$ the rest of the nodes are on just based on $I$.
\begin{itemize}
    \item If $i\in I$, then we must have $u_S\in B$ for $|S|\neq 2$ since those edges $(u_S, v_i)$ have infinite capacity.
    \item For any set $S$, if $i\notin I$ for every $i\in S$, then we can assume $u_S\in A$ since there are no edges coming out of $u_S$ except the $(u_S, v_i)$.
    \item For a doubleton set $\{i,j\}$, if $i\in I$ and $j\notin I$, then the edge capacity $c(s, u_{\{i,j\}}) = \Pr_{S'\sim\mu}(S'=S)$ coming in is larger than the edge capacity $c(u_{\{i,j\}}, v_i) = \frac{1+\alpha_j}{2}\Pr_{S'\sim\mu}(S'=S)$ going out, so $u_{\{i,j\}}\in A$.
    \item For a doubleton set $\{i,j\}$, if both $i,j\in I$, then the edge capacity $c(s, u_{\{i,j\}}) = \Pr_{S'\sim\mu}(S'=S)$ coming in is smaller than the sum of the edge capacities
    \begin{equation*}
        c(u_{\{i,j\}}, v_i) + c(u_{\{i,j\}}, v_j) = \left(\frac{1+\alpha_i}{2} + \frac{1+\alpha_j}{2}\right)\Pr_{S'\sim\mu}(S'=S)
    \end{equation*}
    going out, so $u_{\{i,j\}}\in B$.
\end{itemize}

Now we can compute the total capacity of the cut $(A,B)$:
\begin{equation*}
\begin{split}
    c(A,&B) = \sum_{S:|S|\neq 2,S\cap I\neq\emptyset} c(s, u_S) + \sum_{i,j\in I:i\neq j}c(s, u_{\{i,j\}}) + \sum_{i\in I}\sum_{j\notin I}c(u_{\{i,j\}}, v_i) + \sum_{i\notin I}c(v_i, t)\\
    & = \sum_{S:S\cap I\neq\emptyset}c(s, u_S) - \sum_{i\in I}\sum_{j\notin I}c(s, u_{\{i,j\}}) + \sum_{i\in I}\sum_{j\notin I}c(u_{\{i,j\}}, v_i) + \sum_{i\notin I}c(v_i, t)\\
    & = \sum_{S:S\cap I\neq\emptyset}\Pr_{S'\sim\mu}(S'=S) - \sum_{i\in I}\sum_{j\notin I}\Pr_{S'\sim\mu}(S'=\{i,j\})\left(1 - \frac{1+\alpha_j}{2}\right) + \sum_{i\notin I}p_i\Pr_{S'\sim\mu}(i\in S')\\
    & = \Pr_{S'\sim\mu}(S'\cap I\neq\emptyset) - \sum_{i\in I}\sum_{j\notin I}\alpha_i\alpha_j\prod_{k\notin\{i,j\}}(1-\alpha_k)\left(1-\frac{1+\alpha_j}{2}\right) + \sum_{i\notin I}p_i\alpha_i\\
    & = 1 - \prod_{i\in I}(1-\alpha_i) - \sum_{i\in I}\sum_{j\notin I}\alpha_i\alpha_j\prod_{k\notin\{i,j\}}(1-\alpha_k)\left(1-\frac{1+\alpha_j}{2}\right) + 1 - \prod_{i\in I}(1-\alpha_i) - \sum_{i\in I}p_i\alpha_i\\
    & = 1 - \prod_{i\in I}(1-\alpha_i) - \frac12\prod_{k=1}^n(1-\alpha_k)\sum_{i\in I}\frac{\alpha_i}{1-\alpha_i}\sum_{j\notin I}\alpha_j\\
    & \quad + 1 - \prod_{i\in I}(1-\alpha_i) - \left(1 - \prod_{k=1}^n(1-\alpha_k)\right)\sum_{i\in I}\alpha_i
\end{split}
\end{equation*}
using \eqref{eq:specific_no_waste} for the second-to-last equality and substituting $p_i = 1-\prod_{k=1}^n(1-\alpha_k)$ and doing some algebraic rearrangement for the last equality. Rearranging, \eqref{eq:specific_generalized_border} is equivalent to the above being least $1 - \prod_{i\in [n]}(1-\alpha_i)$. The result follows from the max-flow min-cut theorem.
\end{proof}

Next, we prove that allocation probabilities satisfying \cref{lem:specific_border_modification} do exist.
The key observation to showing that \eqref{eq:specific_generalized_border} holds is noticing that the left-hand side of the inequality is Schur-concave in the variables $\alpha_i$ for $i \in I$ and Schur-convex in the variables $\alpha_j$ for $j \not\in I$.
Using properties of Schur-convex and Schur-concave functions, it suffices to check the inequality for the special case when $\alpha_i$ is the same for all $i \in I$ and there is only one non-zero $\alpha_j$ outside of the set $I$. \cref{fig:dxl-crazy-function} plots the maximum possible value of the left hand side of \eqref{eq:specific_generalized_border}, depending on the $\alpha$ values. We emphasize how close it gets to the bound of $1$ with sets $I$ of size $5$ or higher.

\begin{lemma}
\label{lem:key_lemma}
    For any vector of agents' fair shares $(\alpha_i)$, for every $I\subseteq[n]$,
\begin{equation}
    \label{eq:modifiedmaximizationobjective}
    \left(1 - \prod_{k=1}^n (1-\alpha_k)\right)\sum_{i\in I}\alpha_i + \prod_{i\in I}(1-\alpha_i) + \frac12\prod_{k=1}^n (1-\alpha_k)\sum_{i\in I}\frac{\alpha_i}{1-\alpha_i}\sum_{j\notin I}\alpha_j\leq 1.
\end{equation}
\end{lemma}

\begin{proof}
	Notice that if $I = \emptyset$ or $I = [n]$, then this is indeed true, so assume $\emptyset\subsetneq I\subsetneq[n]$.
	Let $X = \sum_{i\in I}\alpha_i$ and
    \begin{equation*}
        K = \left\{((x_i)_{i\in I}, (y_j)_{j\notin I})\in [0,1]^I\times[0,1]^{[n]\setminus I}: \sum_{i\in I}x_i = X,\sum_{j\notin I}y_j=1-X\right\}.
    \end{equation*}
    Define the function $f:K\to\mathbb R$ by
	\begin{equation*}
		\begin{split}
			f((x_i)_{i\in I}, & (y_j)_{j\notin I}) = \left(1 - \prod_{i\in I}(1-x_i)\prod_{j\notin I}(1-y_j)\right)X                                                                      \\
			                  & + \prod_{i\in I}(1-x_i) + \frac12\left(\prod_{i\in I}(1-x_i)\right)\left(\prod_{j\notin I}(1-y_j)\right)\left(\sum_{i\in I}\frac{x_i}{1-x_i}\right)(1-X).
		\end{split}
	\end{equation*}
	The left-hand side of \eqref{eq:modifiedmaximizationobjective} is precisely $f((\alpha_i)_{i\in I}, (\alpha_j)_{j\notin I})$, so it suffices to show that the maximum of $f$ is at most $1$.
	By taking derivatives, we check in \cref{ssec:appendix_check_schur_convexity_proof} that $f(\cdot, (y_j)_{j\notin I})$ is Schur-concave for each $(y_j)_{j\notin I}$ and that $f((x_i)_{i\in I}, \cdot)$ is Schur-convex for each $(x_i)_{i\in I}$.
	Therefore, the maximum of $f$ occurs at $((x_i^\star)_{i\in I}, (y_j^\star)_{j\notin I})$ when each $x_i^\star = X/m$ where $m=|S|$ and there is a single $y_{j^\star}^\star = 1-X$ and all other $y_j^\star=0$.
	In that case,
	\begin{align*}
		\label{eq:f_simplified}
		f((x_i^\star)_{i\in I}, (y_j^\star)_{j\notin I}) &\leq \left(1-\frac{X^2 (m X+m-2 X)}{2 (m-X)}\right) \left(1-\frac{X}{m}\right)^m+X
        = 1 - (1-X)(1-g(X))
	\end{align*}
where
\begin{equation*}
	g(X) = \frac{\left(1-\frac{X}{m}\right)^m \left(m X^2+2 m X+2 m-2 X^2-2 X\right)}{2(m-X)}.
\end{equation*}
We plot $f$ in \cref{fig:dxl-crazy-function}. It is easy to verify that $g(0) = 1$, and hence if we set $X=0$ or $X=1$, then the above bound is $1$. To complete the proof, we claim that $g'(X) \leq 0$ (i.e., $g(X)$ is decreasing) in $[0,1]$.
This is indeed the case, as one can readily check
\begin{equation*}
	g'(X) = -\frac{X \left(1-\frac{X}{m}\right)^m \left(mX(m-1) + 2(m-X)\right)}{2(m-X)^2} \leq 0.
\end{equation*}
\end{proof}

\begin{figure}
    \centering
    \includegraphics[width=0.5\textwidth]{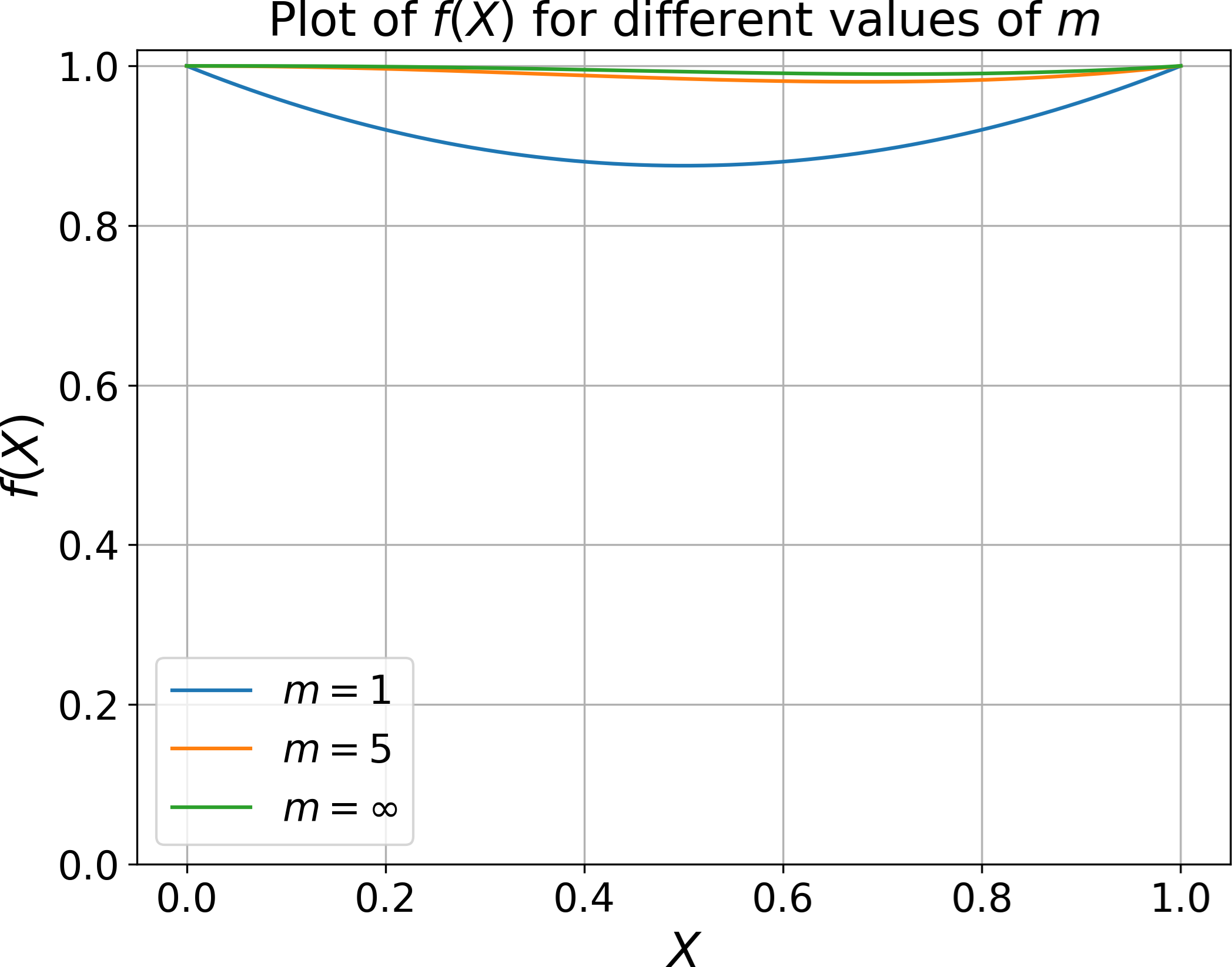}
    \caption{Plot of $f(X)$, the maximum value of the left-hand side of \eqref{eq:modifiedmaximizationobjective} for a fixed size $m$ of $I$. It gets super close to $1$, and we prove that it never exceeds $1$. }
    \label{fig:dxl-crazy-function}
\end{figure}

\cref{lem:specific_border_modification,lem:key_lemma} establish the existence of allocation probabilities $p_i^S$ inducing interim allocation probabilities $p_i = 1 - \prod_{k=1}^n(1-\alpha_k)$ that also satisfy upper bounds $p_j^{\{i,j\}} \leq \frac{1+\alpha_i}{2}$. The interim allocation probabilities show the utility guarantee at equilibrium of \cref{thm:nash}. By \cref{lem:bangforbuck} with $\bar p = \frac{1+\alpha_i}{2}$, the upper bounds prove the robustness claim of \cref{thm:robustness_guarantee}.

\subsection{Interim Allocation Probabilities Do Not Guarantee Robustness}
\label{ssec:interim_allocation_probabilities_do_not_guarantee_robustness}

In this section, we examine if the upper-bound restrictions on $p_j^{\{i, j\}}$ of \cref{lem:bangforbuck} are necessary for the robustness guarantee or just specifying the interim allocations is enough.
We will show that these restrictions are indeed necessary: we construct an example where the allocation probabilities $(p_i^S)$ induce the desired interim probabilities of \cref{thm:worst_case_interim_allocation_probabilities} but some agent $i$ cannot guarantee more than an $2\alpha_i$ fraction of her ideal utility robustly.

Consider the symmetric case where each $\alpha_i = 1/n$.
For each $S$ such that $i\in S$, let
\begin{equation}
\label{eq:assymetric_allocation_probabilities}
	p_i^S = \begin{cases}1/|S| & \text{if $1\notin S$}\\1 & \text{if $|S|\neq 2$ and $i=1$}\\0 & \text{if $|S|\neq 2$ and $1\in S$ and $i\neq 1$}\\1/n & \text{if $|S| = 2$ and $i = 1$}\\1-1/n & \text{if $|S|=2$ and $i\neq 1$}\end{cases}.
\end{equation}
In words, we set the allocation probabilities such that if agent $1$ does not bid, we allocate uniformly at random, and if agent $1$ bids and more than one other agent bids, we allocate to agent $1$ with probability $1$, and if agent $1$ bids and exactly one other agent bids, we allocate to agent $1$ with probability $1/n$.
We carefully picked these allocation probabilities such that they induce the same interim allocation probabilities $p_i = 1 - \prod_{j=1}^n(1-\alpha_j)$ as in \cref{thm:worst_case_interim_allocation_probabilities}.
In addition, the probabilities violate the restriction of \cref{lem:bangforbuck} as much as possible: $p_j^{\{1, j\}} = 1 - 1/n \gg 1/2$ when $n$ is large.

Even though the $p_i^S$ induce the correct interim allocation probabilities, they do not guarantee robustness by the following proposition, proved formally in \cref{ssec:appendix_upper_bound_proofs}.
\begin{proposition}
	\label{prop:assymetric_non_robustness}
	Using the allocation probabilities $p_i^S$ as defined by \eqref{eq:assymetric_allocation_probabilities}, player $1$ does not have a $\lambda$-robust strategy for any $\lambda > \frac{n-1}{n^2} + \frac{1}{n} - O\left(\sqrt{\nicefrac{\log T}{T}}\right)$.
\end{proposition}
\begin{proof}[Proof sketch]
	Suppose other players $i\neq 1$ take turns bidding, exactly one agent $i\neq 1$ bidding per round.
	Since they have total budget a $\frac{n-1}{n}$ fraction of the rounds, they can do this for all rounds but $T/n$.
	By how the $p_i^S$'s are set, agent $1$ can only win the item with probability $1/n$ on these rounds.
	This results in agent $1$'s utility to be extremely low.
\end{proof}

\subsection{A \texorpdfstring{$1/2$}{1/2} Hardness Result}
\label{ssec:half_robustness_hardness}

In this section, we show that our analysis for the robustness bound of \cref{thm:robustness_guarantee} is approximately tight if the fair shares are small.
Specifically, we show that the \BRB mechanism cannot guarantee every agent more than a $1/2 + \sum_i \alpha_i^2/2$ fraction of her ideal utility.
We show this for any allocation probabilities $p_i^S$ that the mechanism could use, not just the ones that induce optimal performance at equilibrium.

The following lemma is a generalization of \cref{prop:assymetric_non_robustness} and its intuition is the same as the proof sketch of \cref{prop:assymetric_non_robustness}: it stems from the case where the other players $i\neq 1$ take turns bidding, exactly one agent $i\neq i$ bidding per round.
Each agent $j\neq i$ can bid for an $\alpha_j$ fraction of the time.
Agents $j\neq i$ will not bid for a $1-\sum_{j\neq i}\alpha_j = \alpha_i$ fraction of the time, so agent $i$ can win these rounds with no competition.
At each round that agent $j$ bids, agent $i$ can only win with probability $p_i^{\{i,j\}}$ conditioned on agent $i$ bidding.
Summing the probabilities that agent $i$ can win conditioned on bidding over all times, we obtain the $\alpha_i + \sum_{j\neq i}\alpha_j p_i^{\{i,j\}}$ term below.
\begin{lemma}
    \label{lem:anticorrelatedbidding}
    In the \BRB mechanism, when the slack parameters are set as $\delta_k^T = \sqrt{\nicefrac{6\ln T}{\alpha_kT}}$, no matter what value distribution agent $i$ has, any $\lambda_i$-robust strategy of agent $i$ satisfies
    \begin{equation*}
        \lambda_i \leq \alpha_i + \sum_{j\neq i}\alpha_j p_i^{\{i,j\}} + O\left(\sqrt{\frac{\log T}{T}}\right).
    \end{equation*}
\end{lemma}
We formally prove the lemma in \cref{ssec:appendix_upper_bound_proofs}.
We can use the bound in the above lemma to show that there is always an agent that cannot have a $\lambda$-robust strategy for a constant $\lambda$ greater than $1/2$ not depending on the fair shares $(\alpha_i)$ or the time horizon $T$.
\begin{theorem}
	\label{thm:robustnesshardness}
	In the \BRB mechanism, when the slack parameters are set as $\delta_k^T = \sqrt{\nicefrac{6\ln T}{\beta_k T}}$, no matter what value distributions the agents have, if every agent $i$ has a $\lambda$-robust strategy, then
	\begin{equation*}
		\lambda \leq \frac12 + \frac12\sum_{i=1}^n \alpha_i^2 + O\left(\sqrt{\frac{\log T}{T}}\right).
	\end{equation*}
\end{theorem}
The theorem can be proved using \cref{lem:anticorrelatedbidding} by summing the terms of that lemma. The full proof is included in \cref{ssec:appendix_convex_combination_upper_bound_proof}.

\crefname{algorithm}{Algorithm}{Algorithms} 
\Crefname{algorithm}{Algorithm}{Algorithms} 

\section{Computationally Efficient \texorpdfstring{\BRB}{BRB} Mechanism}
\label{sec:computation}

In previous sections, we gave constraints on the allocation probabilities $p_i^S$ to obtain optimal equilibrium utility and $1/2$-robustness guarantees.
Specifically, we showed that if we can make the interim allocations $p_i$ equal to $1-\prod_{j\in[n]}(1-\alpha_j)$, i.e., if $\mu$ is the distribution of subsets of $[n]$ where each agent $i$ independently requests with probability $\alpha_i$,
\begin{equation}
\label{eq:equilibrium_utility_condition_repeated}
    \E_{S\sim\mu}[p_i^S\mid i\in S] = 1-\prod_{j\in[n]}(1-\alpha_j),
\end{equation}
then each agent will obtain a $(1-\prod_{j\in[n]}(1-\alpha_j))$-fraction of their ideal utility at equilibrium, and if we have the upper bounds
\begin{align}
    p_j^S & \leq \frac{1+\alpha_i}{2},\,p_i^S \leq \frac{1+\alpha_j}{2} & \text{if $S=\{i,j\}$}, \label{eq:robustness_condition_repeated}
\end{align}
then we obtain $(\frac12+\frac12\alpha_i^2)$-robustness.
\cref{lem:specific_border_modification} showed that allocation probabilities $p_i^S$ exist that satisfy \eqref{eq:equilibrium_utility_condition_repeated} and \eqref{eq:robustness_condition_repeated} simultaneously.

However, we have not shown how to compute these allocation probabilities. Because there are exponentially many $p_i^S$'s, we cannot hope to have an efficient algorithm to output all of them at once. Instead, we can only hope for an algorithm that given an $S$, outputs a probability distribution $(p_i^S)_{i\in S}$ in polynomial time that satisfies \eqref{eq:equilibrium_utility_condition_repeated} in expectation and 
\eqref{eq:robustness_condition_repeated}.
Then, at a given round in the \BRB mechanism, if a set of agents $S^*$ bids for the resource, we use the allocation probabilities $(p_i^{S^*})$ computed from our algorithm run on $S^*$.

To compute the allocation probabilities in this fashion, we use the framework presented in \cite{bhalgat2013optimal}, originally presented as a framework to design optimal Bayesian Incentive Compatible auctions.
This way, we obtain allocation probabilities that always satisfy \eqref{eq:robustness_condition_repeated}, and approximately satisfy \eqref{eq:equilibrium_utility_condition_repeated} with high probability.

\subsection{Using the Framework of \texorpdfstring{\cite{bhalgat2013optimal}}{[5]} to Compute Allocation Probabilities}

Our use of the framework in \cite{bhalgat2013optimal} is similar to the examples they present. For completeness, we show how to use their framework in the context of our problem.

\subsubsection{Multiplicative Weights Framework to Solve Linear Programs}
\label{ssec:MWU_for_LPs}

The framework of \cite{bhalgat2013optimal} uses the Multiplicative Weights algorithm (MWU) to find feasible solutions to linear programs as done in \cite{arora2012multiplicative}. We first summarize the linear programming environment.
\begin{itemize}
\item A closed set $P\subseteq\mathbb R^m$ that contains the feasible region.
\item The number of rounds $K$, which will also affect the accuracy of our solution.
\item Matrices $A^{(k)}\in \mathbb R^{l\times m}$ and vectors $b^{(k)}\in\mathbb R^l$ for each $k\in [K]$, each corresponding to a per-round set of linear constraints $A^{(k)} x \ge b^{(k)}$ for $x \in P$.
\item Oracles that solve $\max_{x\in P} y^\top A^{(k)}x$ for every given $k\in [K]$ and dual vector $y\geq 0$. 
\item A fixed (known) number $\rho$ such that $\|A^{(k)}x-b\|_\infty\leq \rho$ for each $k\in [K]$.
\end{itemize}

The goal is to find solutions $x^{(k)}$ for each $k \in [K]$ such that the total average violation $\frac{1}{K}\sum_{k} (A^{(k)}x^{(k)} - b^{(k)})$ is small (the averaging will approximate the expectation of \cref{eq:equilibrium_utility_condition_repeated}).
To do so, we run MWU for $K$ iterations, maintaining and updating a dual vector $y \in \mathbb{R}^l$. Let $y^{(k)}$ be the weight vector at iteration $k$. On iteration $k$, we use our oracle to find $x^{(k)}$ that solves $\max_{x\in P}(y^{(k)})^\top A^{(k)}x$. If $(y^{(k)})^\top A^{(k)}x^{(k)} < (y^{(k)})^\top x^{(k)}$, then we declare that the LP $A^{(k)}x\geq b^{(k)}$ is infeasible. Otherwise, we run a multiplicative weights update on the weight vector $y^{(k)}$ with cost vector $(A^{(k)}x^{(k)} - b^{(k)})/\rho$ to obtain a new weight vector $y^{(k+1)}$, i.e.,
\begin{equation*}
    y^{(k+1)}_i
    =
    \left(1 - \eta \frac{A^{(k)}x^{(k)} - b^{(k)}}{\rho} \right)
    y^{(k)}_i
\end{equation*}

This is the same as the multiplicative weights procedure in \cite{arora2012multiplicative} except that the linear programs  $A^{(k)}x\geq b^{(k)}$ change with the iteration $k$ as opposed to being fixed throughout. As pointed out by \cite{bhalgat2013optimal}, we can still obtain the following guarantee via the same analysis as in \cite{arora2012multiplicative}.
\begin{theorem}
\label{thm:multiplicative_weights_lp}
If the algorithm declares that an LP $A^{(k)}x\geq b^{(k)}$ is infeasible, then it is indeed so. Otherwise, let the learning rate used in multiplicative weights be $\eta\in(0,1/2)$. Then,
\begin{equation}
\label{eq:approximate_lp_feasibility}
	\frac{1}{K}\sum_{k=1}^K\left(A^{(k)}x^{(k)} - b^{(k)}\right) \geq -\rho\left(\eta+ \frac{\ln l}{\eta K}\right)
\end{equation}
where the inequality is interpreted coordinate-wise.
\end{theorem}

\subsubsection{Using the Multiplicative Weights Framework to Compute Allocation Probabilities}

In this section, we show how to use the setting of \cref{ssec:MWU_for_LPs} for finding allocation probabilities for the \BRB mechanism.
Specifically, we are interested in finding allocation probabilities $(p_i^S)_{S\subseteq[n], i\in S}$ that approximately satisfy \eqref{eq:equilibrium_utility_condition_repeated} and \eqref{eq:robustness_condition_repeated}.
Equivalently we want
\begin{align}
	& \frac{1}{\alpha_i}\sum_{S\subseteq[n]:i\in S}p_i^S\left(\prod_{j\in S}\alpha_j\right)\left(\prod_{j\notin S}(1-\alpha_j)\right) = 1 - \prod_{j\in[n]}(1-\alpha_j) & \forall i\in[n] \label{eq:interim_allocation_constraint},
\end{align}
and the constraints $\mathcal F(S)$ for each $S\subseteq[n]$ to hold, where $\mathcal F(S)$ consists of the constraints
\begin{align*}
	p_i^S & \geq 0 & \forall i\in S\\
	\sum_{i\in S}p_i^S & = 1 & \text{if $S\neq\emptyset$}\\
	p_j^S & \leq \frac{1+\alpha_i}{2},\,p_i^S \leq \frac{1+\alpha_j}{2} & \text{if $S=\{i,j\}$},
\end{align*}
which are linear constraints requiring that $(p_i^S)_{i\in S}$ is indeed a probability distribution over $i\in S$ that satisfy the upper bounds in \eqref{eq:robustness_condition_repeated}. The above constraints define a linear program over variables $p_i^S$.
In terms of the multiplicative weights framework, we encode the allocation probabilities $p_i^S$ as the vector $x$ and the let the set $P$ be the set of allocation probabilities satisfying $\mathcal F(S)$ for each $S\subseteq[n]$.

We now describe an algorithm to compute allocation probabilities for a single round of the \BRB mechanism. At a given time, some subset $S^*$ of agents bid, and we must produce a probability distribution $(p_i^{S^*})_{i\in S^*}$.
At each iteration $k$ of the MWU algorithm, we sample some multiset $\mathcal S^{(k)}$ of $C$ i.i.d. samples from\footnote{Recall $\mu$ is the distribution of subsets of $[n]$ where each agent $i$ appears with probability $\alpha_i$.} $\mu$. We let $\mathcal S_i^{(k)} = \{S\in\mathcal S^{(k)}:i\in S\}$ be the sets of $\mathcal S^{(k)}$ that agent $i$ is part of.
We define an approximate version of \eqref{eq:interim_allocation_constraint} using a parameter $\delta > 0$:
\begin{align}
	& \left|\frac{1}{|\mathcal S_i^{(k)}|}\sum_{S\in\mathcal S_i^{(k)}}p_i^S - \left(1 - \prod_{j\in [n]}(1-\alpha_j)\right)\right|\leq \delta& \forall i\in[n] \label{eq:approximate_interim_probability_constraint}
\end{align}
which form the constraints $A^{(k)}x\geq b^{(k)}$. Specifically, for a fixed agent $i$, we encode the above constraint as the two following (linear) constraints
\begin{align}
	\frac{1}{|\mathcal S_i^{(k)}|}\sum_{S\in\mathcal S_i^{(k)}}p_i^S & \geq \max\left\{0,1 - \prod_{j\in[n]}(1-\alpha_j) - \delta\right\} \label{eq:positive_constraint}\\
	-\frac{1}{|\mathcal S_i^{(k)}|}\sum_{S\in\mathcal S_i^{(k)}}p_i^S& \geq -\min\left\{1,1 - \prod_{j\in[n]}(1-\alpha_j) + \delta\right\}. \label{eq:negative_constraint}
\end{align}
While the maximum with $0$ and minimum with $1$ are not conceptually important, they can be used since the $p_i^S$ are probabilities, and they ensure that $\|A^{(k)}x-b\|_\infty\leq 1$ so we can use $\rho=1$ in the multiplicative weights framework.

Let us write down the objective function to the oracle problem, $\max_{x\in P}y^\top A^{(k)}x$. Let $y_{i,+}$ be the dual multiplier to \eqref{eq:positive_constraint} and let $y_{i,-}$ be the dual multiplier to \eqref{eq:negative_constraint}. Then, the objective $y^\top A^{(k)}x$ is
\begin{equation*}
	\sum_{i\in[n]}\frac{(y_{i,+} - y_{i,-})}{|\mathcal S_i^{(k)}|}\sum_{S\in\mathcal S_i^{(k)}}p_i^S = \sum_{S\in\mathcal S_i^{(k)}}\sum_{i\in S}\frac{(y_{i,+} - y_{i,-})p_i^S}{|\mathcal S_i^{(k)}|}
\end{equation*}
The right-hand side writes the objective as the sum of functions over $S\in\mathcal S_i^{(k)}$, so the problem decouples. Specifically, to solve the oracle problem $\max_{x\in P}y^\top A^{(k)}x$, we just need to solve individual problems $\mathcal A^{(k)}(\mathcal S,y,S)$ for each fixed $S\in\mathcal S_i^{(k)}$, where $y$ denotes the vector of dual variables $y_{i,+}$ and $y_{i,-}$ and $\mathcal A^{(k)}(\mathcal S,y,S)$ is the following problem.
\begin{equation}
\label{eq:decoupled_oracle_problem}
	\mathcal A^{(k)}(\mathcal S,y,S):\quad\max_{(p_i^S)_{i\in S}}\ \sum_{i\in S}\frac{(y_{i,+} - y_{i,-})p_i^S}{|\mathcal S_i^{(k)}|} \quad \text{subject to}\quad \mathcal F(S)
\end{equation}

We emphasize that \eqref{eq:decoupled_oracle_problem} is simple to solve: it is linear and has at most $n$ variables and $4$ linear constraints.

At the end of the $K$ iterations of MWU, we will either have outputted that some LP was infeasible, or we will have dual vectors $y^{(k)}\in\mathbb R^{2n}$ for each $k\in [K]$. In the case that we have outputted that some LP in infeasible, we can output arbitrary $y^{(k)}\in\mathbb R^{2n}$.
Otherwise, we output $(p_i^S)^{(k)}$ as the solution to each $\mathcal A^{(k)}(\mathcal S^{(k)},y^{(k)}, S)$. Then, when the actual set of bidding agents is $S^*$, we allocate according to the probability distribution $\frac1K\sum_{k=1}^K(p_i^{S^*})^{(k)}$.

We give pseudocode for the above algorithm below. In \cref{alg:compute_dual_vectors}, we describe the procedure to compute the dual vectors. In \cref{alg:compute_allocation_probabilities}, we describe how to use \cref{alg:compute_dual_vectors} to compute allocation probabilities.
\floatname{algorithm}{Algorithm}
\begin{algorithm}
    \caption{Compute Dual Vectors}
    \label{alg:compute_dual_vectors}
    \begin{algorithmic}
        \REQUIRE Fair shares $\alpha_i$, number of multiplicative weights iterations $K$, multiplicative weights learning rate $\eta$, number of samples of sets of bidding agents at each iteration $C$, slack parameter $\delta$.

        \STATE Initialize dual variables $y_{i,+}^{(1)}=1$ and $y_{i,-}^{(1)}=1$ for each $i\in[n]$.
        \FOR{$k=1,2,\dots,K$}
            \STATE Sample a multiset $\mathcal S^{(k)}$ of $C$ i.i.d. samples from $\mu$.
            \STATE For each $S\in\mathcal S^{(k)}$, find $(p_i^S)_{i\in S}^{(k)}$ by solving $\mathcal A^{(k)}(\mathcal S^{(k)},y^{(k)}, S)$ defined in \eqref{eq:decoupled_oracle_problem}.
            \STATE Define $\mathcal S_i^{(k)} = \{S\in\mathcal S^{(k)}:i\in S\}$.
            \STATE Define the costs $c_{i,+}^{(k)} = \frac{1}{|\mathcal S_i^{(k)}|}\sum_{S\in \mathcal S_i^{(k)}}p_i^S - \max\left\{0, 1 - \prod_{j\in[n]}(1-\alpha_j) - \delta\right\}$ and $c_{i,-}^{(k)} = -\frac{1}{|\mathcal S_i^{(k)}|}\sum_{S\in \mathcal S_i^{(k)}}p_i^S + \min\left\{1, 1 - \prod_{j\in[n]}(1-\alpha_j) + \delta\right\}$.
            \STATE Update the dual variables as as in the MWU algorithm: $y_{i,+}^{(k+1)} = (1-\eta c_{i,+})y_{i,+}^{(k)}$ and $y_{i,-}^{(k+1)} = (1-\eta c_{i,-})y_{i,-}^{(k)}$.
        \ENDFOR
          \ENSURE Dual vectors $y^{(k)} = (y_{i,+}^{(k)}, y_{i,-}^{(k)})_{i\in[n]}$ for each $k\in K$ and samples $\mathcal S^{(k)}$.  \end{algorithmic}
\end{algorithm}

\begin{algorithm}
    \caption{Compute Allocation Probabilities}
    \label{alg:compute_allocation_probabilities}
    \begin{algorithmic}
        \REQUIRE Set $S^*$ of bidding agents.
        \STATE Obtain dual vectors $y^{(k)}$ and samples $\mathcal S^{(k)}$ for each $k\in [K]$ from \cref{alg:compute_dual_vectors}.
        \STATE Obtain $(p_i^{S^*})^{(k)}$ by solving $\mathcal A^{(k)}(\mathcal S^{(k)},y^{(k)}, S^*)$ for each $k\in [K]$.
        \STATE Output $(p_i^{S^*}) := \frac1K\sum_{k=1}^K(p_i^{S^*})^{(k)}$.
    \end{algorithmic}
\end{algorithm}

The following lemma, which we prove in \cref{ssec:appendix_computation_proofs}, shows that the allocation probabilities approximately induce the correct interim allocation probabilities with high probability, that holds independently for each requesting set $S^*$.
\begin{lemma}
\label{lem:computation_lemma}
Suppose we run \cref{alg:compute_allocation_probabilities} with $S^*\sim\mu$. When we run \cref{alg:compute_dual_vectors} as a subroutine, given a fixed parameter $K$, choosing the other parameters as
\begin{equation}
\label{eq:computation_lemma_parameter_choice}
	\eta = \min\left\{\frac12,\sqrt{\frac{\ln n}{K}}\right\},C = \left\lceil K\left(\frac{\ln nK}{\underline\alpha}\right)^2\right\rceil, \delta = \sqrt{\frac{3\ln nK}{\underline\alpha C}}
\end{equation}
there is an event $E$ contained in the $\sigma$-algebra generated by the samples $(\mathcal S^{(k)})_{k\in [K]}$ of probability at least $1- O\left(\frac1{n^2K^2}\right)$ such that on $E$, for every agent $i$,
\begin{equation*}
	\left|\E_{\substack{S^*\sim\mu}}\left[p_i^{S^*}\,\middle|\, i\in S^*, (\mathcal S^{(k)})_{k\in [K]}\right] - \left(1 - \prod_{j\in [n]}(1-\alpha_j)\right)\right|\leq O\left(\sqrt{\frac{\log n}{K}}\right)
\end{equation*}
where $\underline \alpha = \min_j \alpha_j$.

Also, when choosing such parameters the algorithm runs in time $O(\poly(K, n, 1/\underline\alpha))$.
\end{lemma}

\subsection{\texorpdfstring{\BRB}{BRB} with Online Computation of Allocation Probabilities}

We now describe how to run our \BRB mechanism while efficiently computing the allocation probabilities online. At the beginning of the mechanism, we use \cref{alg:compute_dual_vectors} to obtain dual vectors $(y^{(k)})$. Then, at each time step $t$, given the set of bidding agents $S^t$, we can use the dual vectors $(y^{(k)})$ to obtain allocation probabilities $(p_i^{S^t})$ to be used in round $t$ using \cref{alg:compute_allocation_probabilities}. The formal description of our algorithm is in \cref{alg:compute_BRB}.

Our argument to show that this maintains the robustness and equilibrium guarantees goes as follows. The allocation probabilities $p_i^{S^*}$ outputted by \cref{alg:compute_allocation_probabilities} necessarily satisfy $\mathcal F(S^*)$ so they will satisfy \eqref{eq:robustness_condition_repeated}. We will use \cref{lem:computation_lemma} to show that \eqref{eq:equilibrium_utility_condition_repeated} is approximately satisfied with high probability. The equilibrium utility and robustness will then follow.

\begin{algorithm}
    \caption{\BRB with Efficient Computation of Allocation Probabilities $p_i^S$.}
    \label{alg:compute_BRB}
    \begin{algorithmic}
        \STATE Obtain dual vectors $y^{(k)}$ and samples $\mathcal S^{(k)}$ for each $k\in [K]$ from \cref{alg:compute_dual_vectors}.
        \STATE We run \BRB as usual and compute the allocation probabilities online from $y^{(k)}$ and $\mathcal S^{(k)}$ as follows.
        \FOR{$t=1,2,\dots,T$}
            \STATE Observe the set of bidding agents $S^t$.
            \STATE Sample $k^t\sim[K]$.
            \STATE Use the allocation probabilities $(p_i^{S^t})^{(k^t)}_{i\in S^t}$ given by solving $\mathcal A^{(k^t)}(\mathcal S^{(k^t)}, y^{(k^t)}, S^t)$.
        \ENDFOR
    \end{algorithmic}
\end{algorithm}

Formally, we have the following, which we prove in \cref{ssec:appendix_computation_proofs}.
Due to the error by the sampling being proportional to $\tilde O(1 / \sqrt T)$, we get the same performance as in \cref{thm:robustness_guarantee}.

\begin{theorem}\footnote{Outside of this section, we use big $O$ notation to only indicate dependence on $T$. In all our uses, this suppresses only polynomial dependence on $n$ and $(1/\alpha_j)_{j\in[n]}$. In this section, we explicitly keep the $\poly(n, (1/\alpha_j)_{j\in[n]})$ factors in the big $O$ notation.}
\label{thm:compute_dual_vectors_once}
\cref{alg:compute_BRB} with $K = T$ and other parameters set as in \cref{lem:computation_lemma} has the following guarantees.
Each player $i$ playing an $\alpha_i$-aggressive strategy is an $O\left(\sqrt{\frac{\log T}{T}}\poly(n, (1/\alpha_j)_{j\in[n]})\right)$-equilibrium in which player $i$ receives utility
\begin{equation*}
	\frac1T\sum_{t=1}^T U_i[t] \geq \left(1-\prod_{j\in[n]}(1-\alpha_j)\right)v_i^\star - O\left(\sqrt{\frac{\log T}{T}}\poly(n, (1/\alpha_j)_{j\in[n]})\right)
\end{equation*}
with probability at least $1-O(\poly(n, (1/\alpha_j)_{j\in[n]})/T^2)$. Furthermore, regardless of the behavior of other agents $j\neq i$, if agent $i$ plays an $\alpha_i$-aggressive strategy, she is guaranteed utility
\begin{equation*}
	\frac1T\sum_{t=1}^T U_i[t] \geq \frac12 + \frac12\alpha_i^2 - O\left(\sqrt{\frac{\log T}{T}}\poly(n, (1/\alpha_j)_{j\in[n]})\right)
\end{equation*}
with probability at least $1-O(\poly(n, (1/\alpha_j)_{j\in[n]})/T^2)$.

\end{theorem}

\crefname{algorithm}{Mechanism}{Mechanisms} 
\Crefname{algorithm}{Mechanism}{Mechanisms} 

\appendix
\section{Warm Up: Repeated Allocation with Two Agents}
\label{sec:twoagents}

To motivate the design of our mechanism and show how we can achieve good robust equilibria as outlined in \cref{ssec:ideal_benchmarks}, we first give a simplified version of our mechanism in the special case where there are only two agents.
Our mechanism lets each agent decide to bid for the resource or not in each round, up to some limited number of bids.
It makes sense that when only one agent bids (and has not exceeded her bid budget), she is guaranteed the item.
What is unclear is who gets allocated when both agents bid.
We show that carefully randomizing this allocation leads to the optimal guarantee of \cref{lem:centralallocationworstcase}, both in terms of the equilibrium and the robust setting, i.e., $\lnash = \lrob = 1 - (1 - \alpha_1)(1 - \alpha_2)$.
To extend this result to $n$ agents, we need to set exponentially many parameters (the allocations when any subset $S$ agents bids), which we do by strengthening Border's theorem in \cref{sec:border}.

Our mechanism proceeds as follows:
Each agent $i$ starts with $\alpha_i T $ bid tokens.
At each round, each agent can bid $0$ or $1$.
If both agents bid $0$, no one receives the item.
If only one agent bids $1$, then she receives the item in that round.
If both agents bid $1$, then we allocate the item to agent $1$ with probability $p_1^{\{1,2\}}=p$ and to agent $2$ with probability $p_2^{\{1,2\}}=1-p$.
Regardless of who wins the item, each agent $i$ pays their bid.
If an agent's budget becomes $0$, they can not bid.

Suppose agent $1$ has a fair share $\alpha_1 = \alpha$ and agent $2$ has a fair share $\alpha_2 = 1-\alpha$.
For simplicity, we assume that agent $i$ has value distribution with a continuous CDF $F_i$, and assume that the budget constraint is enforced only in expectation. (This is to help simplify the description of this warm-up case; we later fix this using inflated budgets, as in~\cite{gorokh2021monetary}).
We now show that there is a simple Nash equilibrium for any $p$, where each agent bids only when their value is in the top $\alpha_i$-quantile of her value distribution.

\begin{proposition}[Informal]
	Each agent $i$ bidding at time $t$ if her value $V_i[t]$ is in the top $\alpha_i$-quantile of her value distribution is a Nash equilibrium.
\end{proposition}
\begin{proof}[Proof sketch]
	First, observe that the expected spending of each agent $i$ under this strategy profile is indeed at most $B_i[1] = \alpha_i T$, since each agent will be bidding i.i.d. $\mathrm{Bernoulli}(\alpha_i)$. Next, note that there is nothing that agent $1$ can do to affect player $2$'s behavior.
	She must solve her own stochastic control problem where her policy does not affect agent $2$'s behavior.
	At any time $t$, if agent $1$ bids, her probability of winning is $p_1 = \alpha + (1-\alpha)p$,
	where the $\alpha$ term is the probability that agent $2$ does not bid, and the $(1-\alpha)p$ term is the probability that agent $2$ bid, and agent $1$ wins give both agents bid.
	Thus, agent $1$'s control problem reduces to selecting $\alpha_1T$ time periods in which to bid, and she wins a $p_1$ fraction of them in expectation. Agent $1$ should bid on her highest-valued $\alpha_iT$ time periods, which is precisely what bidding whenever her value $V_i[t]$ is in the top $\alpha_i$-quantile of her distribution does (in expectation).
\end{proof}

We next pick $p$.
Note that agent $1$'s expected utility in the above equilibrium is
\begin{equation*}
    \frac1T\sum_{t=1}^T \mathbb E[U_1[t]] = \frac1T\sum_{t=1}^T p_1\mathbb E[V_1[t]\pmb1\{V_1 > F_1^{-1}(1-\alpha_1)\}] = p_1v_1^\star
\end{equation*}
where $p_1 = \alpha + (1-\alpha)p$ is the probability agent $1$ wins a round conditioned on her bidding. Symmetrically, player $2$ gets expected utility $\frac1T\sum_t \mathbb E[U_2[t]] = p_2v_2^\star$ where $p_2 = (1-\alpha) + \alpha(1-p) = 1 - \alpha p$.
Setting $p=1-\alpha$ guarantees both agents $1-\alpha(1-\alpha)$ fraction of their ideal utility.
This achieves $\lnash = 1-\alpha(1-\alpha)$ and by \cref{lem:centralallocationworstcase}, this is the best factor any allocation rule can get.

The choice of $p = 1-\alpha$ might seem unusual, especially when $\alpha$ is close to $0$ or $1$ where one agent wins most times. The intuition behind this is as follows: if agent $1$ has small fair share and agent $2$ has a high one, then using a high $p$ and favoring agent $1$ has a small effect on agent $2$, most of whose bids are uncontested.
Therefore, if we want equal outcomes, i.e., $p_1 = p_2$, we have to favor the agent with the small fair share, breaking most ties in her favor.

In this special case, it is not hard to argue that agent $1$ enjoys the same utility guarantee \emph{even if agent $2$ acts adversarially}.
Since agent $1$'s bids are i.i.d. across time, the worst agent $2$ can do is bid as much as her budget constraint allows, which is the same as her equilibrium behavior.
Using a symmetric argument for agent $2$, we get the following proposition.
\begin{proposition}[Informal]
	In the above mechanism with $p = 1 - \alpha$, agent $i$ bidding when her value is in her top $\alpha_i$-quantile is a $\lrob$-robust strategy with $\lrob=\qty\big( 1 - \alpha_i(1-\alpha_i) )$.
\end{proposition}

This means we achieve $\lrob = 1 - (1-\alpha_1)(1-\alpha_2)$ which also matches the upper bound of any allocation of \cref{lem:centralallocationworstcase}.

In coming sections, we generalize this mechanism for $n \ge 2$ players, where the main difficulty is resolving concurrent bids for any subset $S \subseteq [n]$ with $|S| \ge 2$.
Since there are exponentially many such subsets, it is much harder to find a good or simple allocation rule that has the same properties as our choice of $p$ above.
We will show that we can get the same optimal factor under equilibrium performance, $\lnash = 1 - \prod_i (1 - \alpha_i)$, using Border's theorem.
For the robust guarantee, we show that a simple application of Border's theorem is not enough.
Instead, we have to consider more restrictions on how we handle concurrent bids, which leads to the strengthening of Border's theorem in \cref{lem:specific_border_modification}. 
Unfortunately, the robust guarantee weakens for $n > 2$, and we get $\lrob \approx 1/2$, which we show is tight under any rule for resolving concurrent bids.

\section{Better Equilibrium Guarantees}
\label{sec:beta_mechanism}
We designed the \mechanism mechanism to guarantee each agent a $1-\prod_{j=1}^n(1-\alpha_j)$ fraction of their $\alpha_i$-ideal utility.
While this may be the best factor with worst-case value distributions, there are allocation procedures that can do a lot better with other value distributions \cite{chido,fikioris2023online}.

We present a very simple way to generalize our mechanism to achieve better equilibrium guarantees when the value distributions are not worst case.
We generalize this mechanism by parameterizing it in two ways.
\begin{enumerate}
    \item
          Set the per-round budgets to a parameter $\beta_i$.
          That is, instead of giving everyone $\alpha_i(1+\delta_i^T)T$ tokens, we give them $\beta_i(1+\delta_i^T)T$ tokens where $\beta_i$ is a parameter.
    \item
          Instead of always choosing the allocation probabilities $p_i^S$ in a very specific way as in \cref{thm:worst_case_interim_allocation_probabilities}, let the $p_i^S$'s be parameters.
\end{enumerate}
We formally give this parameterized mechanism in \cref{alg:general_mechanism}.

\floatname{algorithm}{Mechanism}
\begin{algorithm}
    \caption{Generalized Budgeted Border}
    \label{alg:general_mechanism}
    \begin{algorithmic}
        \REQUIRE Per-round budgets $\beta_i\in[0,1]$ and slack parameters $\delta_i^T$.
        \REQUIRE Allocation probabilities $(p_i^S)_{\substack{i\in[n]\\S\subseteq[n]}}$ where $(p_i^S)_{i\in [n]}$ is a probability distribution over agents $i\in S$.
        \STATE Endow each agent with $B_i[1] = \beta_i(1+\delta_i^T)T$ bid tokens.
        \FOR{$t=1,2,\dots,T$}
            \STATE Agents submit bids $b_i^t\in\{0,1\}$.
            \STATE Budgets are enforced: $b_i^t\gets 0$ for each $i$ such that $B_i[t]\leq 0$.
            \STATE Let $S[t] = \{i:b_i^t=1\}$ be the set of bidding agents.
            \STATE A winner $i^t$ is randomly selected from $S[t]$ according to $(p_i^{S[t]})_{i\in[n]}$.
            \STATE Budgets get updated: $B_i[t+1] = B_i[t] - b_i^t$ for every agent $i$.
        \ENDFOR
    \end{algorithmic}
\end{algorithm}

Regardless of choice of $\beta_i$'s and $p_i^S$'s, there is a similar equilibrium as in \cref{prop:approximate_nash} and similar equilibrium utility guarantee as in \cref{thm:nash}.
Before stating this formally, let us define some useful terminology.
First, we generalize the notion of ideal utility.
Adopting terminology from \cite{fikioris2023online}, the \textit{$\beta$-ideal utility} of an agent $i$ is the maximum expected utility they can obtain from a single round if they can obtain the item simply by requesting it, but they are only allowed to request it with probability at most $\beta$.
Formally, the (per-round) $\beta$-ideal utility is the following.
\begin{definition}[$\beta$-ideal utility, $\beta$-ideal utility probability function]
    \label{def:ideal_utility}
    Agent $i$'s $\beta$-ideal utility is the value of the following maximization problem over measurable $\rho_i:[0,\infty)\to[0,1]$:
    \begin{equation*}
        \label{eq:idealutility}
        \begin{split}
            v_i^\star(\beta) = \max\,\E_{V_i\sim\mathcal F_i}[V_i\rho_i(V_i)] \quad\text{subject to}\quad \E_{V_i\sim\mathcal F_i}[\rho_i(V_i)]\leq \beta.
        \end{split}
    \end{equation*}
    We call the optimal solution $\rho_i$ to the above optimization problem agent $i$'s \textit{$\beta$-ideal utility probability function}, which we denote by $(\rho_i^\beta)^\star$.
\end{definition}
Note that the ideal utility, as we previously defined it in \cref{def:alpha_ideal_utility}, is just the $\alpha_i$-ideal utility.

Recall the notion of $\beta$-aggressive strategy from \cref{def:beta_aggressive}, where an agent bids when her value is in the top $\beta$-quantile of her value distribution (subject to her budget constraint). We can also rephrase a $\beta$-aggressive strategy as bidding with probability $(\rho_i^\beta)^\star(V_i[t])$ (subject to the budget constraint).

For the same reason as in \BRB, each agent $i$ playing a $\beta_i$-aggressive strategy is an approximate Nash equilibrium. We formally prove the following in \cref{ssec:appendix_equilibrium_proofs}.
\begin{proposition}
\label{prop:beta_approximate_nash}
    By setting $\delta_i^T = \sqrt{\frac{6\ln T}{\beta_i T}}$, no matter the choice of allocation probabilities $p_i^S$, each player $i$ playing a $\beta_i$-aggressive strategy is an $O\left(\sqrt{\frac{\log T}{T}}\right)$-approximate Nash equilibrium.
\end{proposition}

To discuss the utility guarantee at this approximate equilibrium, we define interim allocation probabilities similar to \cref{def:alpha_interim_allocation_probability}, except that the interim allocation probabilities $p_i$ depend on the per-round budgets $\beta_i$. The previous definition of interim allocation probability is a special case where $\beta_i=\alpha_i$.
\begin{definition}
\label{def:beta_interim_allocation_probability}
Fix per-round bid rates $\beta_i$.
Allocation probabilities $p_i^S$ induce interim allocation probabilities $p_i$ if $p_i$ is the probability that agent $i$ wins the item in a given round conditioned on agent $i$ bidding that round and agents $j\neq i$ bidding independently with probability $\beta_j$ each. Formally, this means that
\begin{equation*}
		\label{eq:beta_interim_allocation_probability}
		p_i = \sum_{S\subseteq[n]:i\in S}p_i^S\left(\prod_{j\in S\setminus\{i\}}\beta_j\right)\left(\prod_{j\in[n]\setminus S}(1-\beta_j)\right).
	\end{equation*}
\end{definition}

We now obtain a similar equilibrium utility guarantee as in \cref{thm:nash} that we prove formally in \cref{ssec:appendix_equilibrium_proofs}.
\begin{theorem}
    \label{thm:general_nash}
    By setting $\delta_i^T = \sqrt{\frac{6\ln T}{\beta_i T}}$, at the approximate equilibrium where each player $i$ plays a $\beta_i$-aggressive strategy, with probability at least $1 - O(1/T^2)$, player $i$ gets utility
    \begin{equation*}
        \frac1T\sum_{t=1}^T U_i[t] \geq p_iv_i^\star(\beta_i) - O\left(\sqrt{\frac{\log T}{T}}\right).
    \end{equation*}
\end{theorem}

As before, Border's Theorem gives a clean characterization of which numbers $p_i\in[0,1]$ can actually be induced by some allocation probabilities $p_i^S$.
Specifically, the below theorem follows from Border's Theorem and is analogous to \cref{thm:alpha_border}, which we explain in \cref{sec:border_details}.
\begin{theorem}
    \label{thm:beta_border}
    Given per-round budgets $\beta_i$ and probabilities $p_i$, there exist allocation probabilities $p_i^S$ such that the $p_i$'s are interim allocation probabilities induced by the allocation probabilities $p_i^S$ if and only if
    \begin{equation}
        \label{eq:borderscriterionnowaste}
        \sum_{i\in [n]}p_i\beta_i = 1 - \prod_{i\in[n]}(1-\beta_i)
    \end{equation}
    and for any $I\subseteq[n]$,
    \begin{equation}
        \label{eq:borderscriterion}
        \sum_{i\in I}p_i\beta_i\leq 1- \prod_{i\in I}(1-\beta_i).
    \end{equation}
\end{theorem}

In general, the principal choose $\beta_i$'s and $p_i$'s in any way they like that satisfies \cref{thm:beta_border}.
Let's give a particular suggestion.
Suppose we set the $\beta_i$'s to be proportional to the fair shares $\alpha_i$.
We can obtain the following using \cref{thm:beta_border}, which has a very similar proof to that of \cref{thm:worst_case_interim_allocation_probabilities} that we formally give in \cref{ssec:interim_allocation_probability_feasibility_proofs}.
\begin{corollary}
    \label{cor:proportional_budgets}
    Let $0 < \gamma \leq \min_{i\in[n]}1/\alpha_i$ and let $\beta_i = \gamma\alpha_i$.
    There exist allocation probabilities $(p_i^S)$ that induce the interim allocation probabilities
    \begin{equation*}
        p_i = \frac{1-\prod_{j\in [n]}(1-\gamma\alpha_j)}{\gamma}.
    \end{equation*}
    These $p_i$ satisfy
    \begin{equation*}
        p_i \geq \frac{1-e^{-\gamma}}{\gamma}.
    \end{equation*}
\end{corollary}

Using the above $p_i^S$'s in our mechanism, at the equilibrium, each agent $i$ gets expected per-round utility at least $\frac{1-e^{-\gamma}}{\gamma}v^\star(\gamma\alpha_i)$.
The function $\frac{1-e^{-\gamma}}{\gamma}$ is decreasing in $\gamma$, while $v^\star(\gamma\alpha_i)$ is nondecreasing in $\gamma$, so if the principal wants to use the above corollary, they could set $\gamma$ appropriately based on the agents' value distributions to maximize the $\frac{1-e^{-\gamma}}{\gamma}v^\star(\gamma\alpha_i)$.
Setting $\gamma=1$ recovers the old mechanism \BRB.
However, other choices of $\gamma$ may be better if the agents had specific value distributions.
For example, the below corollary proved in \cref{ssec:interim_allocation_probability_feasibility_proofs} shows that in the symmetric agent case with $\mathrm{Uniform}([0,1])$ value distributions, we can set $\gamma$ such that each agent gets almost all of their $\alpha_i$-ideal utility, a result similar to the equilibrium guarantee in \cite{chido}.
\begin{corollary}
    \label{cor:uniform_distribution}
    Suppose each agent has fair share $\alpha_i=\frac1n$ and a $\mathrm{Uniform}([0,1])$ value distribution.
    Then, by setting $\beta_i = \gamma\alpha_i$ where $\gamma = \Theta(\log n)$, there exist allocation probabilities $(p_i^S)$ inducing interim allocation probabilities $p_i$ such that
    \begin{equation*}
        p_iv_i^\star(\beta_i) \geq \left(1 - O\left(\frac{\log n}{n}\right)\right)v_i^\star.
    \end{equation*}
\end{corollary}

\section{Any Time Guarantees and Exact Nash Equilibrium in an Infinite Time Horizon}
\label{sec:anytime}
\subsection{Any Time Guarantees}
We can obtain any time utility guarantees by enforcing the budget constraint at any time, i.e., agent $i$ can only bid at most $\alpha_i(1+\delta_i^t)t$ times by round $t$ instead of just at end time $T$. We formally describe this mechanism in \cref{alg:anytime_mechanism} (based off of \generalmechanism in \cref{sec:beta_mechanism} with general per-round budgets $\beta_i$).
\begin{algorithm}
    \caption{Any Time Budgeted Border}
    \label{alg:anytime_mechanism}
    \begin{algorithmic}
        \REQUIRE Per-round budgets $\beta_i\in[0,1]$ and slack parameters $\delta_i^t$.
        \REQUIRE Allocation probabilities $(p_i^S)_{\substack{i\in[n]\\S\subseteq[n]}}$ where $(p_i^S)_{i\in [n]}$ is a probability distribution over agents $i\in S$.
        \FOR{$t=1,2,\dots,T$}
            \STATE Agents submit bids $b_i^t\in\{0,1\}$.
            \STATE Budgets are enforced: $b_i^t\gets 0$ for each $i$ such that $\sum_{s=1}^{t-1} b_i^s\geq \beta_i(1+\delta_i^t)t$.
            \STATE Let $S[t] = \{i:b_i^t=1\}$ be the set of bidding agents.
            \STATE A winner $i^t$ is randomly selected from $S[t]$ according to $(p_i^{S[t]})_{i\in[n]}$.
            \STATE Budgets get updated: $B_i[t+1] = B_i[t] - b_i^t$ for every agent $i$.
        \ENDFOR
    \end{algorithmic}
\end{algorithm}

Formally, we obtain the following theorems, proved in \cref{ssec:appendix_anytime_proofs}.
\begin{theorem}
\label{thm:anytime_equilibrium}
By running \anytimemechanism with $\delta_i^t = \sqrt{\frac{6\ln t}{\beta_i t}}$, each agent playing a $\beta_i$-aggressive strategy is an $O\left(\sqrt{\frac{\log T}{T}}\right)$-approximate equilibrium. At this equilibrium, at each time $t$, with probability at least $1-O(1/\sqrt t)$, player $i$ gets utility
\begin{equation*}
    \frac1t\sum_{s=1}^t U_i[s] \geq p_iv_i^\star(\beta_i) - O\left(\sqrt{\frac{\log t}{t}}\right).
\end{equation*}
\end{theorem}
(Here, $p_i$ is agent $i$'s interim allocation probability with per-round budgets $(\beta_k)$ as defined in \cref{def:beta_interim_allocation_probability}.)

\begin{theorem}
\label{thm:anytime_robustness}
    By running \anytimemechanism with $\beta_i = \alpha_i$, $\delta_i^t = \sqrt{\frac{6\ln t}{\alpha_i t}}$, and allocation probabilities as guaranteed by \cref{lem:specific_border_modification,lem:key_lemma}, if agent $i$ plays an $\alpha_i$-aggressive strategy, regardless of the behavior of other agents $j\neq i$, with probability $1 - O(1/\sqrt t)$, at each time $t$, agent $i$ will obtain utility
    \begin{equation*}
        \frac1t\sum_{s=1}^t U_i[s] \geq \left(\frac12 + \frac12\alpha_i^2\right)v_i^\star - O\left(\sqrt{\frac{\log t}{t}}\right).
    \end{equation*}
\end{theorem}

\subsection{Exact Nash Equilibrium in an Infinite Time Horizon}
\label{sec:infinite}
We show that by running \anytimemechanism in an infinite time horizon, we can make the approximate Nash equilibrium in \cref{thm:nash} into an exact Nash equilibrium.
Each agent will maximize
\begin{equation*}
	\mathbb E\left[\liminf_{t\to\infty}\frac1t\sum_{s=1}^t U_i[s]\right] = \liminf_{t\to\infty}\frac1t\sum_{s=1}^t \mathbb E[U_i[s]],
\end{equation*}
observing uniform integrability for the swap of limit and expectation.

We prove the following theorem in \cref{ssec:appendix_anytime_proofs}.
\begin{theorem}
	\label{thm:exactnash}
	By running \anytimemechanism with $\delta_i^t = \sqrt{\frac{6\ln t}{\alpha_i t}}$, each player $i$ playing a $\beta_i$-aggressive strategy is a Nash equilibrium under which
	\begin{equation*}
		\frac1t\sum_{s=1}^t U_i[s] \to p_iv^\star(\beta_i)
	\end{equation*}
    almost surely as $t\to\infty$.
\end{theorem}

\section{Border's Theorem}
\label{sec:border_details}
\subsection{Obtaining \texorpdfstring{\cref{thm:alpha_border,thm:beta_border}}{Theorems B.5 and 4.2} via a Reduction from Border's Theorem}

We claimed that \cref{thm:alpha_border} and more generally, \cref{thm:beta_border}, is a special case of Border's Theorem.
We shall detail Border's Theorem and argue why our problem is a special case.
Border's Theorem deals with the setting of selling a single item to bidders via a direct-revelation mechanism.
Suppose there are $n$ bidders, each bidder $i$ with type space $\Theta_i$.
Take a given direct-revelation mechanism.
For each reported type profile $\theta = (\theta_1, \dots, \theta_n)$, suppose the allocation rule is that bidder $i$ wins with probability $p_i^\theta$.
Define the \textit{interim allocation rule} to be the mapping $\pi:\bigsqcup_i \Theta_i\to[0,1]$ where $\pi(\theta_i)$ denotes the probability that bidder $i$ will win the item conditioned on reporting $\theta_i$ assuming others are bidding truthfully.
Now suppose instead of taking a given a direct-revelation mechanism and defining the interim allocation rule, we start with an arbitrary function $\pi:\bigsqcup_i\Theta_i\to[0,1]$.
Say that $\pi$ is a \textit{feasible} interim allocation rule if $\pi$ can arise as an interim allocation rule from an actual direct-revelation mechanism's allocation rule.
Obviously, some functions $\pi$ are not feasible interim allocation rules.
For example, $\pi\equiv 1$ is not feasible if $n\geq 2$ because no mechanism can guarantee everyone a probability $1$ of winning regardless of reported type profile.

Border's Theorem gives a succinct characterization of which functions $\pi$ are feasible interim allocation rules.
We state it below.
It was proven by a sequence of previous work \cite{border1991implementation,border2007reduced,mierendorff2011asymmetric}.
\begin{theorem}[Border's Theorem]
	\label{thm:original_border}
	A function $\pi:\bigsqcup_i\Theta_i\to[0,1]$ is a feasible interim allocation rule if and only if for every measurable $\mathcal T = (\mathcal T_1, \dots, \mathcal T_n)\subseteq \Theta_1\times\dots\times\Theta_n$,
	\begin{equation*}
		\E_{\theta\sim\bigtimes_i\Theta_i}\left[\sum_{i\in[n]}\pi_i(\theta_i)\pmb1_{\{\theta\in\mathcal T\}}\right] \leq 1 - \prod_{i\in[n]}\left(1 - \Pr_{\theta_i\sim\Theta_i}(\theta_i\in\mathcal T_i)\right).
	\end{equation*}
\end{theorem}
The problem of determining which probabilities $p_i\in[0,1]$ are interim allocation probabilities induced by allocation probabilities $p_i^S$ can be seen as a special case of Border's Theorem. Specifically, assume on a given round in our mechanism that each agent $i$ bids with probability $\beta_i$, and we want to set allocation probabilities $p_i^S$ such that each agent's probability of winning the item is $p_i$ conditioned on bidding. We can think of this problem as a special case of determining whether an interim allocation rule $\pi$ is feasible as follows. Think of each agent $i$ as a bidder with $2$ types: ``bidding'' and ``not bidding'' where the probability they have the ``bidding'' type is $\beta_i$. Set $\pi$ to be the function that is $0$ on the ``not bidding'' type and $p_i$ on the ``bidding'' type for agent $i$. Determining whether it is feasible to guarantee each agent probability $p_i$ of winning conditioned on bidding is then equivalent to $\pi$ being a feasible interim allocation rule. Using Border's Theorem to characterize whether $\pi$ is feasible, combined with the additional requirements that we never allocate the item to an agent who is not bidding and we always allocate the item to some agent if at least one agent bids, we obtain the following theorem (a generalization of \cref{thm:alpha_border} and a restatement of \cref{thm:beta_border}).
\begin{theorem}
\label{thm:borderscriterion}
	At a given round in our mechanism \generalmechanism, suppose each agent $i$ bids with probability $\beta_i$ independently across agents.
	Given arbitrary probabilities $p_i\in[0,1]$, there exist allocation probabilities $p_i^S$ such that the probability that $i$ wins the item conditioned on bidding is $p_i$ (i.e., the allocation probabilities $p_i^S$ induce the interim allocation probabilities $p_i$) if and only if
	\begin{equation}
		\label{eq:border_no_waste}
		\sum_{i\in [n]}p_i\beta_i = 1 - \prod_{i\in[n]}(1-\beta_i)
	\end{equation}
	and for every $I\subseteq[n]$,
	\begin{equation}
		\label{eq:borders_criterion}
		\sum_{i\in I}p_i\beta_i \leq 1 - \prod_{i\in I}(1-\beta_i).
	\end{equation}
\end{theorem}
We can see that the ``only if'' direction is obvious as follows.
The left-hand side of \eqref{eq:border_no_waste} is the probability that some agent wins the item.
The right-hand side of \eqref{eq:border_no_waste} is the probability that some agent bids.
These must be equal since we allocate the item whenever at least one agent bids.
Now consider \eqref{eq:borders_criterion}.
The left-hand side is the probability that an agent $i\in I$ wins the item.
The right-hand side is the probability that $S$ contains an agent $i\in I$.
Since an agent can only win the item if she bids, the left-hand side can be at most the right-hand side.

Now let us formally use Border's Theorem to prove the ``if'' direction.
\begin{proof}[Proof of \cref{thm:beta_border}]
	Assume \eqref{eq:border_no_waste} and \eqref{eq:borders_criterion} hold.
	We do a reduction to the problem of determining feasibility of an interim allocation rule.
	Define a type space $\Theta_i = \{0_i,1_i\}$ for each agent $i$ where $\Pr_{\theta_i\sim\Theta_i}(\theta_i=1_i)=\beta_i$.
	Define $\pi$ by
	\begin{align*}
		\pi(0_i) & = 0,\,\pi(1_i) = p_i.
	\end{align*}
	Take any $(\mathcal T_1, \dots, \mathcal T_n)\subseteq\Theta_1\times\dots\times\Theta_n$.
	Without loss of generality, assume each $\mathcal T_i\subsetneq\Theta_i$.
	Let $I = \{i:\mathcal T_i=\{1_i\}\}$.
	Then,
	\begin{equation*}
		\E_{\theta\sim\bigtimes_i\Theta_i}\left[\sum_{i\in[n]}\pi_i(\theta_i)\pmb1_{\{\theta\in\mathcal T\}}\right] = \sum_{i\in I}p_i\beta_i \leq 1 - \prod_{i\in I}(1-\beta_i) \leq 1 - \prod_{i\in[n]}\left(1 - \Pr_{\theta_i\sim\Theta_i}(\theta_i\in\mathcal T_i)\right)
	\end{equation*}
	by \eqref{eq:borders_criterion}.
	By \cref{thm:original_border}, $\pi$ is a feasible interim allocation rule.
	Find a direct-revelation mechanism with interim allocation rule $\pi$ where $p_i^\theta$ is the probability that bidder $i$ wins if the reported type profile is $\theta$.
	Let $p_i^S = p_i^\theta$ where $\theta$ has $S = \{i\in[n]:\theta_i=1\}$.
	Notice that $0 = \pi(0_i)$ means that we never allocate the item to an bidder with type $0_i$, which implies that $p_i^S = 0$ for $i\notin S$.
	For $i\in S$, $p_i = \pi(1_i)$ is the probability that we allocate the item to a bidder with type $1_i$ conditioned on them having type $1_i$.
	By construction, this is the same conditional probability that we allocate the item to an agent $i$ conditioned on them bidding in \generalmechanism with allocation probabilities $p_i^S$.
	The sum $\sum_{i\in[n]}p_i\beta_i$ is the probability that bidder $i$ wins the item, and $1 - \prod_{i\in[n]}(1-\beta_i)$ is the probability that some agent $i$ has type $1_i$.
	By \eqref{eq:border_no_waste}, the direct-revelation mechanism must allocate the item with probability $1$ to some agent if there is an agent with reported type $1_i$.
	Since $p_i^S = 0$ unless $i\in S$, \eqref{eq:border_no_waste} implies that $\sum_{i\in S}p_i^S = 1$, so $(p_i^S)_{i\in S}$ is indeed a probability distribution over $i\in S$.
\end{proof}

\subsection{A DMMF Proof of \texorpdfstring{\cref{thm:borderscriterion}}{Theorem D.2}}
\label{sec:dmmfproofofborder}
Although \cref{thm:borderscriterion} is a special case of Border's Theorem, to aid intuition, we give novel proof of \cref{thm:borderscriterion} for our special case of allocation to bidding agents by simulating a different request-based allocation mechanism for the same setting.

We use the Dynamic Max-Min Fairness Mechanism (DMMF) introduced for our fair allocation setting by \cite{fikioris2023online}.
The DMMF mechanism goes as follows.
In each round $t$, each agent $i$ can decide to request or not.
Letting $S[t]$ be the set of agents that request, and $W_i[t]$ be the number of items that agent $i$ has won prior to and including round $t$, the principal allocates the item to the bidding agent $i\in S[t]$ that has the smallest value of $\frac{W_i[t-1]}{\alpha_i}$, the number of wins so far normalized by fair share.

The following was implicitly proven by \cite{chido} via a Foster-Lyapunov argument.
\begin{theorem}
	\label{thm:dmmfpositiverecurrentcondition}
	Define $Y_i[t] = \frac{W_i[t]}{\alpha_i} - \sum_{j=1}^n W_j[t]$.
	The vectors $Y[t] = (Y_j[t])_{j=1}^n$ form an irreducible Markov chain.
	Suppose each agent $i$ is bidding i.i.d.
	$\mathrm{Bernoulli}(\beta_i)$ across rounds and independently across agents such that for any $\emptyset\subsetneq I\subsetneq[n]$,
	\begin{equation}
		\label{eq:subgroupstability}
		\frac{1-\prod_{i\in I}(1-\beta_i)}{\sum_{i\in I}\alpha_i} > 1 - \prod_{j=1}^n(1-\beta_j).
	\end{equation}
	Then, the Markov chain $(Y[t])$ is positive recurrent, and
	\begin{equation}
		\label{eq:dmmfnumberwins}
		\frac{W_i[T]}{T}\overset{\mathrm{a.s.}}\to \alpha_i\left(1 - \prod_{j=1}^n (1-\beta_j)\right).
	\end{equation}
\end{theorem}

Now we give our DMMF proof of Border's Theorem, which will roughly go as follows.
Given interim allocation probabilities $p_i$ satisfying the conditions of \cref{thm:beta_border}, we shall define fair shares $\alpha_i$ such that when we run DMMF, the long-run fraction of items that agent $i$ wins is $\beta_i p_i$.
Let $\mu$ be the distribution over $S\subseteq[n]$ where each $i\in S$ independently with probability $\beta_i$.
We shall show that these fraction items come from the expectation of $\beta_i p_i^S$ over sets $S\sim\mu$ of bidding agents, where $p_i^S$ is the long-run fraction of times $t$ that agent $i$ would have won had they requested and $S$ were the set of agents who requested at time $t$.
These allocation probabilities $p_i^S$'s will therefore induce the interim allocation probabilities $p_i$.

\begin{proof}[DMMF proof of \cref{thm:borderscriterion}]
	Necessity of Border's Criterion for the $p_i^S$'s to exist was previously argued to be obvious in the remarks after \cref{thm:borderscriterion}.
	For sufficiency, let $(\beta_i)$ be bid rates and $(p_i)$ be interim allocation probabilities satisfying Border's Criterion.
	We shall first assume that \eqref{eq:borders_criterion} holds with strict inequality for $I\subsetneq[n]$; we will then extend our proof to the general case with a topological argument.

	Define fair shares
	\begin{equation*}
		\alpha_i = \frac{\beta_ip_i}{1 - \prod_{j=1}^n (1-\beta_j)}.
	\end{equation*}
	to be used in DMMF.
	By \eqref{eq:border_no_waste}, the fair shares indeed satisfy $\sum_{i=1}^n \alpha_i = 1$.
	With these fair shares, by rearrangement, \eqref{eq:borders_criterion} holding with strict inequality and \eqref{eq:subgroupstability} are equivalent.
	Therefore, the irreducible Markov chain $(Y[t])$ in \cref{thm:dmmfpositiverecurrentcondition} is positive recurrent.

	In the DMMF process, let $\mathcal S_i[t]$ be the set of subsets $S\subseteq[n]$ of agents such that if $S$ were to be the set of agents who requested at time $t$, agent $i$ would have won the item,
	\begin{equation*}
		\mathcal S_i[t] = \left\{S\subseteq[n]:i = \argmin_{j\in S}\frac{W_j[t-1]}{\alpha_j}\right\}.
	\end{equation*}
	Observe that $\mathcal S_i[t]$ can be written as a function of $Y[t-1]$.
	By the pointwise ergodic theorem, for any $S\subseteq[n]$,
	\begin{equation*}
		\frac1T\sum_{t=1}^T \pmb1\{S\in \mathcal S_i[t]\}\overset{\mathrm{a.s.}}\to p_i^S
	\end{equation*}
	for some constant $p_i^S\in [0,1]$.

	Let $S[t]$ be the actual set of agents that request at time $t$.
	By \eqref{eq:dmmfnumberwins} and the fact that agent $i$ actually wins the item at time $t$ if and only if $S[t]\in \mathcal S_i[t]$,
	\begin{equation*}
		\frac1T\sum_{t=1}^T\pmb1\{S[t]\in \mathcal S_i[t]\}\overset{\mathrm{a.s.}}\to\alpha_i\left(1-\prod_{j=1}^n(1-\beta_j)\right).
	\end{equation*}
        Recall that $\mu$ is the distribution over $S\subseteq[n]$ where each $i\in S$ independently with probability $\beta_i$.
	Suppose we sample the set $S\sim\mu$.
	For each $t$, observe that both $S$ and $S[t]$ have distribution $\mu$ and both are independent of $\mathcal S_i[t]$.
	Therefore,
	\begin{equation}
		\frac1T\sum_{t=1}^T \E_{S\sim\mu}[\pmb1\{S\in\mathcal S_i[t]\}] = \frac1T\sum_{t=1}^T \mathbb E[\pmb1\{S[t]\in\mathcal S_i[t]\}].
	\end{equation}
	Taking $T\to\infty$ and applying bounded convergence, we obtain
	\begin{equation}
		\E_{S\sim\mu}[p_i^S] = \alpha_i\left(1-\prod_{j=1}^n(1-\beta_j)\right) = \beta_i p_i.
	\end{equation}
	We observe that $p_i^S = 0$ if $i\notin S$, and so by definition, the $p_i^S$'s induce the interim allocation probabilities $p_i$.

    That was the proof for the case where \eqref{eq:borders_criterion} holds with strict inequality; we now give a topological argument to extend to the general case.
    Let $B$ be the set of bid rates and interim allocation probabilities $(\beta_i, p_i)$ that satisfy Border's Criterion.
    In \cref{sec:dmmfproofofborder}, we showed that the set of $(\beta_i, p_i)\in B$ satisfying \eqref{eq:borderscriterion} with strict inequality for $I\subsetneq[n]$ are induced by allocation probabilities.
    
    The set of $(\beta_i, p_i)$ that are induced by allocation probabilities is closed (as a subset of Euclidean space): the set
    \begin{equation*}
    	\left\{(\beta_i, p_i, p_i^S): p_i = \sum_{S\subseteq[n]:i\in S}p_i^S\left(\prod_{j\in S\setminus\{i\}}\beta_j\right)\left(\prod_{j\in[n]\setminus S}(1-\beta_j)\right)\right\}
    \end{equation*}
    is compact being the inverse image of a closed set under a continuous mapping, and the set of $(\beta_i, p_i)$ that are induced by allocation probabilities is just the projection of this set onto the first two coordinates.
    Hence, it suffices to show that the set of $(\beta_i, p_i)\in B$ satisfying \eqref{eq:borderscriterion} with strict inequality for $I\subseteq[n]$ is dense in $B$.
    Given any such $(\beta_i, p_i)\in B$, letting $\epsilon > 0$, define $(\beta_i',p_i')$ by
    \begin{align*}
    	\beta_i' & = \beta_i + \epsilon                                                                          \\
    	p_i'     & = \frac{\beta_ip_i}{\beta_i'}\cdot\frac{\prod_{j=1}^n(1-\beta_j')}{\prod_{j=1}^n(1-\beta_j)}.
    \end{align*}
    Then, $(\beta_i', p_i')$ can be made arbitrarily close to $(\beta_i, p_i)$.
    They can be seen to satisfy Border's Criterion where \eqref{eq:borderscriterion} is strict for $I\subsetneq[n]$ as follows.
    We have increased $\beta_i$ and decreased $p_i$ to get $(\beta_i', p_i')$ in such a way that both the left-hand side and right-hand side of \eqref{eq:borderscriterionnowaste} increase in the same amount to maintain equality.
    In \eqref{eq:borderscriterion} for $I\subsetneq[n]$, considering the changes caused by each coordinate $i\in I$ one at a time, the left-hand side increases by the same amount as in \eqref{eq:borderscriterionnowaste} but the right-hand side increases by a larger amount since the partial derivative of the right-hand side with respect to some $\beta_i\in I$ is (strictly) decreasing in $I$ (when all $\beta_j>0$ for $j\in I$).
    Therefore, with the $(\beta_i', p_i')$, \eqref{eq:borderscriterion} holds with strict inequality.
\end{proof}

\subsection{A Flow Network Proof of \texorpdfstring{\cref{thm:borderscriterion}}{Theorem D.2}}
\label{sec:borderscriterionflowproof}

A proof of Border's Theorem based on flow networks and the max-flow min-cut theorem was discovered by \cite{che2013generalized}.
We will give flow network proof in the special case of Border's Theorem that we use (\cref{thm:borderscriterion}) here for convenience.

\begin{figure}
\centering
        \tikzset{every picture/.style={line width=0.75pt}} 

\begin{tikzpicture}[x=0.75pt,y=0.75pt,yscale=-1,xscale=1]

\draw   (99,105.5) .. controls (99,97.49) and (105.49,91) .. (113.5,91) .. controls (121.51,91) and (128,97.49) .. (128,105.5) .. controls (128,113.51) and (121.51,120) .. (113.5,120) .. controls (105.49,120) and (99,113.51) .. (99,105.5) -- cycle ;
\draw   (255,23) .. controls (255,14.16) and (262.16,7) .. (271,7) .. controls (279.84,7) and (287,14.16) .. (287,23) .. controls (287,31.84) and (279.84,39) .. (271,39) .. controls (262.16,39) and (255,31.84) .. (255,23) -- cycle ;
\draw   (252,83.5) .. controls (252,73.84) and (259.84,66) .. (269.5,66) .. controls (279.16,66) and (287,73.84) .. (287,83.5) .. controls (287,93.16) and (279.16,101) .. (269.5,101) .. controls (259.84,101) and (252,93.16) .. (252,83.5) -- cycle ;
\draw   (253,158) .. controls (253,148.61) and (260.61,141) .. (270,141) .. controls (279.39,141) and (287,148.61) .. (287,158) .. controls (287,167.39) and (279.39,175) .. (270,175) .. controls (260.61,175) and (253,167.39) .. (253,158) -- cycle ;
\draw   (412,21.5) .. controls (412,13.49) and (418.49,7) .. (426.5,7) .. controls (434.51,7) and (441,13.49) .. (441,21.5) .. controls (441,29.51) and (434.51,36) .. (426.5,36) .. controls (418.49,36) and (412,29.51) .. (412,21.5) -- cycle ;
\draw   (413,77.5) .. controls (413,69.49) and (419.49,63) .. (427.5,63) .. controls (435.51,63) and (442,69.49) .. (442,77.5) .. controls (442,85.51) and (435.51,92) .. (427.5,92) .. controls (419.49,92) and (413,85.51) .. (413,77.5) -- cycle ;
\draw   (412,158.5) .. controls (412,150.49) and (418.49,144) .. (426.5,144) .. controls (434.51,144) and (441,150.49) .. (441,158.5) .. controls (441,166.51) and (434.51,173) .. (426.5,173) .. controls (418.49,173) and (412,166.51) .. (412,158.5) -- cycle ;
\draw   (528,100.5) .. controls (528,92.49) and (534.49,86) .. (542.5,86) .. controls (550.51,86) and (557,92.49) .. (557,100.5) .. controls (557,108.51) and (550.51,115) .. (542.5,115) .. controls (534.49,115) and (528,108.51) .. (528,100.5) -- cycle ;
\draw    (113.5,91) -- (253.2,23.87) ;
\draw [shift={(255,23)}, rotate = 154.33] [color={rgb, 255:red, 0; green, 0; blue, 0 }  ][line width=0.75]    (10.93,-3.29) .. controls (6.95,-1.4) and (3.31,-0.3) .. (0,0) .. controls (3.31,0.3) and (6.95,1.4) .. (10.93,3.29)   ;
\draw    (124,95) -- (250.01,82.69) ;
\draw [shift={(252,82.5)}, rotate = 174.42] [color={rgb, 255:red, 0; green, 0; blue, 0 }  ][line width=0.75]    (10.93,-3.29) .. controls (6.95,-1.4) and (3.31,-0.3) .. (0,0) .. controls (3.31,0.3) and (6.95,1.4) .. (10.93,3.29)   ;
\draw    (125,114.5) -- (251.11,157.36) ;
\draw [shift={(253,158)}, rotate = 198.77] [color={rgb, 255:red, 0; green, 0; blue, 0 }  ][line width=0.75]    (10.93,-3.29) .. controls (6.95,-1.4) and (3.31,-0.3) .. (0,0) .. controls (3.31,0.3) and (6.95,1.4) .. (10.93,3.29)   ;
\draw    (287,157.5) -- (410,158.48) ;
\draw [shift={(412,158.5)}, rotate = 180.46] [color={rgb, 255:red, 0; green, 0; blue, 0 }  ][line width=0.75]    (10.93,-3.29) .. controls (6.95,-1.4) and (3.31,-0.3) .. (0,0) .. controls (3.31,0.3) and (6.95,1.4) .. (10.93,3.29)   ;
\draw    (287,83) -- (411,77.59) ;
\draw [shift={(413,77.5)}, rotate = 177.5] [color={rgb, 255:red, 0; green, 0; blue, 0 }  ][line width=0.75]    (10.93,-3.29) .. controls (6.95,-1.4) and (3.31,-0.3) .. (0,0) .. controls (3.31,0.3) and (6.95,1.4) .. (10.93,3.29)   ;
\draw    (289,23.5) -- (410,21.53) ;
\draw [shift={(412,21.5)}, rotate = 179.07] [color={rgb, 255:red, 0; green, 0; blue, 0 }  ][line width=0.75]    (10.93,-3.29) .. controls (6.95,-1.4) and (3.31,-0.3) .. (0,0) .. controls (3.31,0.3) and (6.95,1.4) .. (10.93,3.29)   ;
\draw    (284,72) -- (411.1,30.62) ;
\draw [shift={(413,30)}, rotate = 161.97] [color={rgb, 255:red, 0; green, 0; blue, 0 }  ][line width=0.75]    (10.93,-3.29) .. controls (6.95,-1.4) and (3.31,-0.3) .. (0,0) .. controls (3.31,0.3) and (6.95,1.4) .. (10.93,3.29)   ;
\draw    (441,21.5) -- (534.36,86.85) ;
\draw [shift={(536,88)}, rotate = 214.99] [color={rgb, 255:red, 0; green, 0; blue, 0 }  ][line width=0.75]    (10.93,-3.29) .. controls (6.95,-1.4) and (3.31,-0.3) .. (0,0) .. controls (3.31,0.3) and (6.95,1.4) .. (10.93,3.29)   ;
\draw    (441,158.5) -- (534.2,113.86) ;
\draw [shift={(536,113)}, rotate = 154.41] [color={rgb, 255:red, 0; green, 0; blue, 0 }  ][line width=0.75]    (10.93,-3.29) .. controls (6.95,-1.4) and (3.31,-0.3) .. (0,0) .. controls (3.31,0.3) and (6.95,1.4) .. (10.93,3.29)   ;
\draw    (442,77.5) -- (526.07,99.98) ;
\draw [shift={(528,100.5)}, rotate = 194.97] [color={rgb, 255:red, 0; green, 0; blue, 0 }  ][line width=0.75]    (10.93,-3.29) .. controls (6.95,-1.4) and (3.31,-0.3) .. (0,0) .. controls (3.31,0.3) and (6.95,1.4) .. (10.93,3.29)   ;

\draw (109,101) node [anchor=north west][inner sep=0.75pt]   [align=left] {$\displaystyle s$};
\draw (259,104) node [anchor=north west][inner sep=0.75pt]   [align=left] {$\displaystyle \vdots $};
\draw (421,102) node [anchor=north west][inner sep=0.75pt]   [align=left] {$\displaystyle \vdots $};
\draw (538,97.4) node [anchor=north west][inner sep=0.75pt]    {$t$};
\draw (334,119.4) node [anchor=north west][inner sep=0.75pt]    {$\cdots $};
\draw (260,18.4) node [anchor=north west][inner sep=0.75pt]    {$u_{\{1\}}$};
\draw (252,75.4) node [anchor=north west][inner sep=0.75pt]    {$u_{\{1,2\}}$};
\draw (257,151.4) node [anchor=north west][inner sep=0.75pt]    {$u_{\{n\}}$};
\draw (420,18.4) node [anchor=north west][inner sep=0.75pt]    {$v_{1}$};
\draw (420,73.4) node [anchor=north west][inner sep=0.75pt]    {$v_{2}$};
\draw (418,153.4) node [anchor=north west][inner sep=0.75pt]    {$v_{n}$};
\draw (134.88,54.56) node [anchor=north west][inner sep=0.75pt]  [rotate=-334.91]  {$\Pr( S'=\{1\})$};
\draw (151.46,73.83) node [anchor=north west][inner sep=0.75pt]  [rotate=-352.96]  {$\Pr( S'=\{1,2\})$};
\draw (150.27,102.45) node [anchor=north west][inner sep=0.75pt]  [rotate=-18.18]  {$\Pr( S'=\{n\})$};
\draw (341,8.4) node [anchor=north west][inner sep=0.75pt]    {$\infty $};
\draw (324.91,44.76) node [anchor=north west][inner sep=0.75pt]  [rotate=-340.02]  {$\infty $};
\draw (336.56,65.92) node [anchor=north west][inner sep=0.75pt]  [rotate=-356.57]  {$\infty $};
\draw (336.25,142.13) node [anchor=north west][inner sep=0.75pt]  [rotate=-1.83]  {$\infty $};
\draw (459.8,9.03) node [anchor=north west][inner sep=0.75pt]  [rotate=-35.31]  {$p_{1}\Pr( 1\in S')$};
\draw (446.93,57.89) node [anchor=north west][inner sep=0.75pt]  [rotate=-17.56]  {$p_{2}\Pr( 2\in S')$};
\draw (437.83,137.84) node [anchor=north west][inner sep=0.75pt]  [rotate=-333.31]  {$p_{n}\Pr( n\in S')$};

\end{tikzpicture}
    \vspace{-1cm}
	\caption{\emph{Flow network that can be used to prove \cref{thm:beta_border}.
            We let $S'$ be the random set of bidding agents where each agent $i$ lies in $S'$ independently with probability $\beta_i$.
            Then, there are three kind of edges: edges whose flow corresponds to the probability of observing a specific $S'$ (left), edges whose flow corresponds to how we randomly allocate the item condition on observing a specific $S'$ (middle), and edges whose flow represent the probability that a specific agent gets the item (right).
    		There is a flow of value $\Pr(S'\neq\emptyset)$ if and only if there exist allocation probabilities $p_i^S$ inducing the interim allocation probabilities $p_i$.
            In other words, the flows $p_i^S\Pr(S' = S)$ in the middle transform the probabilities that agents in a certain set $S$ bid to an agent $i\in S$ being allocated.
            We obtain the conditions in \cref{thm:beta_border} by analyzing every minimum-cut of this network.
	}}
	\label{fig:beta_border_flow_network}
\end{figure}
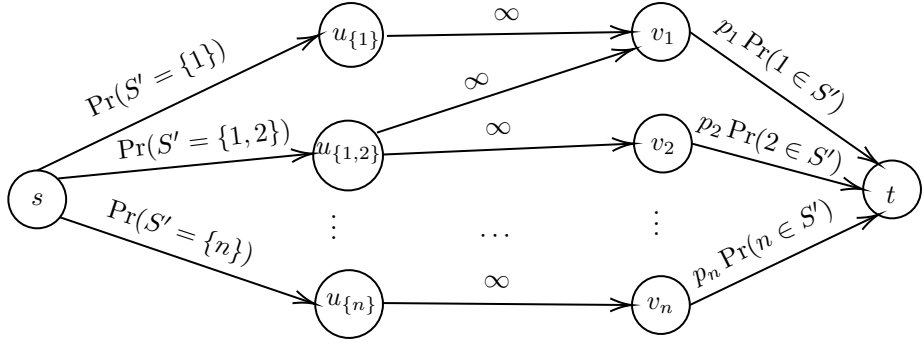

\begin{proof}[Flow network proof of \cref{thm:borderscriterion}]
    As argued in the remarks after \cref{thm:borderscriterion}, \eqref{eq:border_no_waste} is necessary for the $p_i^S$'s to exist, so we assume it.

    Let $\mu$ be the distribution over $S'\subseteq[n]$ where each $i\in S'$ independently with probability $\beta_i$.
    Create an $s$-$t$ flow network as follows.
    For each nonempty $S\subseteq[n]$ create a node $u_S$ and an edge $(s, u_S)$ with capacity
    \begin{equation}
        c(s, u_S) = \Pr_{S'\sim\mu}(S' = S).
    \end{equation}
    For each $i\in[n]$, create a node $v_i$.
    For each $S$ such that $i\in S$, add an edge $(u_S, v_i)$ with infinite capacity.
    Also, add an edge $(v_i, t)$ with capacity
    \begin{equation}
        c(v_i, t) = p_i\Pr_{S'\sim\mu}(i\in S').
    \end{equation}
    The flow network is depicted in \cref{fig:beta_border_flow_network}.

    The cut with $s$ on one side and everything else on the other has capacity
    \begin{equation}
        \sum_{S\subseteq[n]:S\neq\emptyset}c(s,u_S) = \sum_{S\subseteq[n]:S\neq\emptyset}\Pr_{S'\sim\mu}(S'=S) = \Pr_{S'\sim\mu}(S'\neq\emptyset) = 1-\prod_{i\in S}(1-\beta_i).
    \end{equation}
    The cut with $t$ on one side and everything else on the other has capacity
    \begin{equation}
        \sum_{i\in[n]}c(v_i, t) = \sum_{i\in[n]}p_i\Pr_{S'\sim\mu}(i\in S') = \sum_{i\in[n]}p_i\beta_i = 1 - \prod_{i\in[n]}(1-\beta_i),
    \end{equation}
    using \eqref{eq:border_no_waste} for the last equality.
    Observe that in this flow network, allocation probabilities $(p_i^S)$ induce the interim allocation probabilities $(p_i)$ if and only if the flow $f$ is feasible where $f(s, u_S) = c(s, u_S)$, $f(v_i,t) = c(v_i,t)$, and $f(u_S, v_i) = p_i^S\Pr_{S'\sim\mu}(S' = S)$.
    Since both the $s$-$t$ cuts with $s$ on one side and everything else on the other and the cut with $t$ on one side and everything else on the other both have cut capacity $1 - \prod_{i\in[n]}(1-\beta_i)$, it suffices to show that \eqref{eq:borders_criterion} holds only if there is a feasible flow of flow value equal to this cut capacity.

    Take any minimum-capacity $s$-$t$ cut $(A,B)$.
    Since the edges $(u_S, v_i)$ have infinite capacity, if $v_i\in B$ then $u_S\in B$ for any $S$ such that $i\in S$.
    Conversely, for any $S$, if $v_i\in A$ for every $i\in S$, then $u_S\in A$ since there are no edges coming out of $u_S$ except the $(u_S, v_i)$.
    Thus, the cut $(A,B)$ is completely characterized by which nodes $v_i\in B$.
    Let $I\subseteq[n]$ be the subset such that $v_i\in B$ if and only if $i\in I$.
    The total capacity of this cut is
    \begin{equation}
        \begin{split}
            c(A,B) & = \sum_{S\subseteq[n]:S\cap I\neq\emptyset}c(s, u_S) + \sum_{i\notin I}c(v_i, t)                                 \\
                   & = \sum_{S\subseteq[n]:S\cap I\neq\emptyset}\Pr_{S'\sim\mu}(S' = S) + \sum_{i\notin I}p_i\Pr_{S'\sim\mu}(i\in S') \\
                   & = \Pr_{S'\sim\mu}(S'\cap I\neq\emptyset) + \sum_{i\notin I}p_i\beta_i                                            \\
                   & = 1 - \prod_{i\in I}(1-\beta_i) + 1 - \prod_{i\in [n]}(1-\beta_i) - \sum_{i\in I}p_i\beta_i,
        \end{split}
    \end{equation}
    using \eqref{eq:border_no_waste} for the last line.
    Rearranging, one can see that \eqref{eq:borders_criterion} is equivalent to the above being least $1 - \prod_{i\in [n]}(1-\beta_i)$.
\end{proof}

\subsection{\texorpdfstring{\cref{thm:borderscriterion}}{Theorem D.2} with Arbitrary Upper Bounds}
The proof similar to the proof of Border's Theorem given in \cite{che2013generalized}. We construct a flow network that is feasible if and only if there are allocation probabilities inducing given interim allocation probabilities that satisfy some upper bounds. We obtain inequalities that determine whether the flow network is feasible by analyzing every possible minimum cut.

\begin{theorem}
	\label{thm:border_arbitrary_upper_bound}
    Given upper bounds $\bar p_i^S$, there exists $p_i^S$'s that induce interim allocation probabilities $p_i$ such that $p_i^S\leq \bar p_i^S$ for every $S\neq\emptyset$ if and only if
	\begin{equation}
		\label{eq:modifiedborderscriterionnowaste}
		\sum_{i\in[n]}p_i\beta_i = 1 - \prod_{i\in[n]}(1-\beta_i)
	\end{equation}
	and for any $I\subseteq[n]$,
	\begin{equation*}
		\label{eq:modifiedborderscriterion}
		\sum_{i\in I}p_i\beta_i + \prod_{i\in I}(1-\beta_i) + \sum_{S\in\mathcal S}\left(\prod_{i\in S}\beta_i\right)\left(\prod_{i\notin S}(1-\beta_i)\right)\left(1 - \sum_{i\in S\cap I}\bar p_i^S\right)\leq 1
	\end{equation*}
	where $\mathcal S = \{S\subseteq[n]:S\cap I\neq\emptyset, \sum_{i\in S\cap I}\bar p_i^S\leq 1\}$.
\end{theorem}

\begin{proof}
	Since \eqref{eq:modifiedborderscriterionnowaste} is necessary for the $p_i^S$'s to exist in \cref{thm:alpha_border}, it is also necessary here, so we assume it.

	Create an $s$-$t$ flow network as follows.
        Let $\mu$ be the probability distribution over subsets $S'\subseteq[n]$ where each $i\in S'$ independently with probability $\beta_i$.
	For each nonempty $S\subseteq[n]$ create a node $u_S$ and connect it with an edge $(s, u_S)$ to the source node with capacity
	\begin{equation*}
		c(s, u_S) = \Pr_{S'\sim\mu}(S' = S).
	\end{equation*}
	For each $i$, create a node $v_i$.
	For each $S\neq\emptyset$ such that $i\in S$, add an edge $(u_S, v_i)$ with capacity
        \begin{equation*}
           c(u_S, v_i) = \bar p_i^S\Pr_{S'\sim\mu}(S'=S)
        \end{equation*}
	Also, add an edge $(v_i, t)$ with capacity
	\begin{equation*}
		c(v_i, t) = p_i\Pr_{S'\sim\mu}(i\in S').
	\end{equation*}
	The flow network is depicted in \cref{fig:arbitrary_upper_bound_flow_network}.

\begin{figure}
\centering
        \tikzset{every picture/.style={line width=0.75pt}} 

\begin{tikzpicture}[x=0.75pt,y=0.75pt,yscale=-1,xscale=1]

\draw   (109,111.5) .. controls (109,103.49) and (115.49,97) .. (123.5,97) .. controls (131.51,97) and (138,103.49) .. (138,111.5) .. controls (138,119.51) and (131.51,126) .. (123.5,126) .. controls (115.49,126) and (109,119.51) .. (109,111.5) -- cycle ;
\draw   (255,30.5) .. controls (255,22.49) and (261.49,16) .. (269.5,16) .. controls (277.51,16) and (284,22.49) .. (284,30.5) .. controls (284,38.51) and (277.51,45) .. (269.5,45) .. controls (261.49,45) and (255,38.51) .. (255,30.5) -- cycle ;
\draw   (252,92) .. controls (252,82.61) and (259.61,75) .. (269,75) .. controls (278.39,75) and (286,82.61) .. (286,92) .. controls (286,101.39) and (278.39,109) .. (269,109) .. controls (259.61,109) and (252,101.39) .. (252,92) -- cycle ;
\draw   (253,164.5) .. controls (253,156.49) and (259.49,150) .. (267.5,150) .. controls (275.51,150) and (282,156.49) .. (282,164.5) .. controls (282,172.51) and (275.51,179) .. (267.5,179) .. controls (259.49,179) and (253,172.51) .. (253,164.5) -- cycle ;
\draw   (421,31.5) .. controls (421,23.49) and (427.49,17) .. (435.5,17) .. controls (443.51,17) and (450,23.49) .. (450,31.5) .. controls (450,39.51) and (443.51,46) .. (435.5,46) .. controls (427.49,46) and (421,39.51) .. (421,31.5) -- cycle ;
\draw   (422,87.5) .. controls (422,79.49) and (428.49,73) .. (436.5,73) .. controls (444.51,73) and (451,79.49) .. (451,87.5) .. controls (451,95.51) and (444.51,102) .. (436.5,102) .. controls (428.49,102) and (422,95.51) .. (422,87.5) -- cycle ;
\draw   (421,168.5) .. controls (421,160.49) and (427.49,154) .. (435.5,154) .. controls (443.51,154) and (450,160.49) .. (450,168.5) .. controls (450,176.51) and (443.51,183) .. (435.5,183) .. controls (427.49,183) and (421,176.51) .. (421,168.5) -- cycle ;
\draw   (538,112.5) .. controls (538,104.49) and (544.49,98) .. (552.5,98) .. controls (560.51,98) and (567,104.49) .. (567,112.5) .. controls (567,120.51) and (560.51,127) .. (552.5,127) .. controls (544.49,127) and (538,120.51) .. (538,112.5) -- cycle ;
\draw    (123.5,97) -- (253.22,31.4) ;
\draw [shift={(255,30.5)}, rotate = 153.17] [color={rgb, 255:red, 0; green, 0; blue, 0 }  ][line width=0.75]    (10.93,-3.29) .. controls (6.95,-1.4) and (3.31,-0.3) .. (0,0) .. controls (3.31,0.3) and (6.95,1.4) .. (10.93,3.29)   ;
\draw    (134,101) -- (250.01,92.15) ;
\draw [shift={(252,92)}, rotate = 175.64] [color={rgb, 255:red, 0; green, 0; blue, 0 }  ][line width=0.75]    (10.93,-3.29) .. controls (6.95,-1.4) and (3.31,-0.3) .. (0,0) .. controls (3.31,0.3) and (6.95,1.4) .. (10.93,3.29)   ;
\draw    (135,120.5) -- (251.13,163.8) ;
\draw [shift={(253,164.5)}, rotate = 200.45] [color={rgb, 255:red, 0; green, 0; blue, 0 }  ][line width=0.75]    (10.93,-3.29) .. controls (6.95,-1.4) and (3.31,-0.3) .. (0,0) .. controls (3.31,0.3) and (6.95,1.4) .. (10.93,3.29)   ;
\draw    (282,164.5) -- (419,168.44) ;
\draw [shift={(421,168.5)}, rotate = 181.65] [color={rgb, 255:red, 0; green, 0; blue, 0 }  ][line width=0.75]    (10.93,-3.29) .. controls (6.95,-1.4) and (3.31,-0.3) .. (0,0) .. controls (3.31,0.3) and (6.95,1.4) .. (10.93,3.29)   ;
\draw    (286,98) -- (421,94.06) ;
\draw [shift={(423,94)}, rotate = 178.33] [color={rgb, 255:red, 0; green, 0; blue, 0 }  ][line width=0.75]    (10.93,-3.29) .. controls (6.95,-1.4) and (3.31,-0.3) .. (0,0) .. controls (3.31,0.3) and (6.95,1.4) .. (10.93,3.29)   ;
\draw    (284,30.5) -- (419,31.49) ;
\draw [shift={(421,31.5)}, rotate = 180.42] [color={rgb, 255:red, 0; green, 0; blue, 0 }  ][line width=0.75]    (10.93,-3.29) .. controls (6.95,-1.4) and (3.31,-0.3) .. (0,0) .. controls (3.31,0.3) and (6.95,1.4) .. (10.93,3.29)   ;
\draw    (284,86) -- (423.09,41.61) ;
\draw [shift={(425,41)}, rotate = 162.3] [color={rgb, 255:red, 0; green, 0; blue, 0 }  ][line width=0.75]    (10.93,-3.29) .. controls (6.95,-1.4) and (3.31,-0.3) .. (0,0) .. controls (3.31,0.3) and (6.95,1.4) .. (10.93,3.29)   ;
\draw    (450,31.5) -- (544.37,98.84) ;
\draw [shift={(546,100)}, rotate = 215.51] [color={rgb, 255:red, 0; green, 0; blue, 0 }  ][line width=0.75]    (10.93,-3.29) .. controls (6.95,-1.4) and (3.31,-0.3) .. (0,0) .. controls (3.31,0.3) and (6.95,1.4) .. (10.93,3.29)   ;
\draw    (450,168.5) -- (544.18,125.83) ;
\draw [shift={(546,125)}, rotate = 155.62] [color={rgb, 255:red, 0; green, 0; blue, 0 }  ][line width=0.75]    (10.93,-3.29) .. controls (6.95,-1.4) and (3.31,-0.3) .. (0,0) .. controls (3.31,0.3) and (6.95,1.4) .. (10.93,3.29)   ;
\draw    (451,87.5) -- (536.08,111.95) ;
\draw [shift={(538,112.5)}, rotate = 196.03] [color={rgb, 255:red, 0; green, 0; blue, 0 }  ][line width=0.75]    (10.93,-3.29) .. controls (6.95,-1.4) and (3.31,-0.3) .. (0,0) .. controls (3.31,0.3) and (6.95,1.4) .. (10.93,3.29)   ;

\draw (119,107) node [anchor=north west][inner sep=0.75pt]   [align=left] {$\displaystyle s$};
\draw (259,113) node [anchor=north west][inner sep=0.75pt]   [align=left] {$\displaystyle \vdots $};
\draw (430,112) node [anchor=north west][inner sep=0.75pt]   [align=left] {$\displaystyle \vdots $};
\draw (548,109.4) node [anchor=north west][inner sep=0.75pt]    {$t$};
\draw (345,128.4) node [anchor=north west][inner sep=0.75pt]    {$\cdots $};
\draw (259,24.4) node [anchor=north west][inner sep=0.75pt]    {$u_{\{1\}}$};
\draw (252,84.4) node [anchor=north west][inner sep=0.75pt]    {$u_{\{1,2\}}$};
\draw (255,159.4) node [anchor=north west][inner sep=0.75pt]    {$u_{\{n\}}$};
\draw (429,28.4) node [anchor=north west][inner sep=0.75pt]    {$v_{1}$};
\draw (429,83.4) node [anchor=north west][inner sep=0.75pt]    {$v_{2}$};
\draw (427,163.4) node [anchor=north west][inner sep=0.75pt]    {$v_{n}$};
\draw (145.35,66.17) node [anchor=north west][inner sep=0.75pt]  [rotate=-332.74]  {$\Pr( S'=\{1\})$};
\draw (149.52,83.23) node [anchor=north west][inner sep=0.75pt]  [rotate=-354.71]  {$\Pr( S'=\{1,2\})$};
\draw (160.64,110.06) node [anchor=north west][inner sep=0.75pt]  [rotate=-21.71]  {$\Pr( S'=\{n\})$};
\draw (305,12.4) node [anchor=north west][inner sep=0.75pt]    {$\overline{p}_{1}^{\{1\}}\Pr( S'=\{1\})$};
\draw (277.45,65.62) node [anchor=north west][inner sep=0.75pt]  [rotate=-342.09]  {$\overline{p}_{1}^{\{1,2\}}\Pr( S'=\{1,2\})$};
\draw (297.56,77.92) node [anchor=north west][inner sep=0.75pt]  [rotate=-356.57]  {$\overline{p}_{2}^{\{1,2\}}\Pr( S'=\{1,2\})$};
\draw (294.25,144.13) node [anchor=north west][inner sep=0.75pt]  [rotate=-1.83]  {$\overline{p}_{n}^{\{n\}}\Pr( S'=\{n\})$};
\draw (472.43,21.53) node [anchor=north west][inner sep=0.75pt]  [rotate=-38.03]  {$p_{1}\Pr( 1\in S')$};
\draw (458.31,68.19) node [anchor=north west][inner sep=0.75pt]  [rotate=-18.53]  {$p_{2}\Pr( 2\in S')$};
\draw (453.07,147.88) node [anchor=north west][inner sep=0.75pt]  [rotate=-336.97]  {$p_{n}\Pr( n\in S')$};

\end{tikzpicture}
    \vspace{-1cm}
	\caption{\emph{Flow network that is used in the proof of \cref{thm:border_arbitrary_upper_bound}.
            The flow network is the same as \cref{fig:beta_border_flow_network} except the middle edges $(u_S, v_i)$ have capacity $\bar p_i^S\Pr(S'=S)$ to enforce the upper bounds on the allocation probabilities.
	}}
	\label{fig:arbitrary_upper_bound_flow_network}
\end{figure}
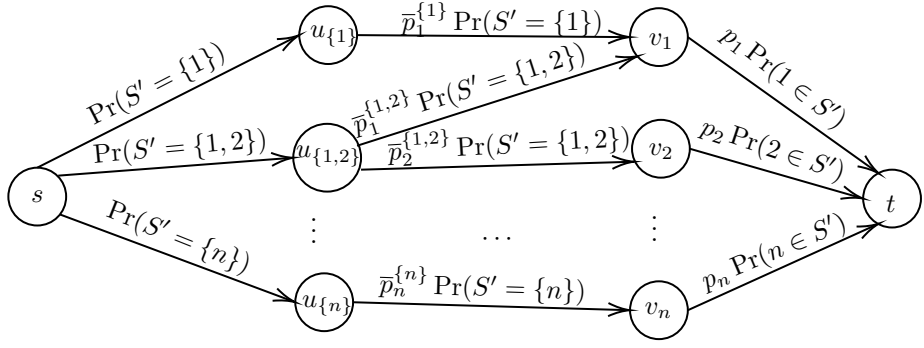
	The cut with $s$ on one side and everything else on the other has capacity
	\begin{equation*}
		\sum_{S\subseteq[n]:S\neq\emptyset}c(s,u_S) = \sum_{S\subseteq[n]:S\neq\emptyset}\Pr_{S'\sim\mu}(S'=S) = \Pr_{S'\sim\mu}(S'\neq\emptyset) = 1-\prod_{i\in S}(1-\beta_i).
	\end{equation*}
	The cut with $t$ on one side and everything else on the other has capacity
	\begin{equation*}
		\sum_{i\in[n]}c(v_i, t) = \sum_{i\in[n]}p_i\Pr_{S'\sim\mu}(i\in S') = \sum_{i\in[n]}p_i\beta_i = 1 - \prod_{i\in[n]}(1-\beta_i),
	\end{equation*}
	using \eqref{eq:modifiedborderscriterionnowaste} for the last equality.
	Observe that in this flow network, allocation probabilities $(p_i^S)$ satisfying upper bounds $p_i^S \leq \bar p_i^S$ induce the interim allocation probabilities $(p_i)$ if and only if the flow $f$ is feasible where $f(s, u_S) = c(s, u_S)$, $f(v_i,t) = c(v_i,t)$, and $f(u_S, v_i) = p_i^S\Pr_{S'\sim\mu}(S' = S)$.
	Since both the $s$-$t$ cuts with $s$ on one side and everything else on the other and the cut with $t$ on one side and everything else on the other both have cut capacity $1 - \prod_{i\in[n]}(1-\beta_i)$, it suffices to show that \eqref{eq:modifiedborderscriterion} holds only if there is a feasible flow of flow value equal to this cut capacity.

	Take any minimum-capacity $s$-$t$ cut $(A,B)$.
    For any $S$, if $v_i\in A$ for every $i\in S$, we can assume $u_S\in A$ since there are no edges coming out of $u_S$ except the $(u_S, v_i)$.
	Thus, the cut $(A,B)$ is completely characterized by which nodes $v_i\in B$ and which nodes $u_S\in A$ for $S$ containing some $i$ such that $v_i\in B$.
	Let $I = \{i\in[n]:v_i\in B\}$ and $\mathcal S = \{S\subseteq 2^{[n]}:S\cap I\neq\emptyset, u_S\in A\}$.
	The total capacity of this cut is
    \begin{equation*}
		\begin{split}
			c(A,B) &= \sum_{S\notin\mathcal S:S\cap I\neq\emptyset}c(s, u_S) + \sum_{S\in\mathcal S}\sum_{i\in S\cap I}c(u_S, v_i) + \sum_{i\notin I}c(v_i, t)\\
            &= \sum_{S:S\cap I\neq\emptyset}c(s, u_S) - \sum_{S\in\mathcal S}c(s, u_S)+ \sum_{S\in\mathcal S}\sum_{i\in S\cap I}c(u_S, v_i) + \sum_{i\notin I}c(v_i, t)\\
            &= \sum_{S:S\cap I\neq\emptyset}\Pr_{S'\sim\mu}(S'=S) - \sum_{S\in\mathcal S}\Pr_{S'\sim\mu}(S'=S) + \sum_{S\in\mathcal S}\sum_{i\in S\cap I}\bar p_i^S\Pr_{S'\sim\mu}(S'=S) + \sum_{i\notin I}p_i\Pr_{S'\sim\mu}(i\in S')\\
            &= \Pr_{S'\sim\mu}(S'\cap I\neq\emptyset) - \sum_{S\in\mathcal S}\Pr_{S'\sim\mu}(S'=S)\left(1-\sum_{i\in S\cap I}\bar p_i^S\right) + \sum_{i\notin I}p_i\beta_i\\
            &= 1 - \prod_{i\in I}(1-\beta_i) - \sum_{S\in\mathcal S}\left(\prod_{i\in S}\beta_i\right)\left(\prod_{i\notin\mathcal S}(1-\beta_i)\right)\left(1-\sum_{i\in S\cap I}\bar p_i^S\right)\\
            &\quad + 1 - \prod_{i\in[n]}(1-\beta_i) - \sum_{i\in I}p_i\beta_i
		\end{split}
	\end{equation*}
	using \eqref{eq:modifiedborderscriterionnowaste} for the last line.
	Rearranging, the above being least $1 - \prod_{i\in [n]}(1-\beta_i)$ is equivalent to
	\begin{equation*}
		\sum_{i\in I}p_i\beta_i + \prod_{i\in I}(1-\beta_i) + \sum_{S\in\mathcal S}\left(\prod_{i\in S}\beta_i\right)\left(\prod_{i\notin S}(1-\beta_i)\right)\left(1 - \sum_{i\in S\cap I}\bar p_i^S\right)\leq 1.
	\end{equation*}
	Thus, the flow network is feasible if and only if the above holds for all $I\subseteq[n]$ and $\mathcal S\subseteq\{S\subseteq [n]:S\cap I\neq\emptyset\}$. To maximize the left-hand side above, we should set $\mathcal S = \{S\subseteq[n]:S\cap I\neq\emptyset, \sum_{i\in S\cap I}\bar p_i^S\leq 1\}$, which gives the theorem statement.
\end{proof}

\section{Deferred Proofs}
\label{sec:appendix_proofs}

\subsection{Deferred Proof of \texorpdfstring{\cref{lem:centralallocationworstcase}}{Lemma 2.2}}
\label{ssec:appendix_central_allocation_worst_case_proof}

\begin{proof}[Proof of \cref{lem:centralallocationworstcase}]
    Under these Bernoulli value distributions, the probability that some agent has value $1$ for the item in a given round is $1-\prod_{j\in[n]}(1-\alpha_j)$, so under any allocation rule,
    \begin{equation}
        \label{eq:centrallyallocatedworstcasetotalutility}
        \frac1T\sum_{t=1}^T\sum_{i\in[n]}\mathbb E[U_i[t]] \leq 1-\prod_{j\in[n]}(1-\alpha_j).
    \end{equation}
    For these value distributions, agent $i$'s ideal utility is $\alpha_i$.
    If each agent obtains a $\lambda$ fraction of their ideal utility in expectation, then
    $
        \frac1T\sum_{t=1}^T\sum_{i\in[n]}\mathbb E[U_i[t]] \geq \lambda \sum_{i\in[n]}\alpha_i = \lambda,
    $ 
    so $\lambda \leq 1-\prod_{j\in[n]}(1-\alpha_j)$ by \eqref{eq:centrallyallocatedworstcasetotalutility}.
\end{proof}
\subsection{Deferred Proofs of \texorpdfstring{\cref{thm:worst_case_interim_allocation_probabilities,cor:proportional_budgets,cor:uniform_distribution}}{Theorem 3.5 and Corollaries B.6 and B.7}}
\label{ssec:interim_allocation_probability_feasibility_proofs}

\begin{proof}[Proof of \cref{thm:worst_case_interim_allocation_probabilities}]
	We prove that the choice of $p_i = 1 - \prod_{j=1}^n (1-\alpha_j)$ satisfies the conditions of \cref{thm:alpha_border}.
	Observe that \eqref{eq:introductory_border_new} holds with equality when $I=[n]$ for this choice of $p_i$.
        Also, observe that \eqref{eq:introductory_border_new} holds for $I=\emptyset$.
        To show that \eqref{eq:introductory_border_new} holds for $I\neq\emptyset$, define $a_I$ as
        \begin{equation*}
            a_I = \frac{1 - \prod_{i\in I}(1-\alpha_i)}{\sum_{i\in I}\alpha_i}.
        \end{equation*}
        With this choice of $p_i$, \eqref{eq:introductory_border_new} says
        \begin{equation*}
            \sum_{i\in I}p_i\alpha_i = \left(1 - \prod_{j=1}^n (1-\alpha_j)\right)\sum_{i\in I}\alpha_i \leq 1 - \prod_{i\in I}(1-\alpha_i),
        \end{equation*}
        which holds if and only if $a_I \geq 1 - \prod_{j=1}^n (1-\alpha_j)$. We shall show that $a_I$ is nonincreasing in $I$, which would suffice because $a_{[n]} = 1 - \prod_{j=1}^n (1-\alpha_j)$. Take any $I\subseteq[n]$ containing some $k\in[n]$ where $I\setminus\{k\}\neq\emptyset$. Compute
        \begin{equation*}
        \begin{split}
        	a_I - a_{I\setminus\{k\}} & = \frac{1 - \prod_{i\in I}(1-\alpha_i)}{\sum_{i\in I}\alpha_i} - \frac{1 - \prod_{i\in I\setminus\{k\}}(1-\alpha_i)}{\sum_{i\in I\setminus\{k\}}\alpha_i}\\
            & = \frac{\left(1 - \prod_{i\in I}(1-\alpha_i)\right)\left(\sum_{i\in I\setminus\{k\}}\alpha_i\right) - \left(1 - \prod_{i\in I\setminus\{k\}}(1-\alpha_i)\right)\left(\sum_{i\in I}\alpha_i\right)}{\left(\sum_{i\in I}\alpha_i\right)\left(\sum_{i\in I\setminus\{k\}}\alpha_i\right)}\\
        	& = \frac{-\alpha_k + \prod_{i\in I\setminus\{k\}}(1-\alpha_i)\left(\sum_{i\in I}\alpha_i - (1-\alpha_k)\sum_{i\in I\setminus\{k\}}\alpha_i\right)}{\left(\sum_{i\in I}\alpha_i\right)\left(\sum_{i\in I\setminus\{k\}}\alpha_i\right)}\\
        	& = \frac{-\alpha_k + \prod_{i\in I\setminus\{k\}}(1-\alpha_i)\alpha_k\left(1 + \sum_{i\in I\setminus\{k\}}\alpha_i\right)}{\left(\sum_{i\in I}\alpha_i\right)\left(\sum_{i\in I\setminus\{k\}}\alpha_i\right)}\\
        	& = \frac{\alpha_k\left(-1 + \prod_{i\in I\setminus\{k\}}(1-\alpha_i)\left(1+\sum_{i\in I\setminus\{k\}}\alpha_i\right)\right)}{\left(\sum_{i\in I}\alpha_i\right)\left(\sum_{i\in I\setminus\{k\}}\alpha_i\right)}.
        \end{split}
        \end{equation*}
This can be seen to be nonpositive by the identities that $1-x\leq e^{-x}$, so $\prod_{i\in I\setminus\{k\}}(1-\alpha_i) \leq \exp\left(-\sum_{i\in I\setminus\{k\}}\alpha_i\right)$, and also $e^{-y}(1+y)\leq 1$ applied to $y = \sum_{i\in I\setminus\{k\}}\alpha_i$.
\end{proof}

The proof of \cref{cor:proportional_budgets} is almost identical to proof of \cref{thm:worst_case_interim_allocation_probabilities}.
\begin{proof}[Proof of \cref{cor:proportional_budgets}]
    We prove that the choice of $p_i = \frac{1-\prod_{j=1}^n (1-\gamma \alpha_j)}{\gamma}$ satisfies the conditions of \cref{thm:beta_border}.
    Observe that \eqref{eq:borderscriterionnowaste} holds with this choice of $p_i$.

    Also, observe that \eqref{eq:borderscriterion} holds for $I=\emptyset$.
        To show that \eqref{eq:borderscriterion} holds for $I\neq\emptyset$, define $a_I$ as
        \begin{equation*}
            a_I = \frac{1 - \prod_{i\in I}(1-\beta_i)}{\sum_{i\in I}\beta_i}.
        \end{equation*}
        With this choice of $p_i$, \eqref{eq:borderscriterion} says
        \begin{equation*}
            \sum_{i\in I}p_i\beta_i = \frac{1 - \prod_{j=1}^n (1-\beta_j)}{\gamma}\sum_{i\in I}\beta_i \leq 1 - \prod_{i\in I}(1-\beta_i).
        \end{equation*}
        This holds if and only if $a_I \geq \frac{1 - \prod_{j=1}^n (1-\beta_j)}{\gamma}$. It suffices to show $a_I$ is nonincreasing in $I$ since $a_{[n]} = \frac{1 - \prod_{j=1}^n (1-\beta_j)}{\gamma}$. To do this, we can simply follow the same proof that $a_I$ as defined in the proof of \cref{thm:worst_case_interim_allocation_probabilities} is nonincreasing, but replace each $\alpha_i$ with $\beta_i$. We did not use the fact that $\sum_{i=1}^n\alpha_i=1$ for that part of the proof of \cref{thm:worst_case_interim_allocation_probabilities}, so everything still works.
\end{proof}

The proof of \cref{cor:uniform_distribution} is an application of \cref{cor:proportional_budgets}.
\begin{proof}[Proof of \cref{cor:uniform_distribution}]
With a $\mathrm{Uniform}([0,1])$ value distribution, the $\beta$-ideal utility is
\begin{equation*}
    v^\star(\beta) = \E_{V\sim\mathrm{Uniform}([0,1])}[V\pmb1\{V > 1-\beta\}] = \frac12\beta(2-\beta).
\end{equation*}
If $\beta = \gamma\alpha$, then
\begin{equation*}
    v^\star(\beta) = \frac{\gamma(1-\gamma\alpha)}{2-\alpha}v^\star(\alpha).
\end{equation*}
Using \cref{cor:proportional_budgets}, we can find interim allocation probabilities $p_i$ such that,
\begin{equation*}
    p_iv^\star(\beta) \geq \frac{1-e^{-\gamma}}{\gamma}v^\star(\beta) \geq \frac{(1-e^{-\gamma})(2-\gamma\alpha)}{2-\alpha}.
\end{equation*}
The result follows from substituting $\alpha = \frac1n$ and $\gamma = \Theta(\log n)$.
\end{proof}

\subsection{Deferred Proofs of \texorpdfstring{\cref{lem:knapsackproblemovertime}, \cref{prop:approximate_nash,prop:beta_approximate_nash}, \cref{thm:nash,thm:general_nash}}{Lemma 4.1, Propositions B.2 and 3.3, Theorems B.4 and 3.6}}
\label{ssec:appendix_equilibrium_proofs}
We shall prove these in the context of our more general mechanism \generalmechanism, which is the same as \mechanism except the each agent gets $\beta_i(1+\delta_i^T)$ budget of bid tokens where $\beta_i\in [0,1]$ is a parameter. \mechanism is \generalmechanism when each $\beta_i = \alpha_i$. See \cref{sec:beta_mechanism} for details.

To prove results about agents' behavior in the mechanism, it will be useful to use the following imaginary game.
This imaginary game will be the same game, but we do not enforce budgets and allow agents to bid regardless of whether they have budget remaining.
Let $\tilde b_i^t$, $\tilde U_i[t]$, and $\tilde i^t$ be the bids, utilities, and winners in this imaginary game, respectively.
We couple the imaginary game and the actual game such that the agents have the same values $V_i[t]$, and $\tilde b_i^t = b_i^t$, $\tilde U_i[t] = U_i[t]$, and $\tilde i^t = i^t$ at all times $t$ in which all agents have budget remaining.

We first prove a lemma about the best strategy for an agent $i$ assuming they 1) do not want to exceed their budget in expectation and 2) the other agents $j\neq i$ are acting in a specific way.
\begin{lemma}
	\label{lem:generalized_knapsack_problem_over_time}
	Fix an agent $i$ and assume all other agents $j\neq i$ are playing in a way such that the sets $S_{\neq i}[t] = \{j\neq i: b_j^t=1\}$ of bidding agents $j\neq i$ are i.i.d.
	across rounds drawn from some distribution $\nu$ over subsets of $[n]\setminus\{i\}$.
	Suppose agent $i$ is trying to maximize her imaginary expected utility subject to the constraint that she does not exceed her budget in expectation; that is, she is choosing $(\tilde b_i^t)$ to solve
	\begin{equation}
		\label{eq:knapsackproblemovertime}
		\max\,\frac1T\sum_{t=1}^T\mathbb E[\tilde U_i[t]] \quad\text{subject to}\quad \sum_{t=1}^T \mathbb E[\tilde b_i^t]\leq \beta_i'T
	\end{equation}
	where $\beta_i' = \beta_i(1+\delta_i^T)$.
	Her optimal strategy is to choose $(\tilde b_i^t)$ to be a $\beta_i'$-aggressive strategy.
	Letting
	\begin{equation}
        \label{eq:generalized_interim_probability}
		p_i(\nu) = \sum_{S\subseteq[n]}p_i^S\Pr_{S_{\neq i}\sim\nu}(S_{\neq i}\cup\{i\} = S)
	\end{equation}
	be the probability that agent $i$ wins a round conditioned on bidding, the $\beta_i'$-aggressive strategy yields agent $i$ utility
	\begin{equation*}
		\frac1T\sum_{t=1}^T \mathbb E[\tilde U_i[t]] = p_i(\nu)v^\star(\beta_i').
	\end{equation*}
\end{lemma}
Notice that \cref{lem:knapsackproblemovertime} is a special case of the above lemma with parameters $\beta_i = \alpha_i$, $\delta_i^T = 0$, and $\nu$ being the distribution over subsets $S_{\neq i}$ where each $j\in S_{\neq i}$ independently with probability $\alpha_j$. In the proof, we use the notion of $\beta$-ideal utility and $\beta$-ideal utility probability function from \cref{def:ideal_utility} defined in \cref{sec:beta_mechanism}.
\begin{proof}[Proof of \cref{lem:generalized_knapsack_problem_over_time}]
	For any strategy $(\tilde b_i^t)$, at any time $t$,
	\begin{equation}
		\label{eq:expectedutilityarbitrarystrategy}
		\begin{split}
			\mathbb E[\tilde U_i[t]] = \mathbb E[V_i[t]\pmb1\{i=\tilde i^t\}] = \mathbb E[V_i[t]\tilde b_i^t\pmb1\{i=\tilde i^t\}] = \mathbb E[\mathbb E[V_i[t]\tilde b_i^t\pmb1\{i=\tilde i^t\}\mid \tilde b_i^t]] \\
			= \mathbb E[V_i[t]\tilde b_i^t]\mathbb E[\pmb1\{i=\tilde i^t\}\mid \tilde b_i^t=1] = p_i(\nu)\mathbb E[V_i[t]\tilde b_i^t].
		\end{split}
	\end{equation}
	Thus, agent $i$'s maximization problem is equivalent to maximizing $\frac1T\sum_{t=1}^T \mathbb E[V_i[t]\tilde b_i^t]$ subject to the same constraint in \eqref{eq:knapsackproblemovertime}.
	It is clear for any feasible solution $(\tilde b_i^t)$ that the solution that bids at time $t$ with probability $\rho[t](V_i[t]) = \Pr(\tilde b_i^t=1\mid V_i[t])$ is also a feasible solution with the same objective value.
	Thus, we can rewrite the agent's maximization problem in terms of maximizing over $\rho[t](V_i[t])$, i.e., the following maximization problem over measurable functions $\rho[t]:[0,\infty)\to[0,1]$:
	\begin{equation*}
		\frac1T\max\,\sum_{t=1}^T \mathbb E[V_i[t]\rho[t](V_i[t])]\quad\text{subject to}\quad \sum_{t=1}^T \mathbb E[\rho[t](V_i[t])]\leq \beta_i' T.
	\end{equation*}
	Given any optimal solution $(\rho[t])$ to the above, observe that setting $\rho^\star[t] = \frac1T\sum_{s=1}^T \rho[s]$ is also a feasible solution with the same objective value.
	Observe then that $\mathbb E[\rho^\star[t](V_i[t])] = \beta_i'$, and so $\rho^\star[t]$ is a feasible solution in \eqref{eq:idealutility}, the definition of $\beta_i'$-ideal utility, and it must maximize the same objective $\mathbb E[V_i[t]\rho[t](V_i[t])]]$.
	Therefore, $\rho^\star[t]$ is exactly the $\beta_i'$-ideal utility probability function $(\rho_i^{\beta_i'})^\star$, so the optimal bidding strategy $(\tilde b_i^t)^\star$ to solve \eqref{eq:knapsackproblemovertime} is precisely a $\beta_i'$-aggressive strategy.

	By \eqref{eq:expectedutilityarbitrarystrategy}, under such a $\beta_i'$-aggressive strategy, agent $i$ obtains utility $p_i(\nu)v^\star(\beta_i')$.
\end{proof}

To start relating the imaginary game to the actual game, we use Chernoff bounds to show that agents $i$ who use $\beta_i$-aggressive strategies will not run out of budget with high probability.
\begin{lemma}
	\label{lem:dontrunoutofmoney}
	If agent $i$ uses a $\beta_i$-aggressive strategy, the probability that they run out of budget is at most $O\left(\frac{1}{T^2}\right)$.
\end{lemma}
\begin{proof}
	By the Chernoff bound,
	\begin{equation*}
		\Pr\left(\sum_{t=1}^Tb_i^t \geq \beta_i(1+\delta_i^T)T\right) \leq \exp\left(-\frac{(\delta_i^T)^2\beta_iT}{2+\delta_i^T}\right).
	\end{equation*}
	The result follows from substituting $\delta_i^T = \sqrt{\frac{6\ln T}{\beta_i T}}$.
\end{proof}

We shall use the following lemma to obtain high probability bounds on the agents' utilities in the actual game. Remember that we coupled in the imaginary game and the actual game such that the agents have the same values $V_i[t]$, and $\tilde b_i^t = b_i^t$, $\tilde U_i[t] = U_i[t]$, and $\tilde i^t = i^t$ at all times $t$ in which all agents have budget remaining. This implies that the strategy used by a player $i$ in the imaginary game directly translates to a strategy used by $i$ in the actual game. We define a $\beta$-aggressive strategy in the imaginary game to be one in which agent $i$ bids when her value is in the top $\beta$-quantile of her value distribution. Notice that if player $i$ is playing a $\beta$-aggressive strategy in the imaginary game, then she is playing a $\beta$-aggressive strategy in the actual game (since the only difference is the budget constraint).
\begin{lemma}
\label{lem:high_probability_utility}
Fix an agent $i$ playing a $\beta_i$-aggressive strategy. Suppose the other agents $j\neq i$ are playing in the imaginary game with strategies as in \cref{lem:generalized_knapsack_problem_over_time}. Let $E_1$ be an event on which agents $j\neq i$ do not run out of budget in the actual game. Then, on an subevent of $E_1$ with probability at least $\Pr(E_1) - O(1/T^2)$,
\begin{equation*}
    \left|\frac1T\sum_{t=1}^T \tilde U_i[t]- p_i(\nu)v_i^\star(\beta_i)\right| \leq O\left(\sqrt{\frac{\log T}{T}}\right).
\end{equation*}
\end{lemma}
\begin{proof}
Observe that the random variables $\tilde U_i[t] = V_i[t]\pmb1\{\tilde i^t = i\}$ are i.i.d. By \eqref{eq:expectedutilityarbitrarystrategy}, each has mean $p_i(\nu)\mathbb E[V_i[t]\tilde b_i^t]$. Since player $i$ is playing a $\beta_i$-aggressive strategy, this mean is $\mathbb E[\tilde U_i[t]] = p_i(\nu)v_i^\star(\beta_i)$. Recall that we assume the value distribution $\mathcal F_i$ is bounded so the $U_i$ are bounded by some $\bar v$. Let $\epsilon > 0$. By Hoeffding's inequality,
\begin{equation}
\label{eq:hoeffding_utility}
    \Pr\left(\left|\sum_{t=1}^T \tilde U_i[t] - p_i(\nu)v_i^\star(\beta_i) T\right| \geq \epsilon \right)\leq 2\exp\left(-\frac{2\epsilon^2}{\bar v^2T}\right)
\end{equation}
Let $E_2$ be the event that the above event does not occur. Let $E_3$ be the event that agent $i$ does not run out of budget. Let $E = E_1\cap E_2\cap E_3$. On $E$,
\begin{equation*}
    \left|\frac1T\sum_{t=1}^T U_i[t] - p_i(\nu)v_i^\star(\beta_i)\right| = \left|\frac1T\sum_{t=1}^T \tilde U_i[t]- p_i(\nu)v_i^\star(\beta_i)\right| \leq \frac{\epsilon}{T}.
\end{equation*}
Substituting $\epsilon = \bar v\sqrt{T\ln T}$, the above is at most $O\left(\sqrt{\frac{\log T}{T}}\right)$, and by also substituting this $\epsilon$ into \eqref{eq:hoeffding_utility} and by \cref{lem:dontrunoutofmoney}, we have $\Pr(E) \geq \Pr(E_1) - O(1/T^2)$, giving the result.
\end{proof}

The following lemma gives some form of continuity in the ideal utility that we need to bound the utility an agent can obtain from deviating from the proposed equilibrium.
\begin{lemma}
	\label{lem:idealutilitylipschitz}
	Let $\beta_i' = (1+\delta_i^T)\beta_i$.
	Then,
	\begin{equation*}
		v_i^\star(\beta') - v_i^\star(\beta) \leq \delta_i^Tv_i^\star(\beta_i).
	\end{equation*}
\end{lemma}
\begin{proof}
	It was proven in \cite{fikioris2023online} that $\beta\mapsto v_i^\star(\beta)$ is concave.
	The lemma statement follows from concavity and the fact that $v_i^\star(0)=0$.
\end{proof}

With those lemmas, we can now prove that each player following a $\beta_i$-aggressive strategy and give their utility guarantee. We use $p_i$ to denote player $i$'s interim allocation probability, which we define in \cref{def:beta_interim_allocation_probability}. Equivalently, $p_i$ denotes $p_i(\nu)$ as defined in \eqref{eq:generalized_interim_probability} when $\nu$ is the distribution over $S_{\neq i}$ where $j\in S_{\neq i}$ independently with probability $\beta_j$. The below theorem is \cref{prop:beta_approximate_nash,thm:general_nash} combined. It is a direct generalization of \cref{prop:approximate_nash}, which can be obtained by setting $\beta_i = \alpha_i$ for each $i$. When using $\beta_i=\alpha_i$ and \cref{lem:specific_border_modification,lem:key_lemma} to set the allocation probabilities, we also obtain \cref{thm:nash}.
\begin{theorem}
\label{thm:appendix_nash}
Suppose we run \generalmechanism with slack parameters $\delta_i^T = \sqrt{\frac{6\ln T}{\beta_i T}}$. Each player $i$ playing a $\beta_i$-aggressive is an $O\left(\sqrt{\frac{\log T}{T}}\right)$-approximate Nash equilibrium. At this approximate equilibrium, with probability at least $1 - O(1/T^2)$, player $i$ gets utility
\begin{equation*}
    \frac1T\sum_{t=1}^T U_i[t] \geq p_iv^\star(\beta_i) - O\left(\sqrt{\frac{\log T}{T}}\right).
\end{equation*}
\end{theorem}
\begin{proof}
	Suppose every agent $i$ is using a $\beta_i$-aggressive strategy.
        Let $E_1$ be the event that no agent runs out of budget.
        By \cref{lem:dontrunoutofmoney}, $\Pr(E_1) \geq 1 - O(1/T^2)$.
        Using \cref{lem:high_probability_utility}, there is an event $E$ of probability at least $1 - O(1/T^2)$ on which
	\begin{equation*}
	   \frac1T\sum_{t=1}^T \tilde U_i[t] \geq p_iv^\star(\beta_i) - O\left(\sqrt{\frac{\log T}{T}}\right).
	\end{equation*}
        This establishes the high probability utility guarantee if every agent $i$ is playing a $\beta_i$-aggressive strategy.
        Now let us show this strategy profile is indeed a Nash equilibrium.
        The high probability utility guarantee translates to a utility guarantee in expectation in that
	\begin{equation}
		\label{eq:originalgameutilityguarantee}
		\frac1T\sum_{t=1}^T \mathbb E[U_i[t]] \geq p_iv^\star(\beta_i) - O\left(\sqrt{\frac{\log T}{T}}\right).
	\end{equation}

	Now we upper bound the utility of any deviating strategy by player $i$, still assuming players $j\neq i$ are following a $\beta_j$-aggressive strategy.
	Any strategy $(b_i^t)$ used by player $i$ in the actual game satisfies the budget constraint $\sum_{t=1}^T b_i^T \leq \beta_i(1+\delta_i^T)T$ almost surely.
	In particular, it satisfies the budget constraint in expectation, so we can use \cref{lem:knapsackproblemovertime} to conclude that under the strategy $(b_i^t)$ in the imaginary game,
	\begin{equation*}
		\frac1T\sum_{t=1}^T \mathbb E[\tilde U_i[t]] \leq p_iv^\star(\beta_i(1+\delta_i^T)).
	\end{equation*}
	Therefore,
	\begin{equation}
		\begin{split}
			\label{eq:originalgameutilityupperbound}
			\frac1T\sum_{t=1}^T \mathbb E[U_i[t]] & = \frac1T\sum_{t=1}^T \mathbb E[U_i[t]\pmb1_{E_1}] + \frac1T\sum_{t=1}^T \mathbb E[U_i[t]\pmb1_{E_1^c}]\\
            & \leq \frac1T\sum_{t=1}^T \mathbb E[\tilde U_i[t]] + \frac1T\sum_{t=1}^T \mathbb E[V_i[t]\pmb1_{E_1^c}]\\
            &\leq p_iv_i^\star(\beta_i(1+\delta_i^T)) + O\left(\frac{1}{T^2}\right).
		\end{split}
	\end{equation}
	By \eqref{eq:originalgameutilityguarantee} and \eqref{eq:originalgameutilityupperbound}, by deviating from a $\beta_i$-aggressive strategy, player $i$ can only gain an additive utility difference of
	\begin{equation*}
		p_iv_i^\star(\beta_i(1+\delta_i^T))T - p_iv^\star(\beta_i)T + O\left(\sqrt{\frac{\log T}{T}}\right).
	\end{equation*}
	By substituting $\delta_i^T = \sqrt{\frac{6\ln T}{\beta_i T}}$ and using \cref{lem:idealutilitylipschitz}, this implies that this additive difference is at most $O\left(\sqrt{\frac{\log T}{T}}\right)$, thus proving the theorem.

\end{proof}

\subsection{Deferred Proofs from \texorpdfstring{\cref{sec:anytime}}{section C}}
\label{ssec:appendix_anytime_proofs}
In all the proofs in this subsection, we use the following notation.
Let $\tilde b_i^t$ be the bids of agent $i$ corresponding to a $\beta_i$-aggressive strategy in the imaginary game introduced in \cref{ssec:appendix_equilibrium_proofs} where there is no budget constraint that agree with the actual bids $b_i^t$ at times $t$ where agent $i$ has budget.
Let $U_i[t,T]$ denote the utility of player $i$ gained at time $t$ under the policy $\tilde b_i^t$ if only the $n$ constraints $\sum_{t=1}^T b_k^t \leq \beta_k(1+\delta_k^T)T$ were enforced.
If agent $i$ deviates to a policy $(b_i^t)$, we let $U_i'[t]$ and $U_i'[t,T]$ be analogous to $U_i[t]$ and $U_i[t,T]$ for the deviating policy.
\begin{proof}[Proof of \cref{thm:anytime_equilibrium}]

By \cref{lem:dontrunoutofmoney}, the event $E_1^s$ that no agent runs out of budget at time $s$ has probability at least $1 - O(1/s^2)$. By the union bound, the event $E_1 = \bigcap_{s=\lceil\sqrt{t}\rceil}^t E_1^s$ that no agent runs out of budget at any time $\lceil\sqrt{t}\rceil$ and $t$ has probability at least $1 - \sum_{s=\lceil\sqrt t\rceil}^t (1-\Pr(E_1^s)) = 1-O(1/\sqrt{t})$. By \cref{lem:high_probability_utility}, on an subevent $E$ of $E_1$ with probability at least $1- O(1/\sqrt t)$,
\begin{equation*}
    \frac1t\sum_{s=1}^t U_i[s, t] = \frac1t\sum_{s=1}^t U_i[s,t] \geq p_iv_i^\star(\beta_i) - O\left(\sqrt{\frac{\log t}{t}}\right).
\end{equation*}
Since on $E_1$, we have $U_i[s,t] = U_i[s]$ for each $s$ between $\lceil \sqrt t\rceil$ and $t$, we obtain on $E$,
\begin{equation*}
    \frac1t\sum_{s=1}^t U_i[s] \geq \frac1t \sum_{s=\lceil \sqrt t\rceil}^t U_i[s,t] \geq  \frac1t \sum_{s=1}^t U_i[s,t] - \frac1t\sum_{s=1}^{\lceil \sqrt t\rceil -1}V_i[s] \geq p_iv_i^\star(\beta_i) - O\left(\sqrt{\frac{\log t}{t}}\right),
\end{equation*}
establishing the utility guarantee.

Suppose agent $i$ deviates to a policy $(b_i^t)'$. By \cref{thm:appendix_nash},
\begin{equation*}
\begin{split}
    \frac1t \sum_{s=1}^t \mathbb E[U_i'[s]] & = \frac1t\sum_{s=1}^t \mathbb E[U_i'[s]\pmb1_{E_1}] + \frac1t\sum_{s=1}^t \mathbb E[U_i'[s]\pmb1_{E_1^c}]\\
    & = \frac1t\sum_{s=1}^{\lceil \sqrt t\rceil}\mathbb E[U_i'[s]]  + \frac1t\sum_{s=\lceil\sqrt t\rceil}^t\mathbb E[U_i'[s, t]] + \frac1t\sum_{s=1}^t \mathbb E[U_i'[s]\pmb1_{E_1^c}]\\
    & \leq \frac1t\sum_{s=1}^{\lceil \sqrt t\rceil-1}\mathbb E[U_i'[s]]  + \frac1t\sum_{s=1}^t\mathbb E[U_i[s, t]]  + \frac1t\sum_{s=1}^t \mathbb E[U_i'[s]\pmb1_{E_1^c}]\\
    & = \frac1t\sum_{s=1}^{\lceil \sqrt t\rceil-1}\mathbb E[U_i'[s]]  + \frac1t\sum_{s=1}^{\lceil \sqrt t\rceil - 1}\mathbb E[U_i[s,t]\pmb1_{E_1}] + \frac1t\sum_{s=\lceil \sqrt t\rceil}^t\mathbb E[U_i[s,t]\pmb1_{E_1}]\\
    & \quad + \frac1t\sum_{s=1}^t \mathbb E[U_i[s,t]\pmb1_{E_1^c}] + \frac1t\sum_{s=1}^t \mathbb E[U_i'[s]\pmb1_{E_1^c}]\\
    & = \frac1t\sum_{s=1}^{\lceil \sqrt t\rceil-1}\mathbb E[U_i'[s]]  + \frac1t\sum_{s=1}^{\lceil \sqrt t\rceil - 1}\mathbb E[U_i[s,t]\pmb1_{E_1}] + \frac1t\sum_{s=\lceil \sqrt t\rceil}^t\mathbb E[U_i[s]\pmb1_{E_1}]\\
    & \quad + \frac1t\sum_{s=1}^t \mathbb E[U_i[s,t]\pmb1_{E_1^c}] + \frac1t\sum_{s=1}^t \mathbb E[U_i'[s]\pmb1_{E_1^c}]\\
    & \leq \frac1t\sum_{s=1}^t \mathbb E[U_i[s]] + O\left(\sqrt{\frac{\log t}{t}}\right).
\end{split}
\end{equation*}
By substituting $t=T$, we see that everyone playing a $\beta_i$-aggressive strategy is indeed an $O\left(\sqrt{\frac{\log T}{T}}\right)$-equilibrium.
\end{proof}

\begin{proof}[Proof of \cref{thm:anytime_robustness}]
Assume without loss of generality that the agents $j\neq i$ never bid when they're out of budget. As in proof of \cref{thm:anytime_equilibrium}, the event $E_1$ that agent $i$ does not run out of budget at any time between $\lceil \sqrt{t}\rceil $ and $t$ has probability at least $1 - O(1/\sqrt{t})$. Then, on $E_1$, no one runs out of budget between time $\lceil \sqrt t\rceil$ and $t$. Using \cref{thm:nash}, there is an subevent $E$ of $E_1$ of probability at least $1 - O(1/\sqrt t)$ on which
\begin{equation*}
\begin{split}
    \frac1t\sum_{s=1}^t U_i[s] & \geq \frac1t\sum_{s=\lceil \sqrt t\rceil}^tU_i[s,t]\\
    & \geq \left(\frac12 + \frac12\alpha_i^2\right)v_i^\star - O\left(\sqrt{\frac{\log t}{t}}\right) - \frac1t\sum_{s=1}^{\lceil \sqrt t\rceil -1}V_i[s]\\
    & \geq \left(\frac12 + \frac12\alpha_i^2\right)v_i^\star - O\left(\sqrt{\frac{\log t}{t}}\right).
\end{split}
\end{equation*}
\end{proof}

\begin{proof}[Proof of \cref{thm:exactnash}]
	By \cref{lem:dontrunoutofmoney} and the Borel-Cantelli Lemma, there exists a random time $t_0$ such that $\sum_{s=1}^t \tilde b_i^s \leq \beta_i(1+\delta_i^t)t$ for all $t > t_0$ and all agents $i$ where $t_0 < \infty$ almost surely.
	Suppose agent $i$ deviates to a policy $(b_i^s)'$.
	Using \cref{thm:appendix_nash},
	\begin{equation*}
		\begin{split}
			\frac1t\sum_{s=1}^t \mathbb E[U_i[s]] & = \mathbb E\left[\frac1t\sum_{s=1}^{t_0}
			U_i[s]\right] + \mathbb E\left[\frac1t\sum_{s=t_0+1}^t U_i'[t, T]\right]         \\ & \leq \mathbb E\left[\frac1t\sum_{s=1}^{t_0}U_i[s]\right] + \mathbb E\left[\frac1t\sum_{s=1}^t U_i'[t, T]\right]\\ & \leq \mathbb E\left[\frac1t\sum_{s=1}^{t_0}U_i[s]\right] + \mathbb E\left[\frac1t\sum_{s=1}^t U_i[s,t]\right] + O\left(\sqrt{\frac{\log t}{t}}\right)\\ & = \mathbb E\left[\frac1t\sum_{s=1}^{t_0}U_i[s]\right] + \mathbb E\left[\frac1t\sum_{s=1}^{t_0}U_i[s,t]\right] + \mathbb E\left[\frac1t\sum_{s=t_0+1}^tU_i[s]\right]+ O\left(\sqrt{\frac{\log t}{t}}\right)\\ & \leq 2\mathbb E\left[\frac1t\sum_{s=1}^{t_0}V_i[s]\right] + \frac1t\sum_{s=1}^t\mathbb E[U_i[s]]+ O\left(\sqrt{\frac{\log t}{t}}\right).
		\end{split}
	\end{equation*}
	Clearly, $\frac1t\sum_{s=1}^{t_0}V_i[s] \overset{\mathrm{a.s.
			}}\to 0$, so by uniform integrability, $\mathbb E\left[\frac1t\sum_{s=1}^{t_0}
		V_i[s]\right]\to 0$.
	Therefore,
	\begin{equation*}
		\liminf_{t\to\infty}\frac1t\sum_{s=1}^t\mathbb E[U_i[s]] \leq \liminf_{t\to\infty}\frac1t\sum_{s=1}^t \mathbb E[U_i[s]],
	\end{equation*}
	proving the Nash equilibrium claim.

	For the utility claim, let $\tilde U_i[s]$ denote the utility of agent $i$ in the imaginary game where no budget constraints are enforced and all agents $j$ are following a $\beta_j$-aggressive strategy.
	Observe that the $\tilde U_i[s]$ are i.i.d.
	$\mathrm{Bernoulli}(p_iv^\star(\beta_i))$.
	By the strong law of large numbers, $\frac1t\sum_{s=1}^t \tilde U_i[s]\overset{\mathrm{a.s.}}\to p_iv^\star(\beta_i)$.
	The utilities in the actual game satisfy
	\begin{equation*}
		\frac1t\sum_{s=1}^t U_i[s] = \frac1t\sum_{s=1}^{t_0}
		U_i[s] + \frac1t\sum_{s=t_0+1}^t U_i[s] = \frac1t\sum_{s=1}^{t_0} U_i[s] + \frac1t\sum_{s=t_0+1}^t \tilde U_i[s], \end{equation*} which has the same limit as $t\to\infty$.
\end{proof}

\subsection{Deferred Proof of \texorpdfstring{\cref{lem:bangforbuck}}{Lemma 5.1}}
\label{ssec:appendix_bang_for_buck_proof}
First, we note that the worst-case value distribution for robustness is a $\mathrm{Bernoulli}(\alpha_i)$ value distribution.
\begin{lemma}
	\label{lem:bernoulliworstcasevaluedistribution}
    Assume player $i$ has a policy $\hat\pi_i$ such that if they had a $\hat{\mathcal F_i} =\mathrm{Bernoulli}(\alpha_i)$ value distribution, regardless of the behavior of other agents $j\neq i$ they would obtain utility
    \begin{equation*}
        \frac1T\sum_{t=1}^T \hat U_i[t] \geq \lambda_i \hat v_i^\star
    \end{equation*}
    with probability at least $1 - O(1/T^2)$ where $\hat v_i^\star = \alpha_i$ is the ideal utility of agent $i$ had they a $\mathrm{Bernoulli}(\alpha_i)$ value distribution.
	Then, if, instead, player $i$ had an arbitrary value distribution $\mathcal F_i$, we can construct a policy $\pi_i$ such that regardless of the behavior of other agents $j\neq i$, player $i$ would obtain utility
    \begin{equation*}
        \frac1T\sum_{t=1}^T U_i[t] \geq \lambda_iv_i^\star - O\left(\sqrt{\frac{\log T}{T}}\right)
    \end{equation*}
    with probability at least $1 - O(1/T^2)$.
\end{lemma}
\begin{proof}
	Suppose agent $i$ has arbitrary value distribution $\mathcal F_i$.
	Construct the policy $\pi_i$ as follows.
	At each time $t$, agent $i$ will sample $\hat V_i[t]\sim\mathrm{Bernoulli}((\rho_i^{\alpha_i})^\star(V_i[t]))$, where $(\rho_i^{\alpha_i})^\star$ is agent $i$'s $\alpha_i$-ideal utility probability function with value distribution $\mathcal F_i$. (We define $\alpha_i$-ideal utility probability function in \cref{def:ideal_utility}; informally, it is $1$ if $V_i[t]$ is in the top $\alpha_i$-quantile of $\mathcal F_i$ and $0$ otherwise.)
	Then, agent $i$ will bid if and only if $\hat\pi_i$ would bid with the Bernoulli value $\hat V_i[t]$.
	In other words, the policy $\pi_i$ is simply following the policy $\hat \pi_i$ but with the Bernoulli values $\hat V_i[t]$ instead of the actual values.

	Notice that the $\hat V_i[t]$ are indeed i.i.d.
	$\mathrm{Bernoulli}(\alpha_i)$.
        By the hypothesis of the lemma,
        \begin{equation*}
            \frac1T\sum_{t=1}^T \hat V_i[t]\pmb1\{i^t = i\} \geq \lambda_i\alpha_i
        \end{equation*}
        on an event $E_1$ of probability at least $1 - O(1/T^2)$.

        Let $\mathcal H_t$ denote the history up to time $t$. Let $\mathcal G_t$ be the $\sigma$-algebra generated by $\mathcal H_t$, $\hat V_i[t+1]$, and $i^{t+1}$. Define the $\mathcal G_t$-adapted process
        \begin{equation*}
            M[t] = \sum_{s=1}^t V_i[s]\pmb1\{i^s = s\} - \frac{v_i^\star}{\alpha_i}\sum_{s=1}^t\hat V_i[s]\pmb1\{i^s = s\}
        \end{equation*}
	Observe that when agent $i$ uses the policy $\hat\pi_i$, everything in the mechanism is independent of the actual values $V_i[t]$ conditioned on the Bernoulli values $\hat V_i[t]$.
	Using this fact,
	\begin{equation*}
		\label{eq:actualvalueversusbernoulli}
		\begin{split}
			\mathbb E[V_i[t]\pmb1\{i^t=i\}\mid \mathcal G_{t-1}] & = \mathbb E[V_i[t]\mid\hat V_i[t]]\pmb1\{i^t=i\}                                          \\
			                                & = \frac{\mathbb E[V_i[t]\pmb1\{\hat V_i[t]=1\}]}{\Pr(\hat V_i[t]=1)}\hat V_i[t]\pmb1\{i^t=i\} \\
			                                & = \frac{\mathbb E[V_i[t](\rho_i^{\alpha_i})^\star(V_i[t])]}{\Pr(\hat V_i[t])=1)}\hat V_i[t]\pmb1\{i^t=i\}                         \\
			                                & = \frac{v_i^\star}{\alpha_i}\cdot\hat V_i[t]\pmb1\{i^t=i\}.
		\end{split}
	\end{equation*}
        Therefore, $M[t]$ is a $\mathcal G_t$-martingale. Let $\bar v$ be an upper bound on the distribution $\mathcal F_i$ (recall we assumed value distributions are bounded). By the Azuma-Hoeffding inequality, for any $\epsilon > 0$,
        \begin{equation*}
            \Pr(M[T] \leq -\epsilon) \leq \exp\left(-\frac{2\epsilon^2}{\bar v^2T}\right).
        \end{equation*}
        Set $\epsilon = \bar v\sqrt{T\ln T}$, so that the above is at most $O(1/T^2)$. Let $E_2$ be the event that the above does not occur and let $E = E_1\cap E_2$. We have $\Pr(E) \geq 1 - O(1/T^2)$. On $E$,
        \begin{equation*}
            \frac1T\sum_{t=1}^T V_i[t]\pmb1\{i^t = i\} = \frac1T M[T] + \frac{v_i^\star}{\alpha_i T}\sum_{t=1}^T \hat V_i[t]\pmb1\{i^t = i\} \geq \lambda_iv_i^\star - O\left(\sqrt{\frac{\log T}{T}}\right),
        \end{equation*}
        thereby proving the lemma.
\end{proof}

Now we can prove \cref{lem:bangforbuck}.
\begin{proof}[Proof of \cref{lem:bangforbuck}]
	Let $\mathcal H_t$ denote the history up to time $t$.
	Let $\mathcal G_t$ be the $\sigma$-algebra generated by $\mathcal H_t$ and $b_j^{t+1}$ for $j\neq i$, i.e., make the bids of agents $j\neq i$ predictable processes with respect to $\mathcal G_t$.
	Define the $\mathcal G_t$-adapted process
	\begin{equation*}
		M_i^t = \sum_{s=1}^t \left(b_i^s\pmb1\{i^t\neq i\} - \alpha_i\bar p\sum_{j\neq i}b_j^s\right).
	\end{equation*}
	If $\sum_{j\neq i}b_j^s=0$, then
	\begin{equation*}
		\mathbb E\left[b_i^s\pmb1\{i^t\neq i\}\,\middle|\,\mathcal G_{s-1}\right] = 0.
	\end{equation*}
	If $\sum_{j\neq i}b_j^s = 1 = b_j^s$, then
	\begin{equation*}
		\mathbb E\left[b_i^s\pmb1\{i^t\neq i\}\,\middle|\,\mathcal G_{s-1}\right] \leq \alpha_ip_j^{\{i,j\}} \leq \alpha_i\bar p.
	\end{equation*}
	If $\sum_{j\neq i}b_j^s \geq 2$, then
	\begin{equation*}
		\mathbb E\left[b_i^s\pmb1\{i^t\neq i\}\,\middle|\,\mathcal G_{s-1}\right] \leq \alpha_i.
	\end{equation*}
	In any case, since $\bar p\geq \frac12$,
	\begin{equation*}
		\mathbb E\left[b_i^s\pmb1\{i^t\neq i\}\,\middle|\,\mathcal G_{s-1}\right] \leq \alpha_i\bar p\sum_{j\neq i}b_j^s.
	\end{equation*}
	Therefore, $M_i^t$ is a $\mathcal G_t$-supermartingale.
	Letting $\delta^T = \max_{j\neq i}\delta_j^T$, observe that
	\begin{equation*}
		\sum_{s=1}^T \sum_{j\neq i}b_j^s \leq (1-\alpha_i)(1+\delta^T)T \leq (1-\alpha_i)T + O(\sqrt{T\log T}).
	\end{equation*}
	So,
	\begin{equation*}
		\begin{split}
			\sum_{s=1}^T b_i^s\pmb1\{i^t\neq i\} & \leq \sum_{s=1}^T \mathbb E[b_i^s\pmb1\{i^t\neq i\}\mid\mathcal G_{s-1}] + M_i^T \\
			                                     & \leq \alpha_i\bar p\sum_{s=1}^T\sum_{j\neq i}b_j^s + M_i^T                       \\
			                                     & \leq \alpha_i\bar p(1-\alpha_i)T + O(\sqrt{T\log T}) + M_i^T.
		\end{split}
	\end{equation*}
	This will be at most $\alpha_i\bar p(1-\alpha_i)T + O(\sqrt{T\log T})$ on an event $E_1$ of probability at least $1-O\left(\frac{1}{T^2}\right)$ by Azuma-Hoeffding applied to $M_i^T$.
	Also, agent $i$ will bid at least $\alpha_i T - (\sqrt{T\log T})$ times on some event $E_2$ or probability at least $1-O\left(\frac{1}{T^2}\right)$ by standard Chernoff bounds.
	Let $E = E_1\cap E_2$.
        We lower bound agent $i$'s utility on $E$.
	Appealing to \cref{lem:bernoulliworstcasevaluedistribution}, we assume without loss of generality that agent $i$ has a $\mathrm{Bernoulli}(\alpha_i)$ value distribution.
	Then, the $\alpha_i$-aggressive strategy is simply the strategy of bidding when $V_i[t]=1$ if there is budget left.
	Agent $i$'s utility can then be bounded below on $E$ as
	\begin{equation*}
		\begin{split}
			\frac1T\sum_{t=1}^T U_i[t] & = \frac1T\sum_{t=1}^T V_i[t]b_i^t\pmb1\{i^t=i\} = \frac1T\sum_{t=1}^T b_i^t\pmb1\{i^t=i\}\\
            & = \frac1T\sum_{t=1}^T b_i^t - \frac1T\sum_{t=1}^Tb_i^t\pmb1\{i\neq i^t\} \geq \left(1 - \alpha_i\bar p(1-\alpha_i)\right)\alpha_i - O\left(\sqrt{\frac{\log T}{T}}\right).
		\end{split}
	\end{equation*}
	Because the $\alpha_i$-ideal utility under a $\mathrm{Bernoulli}(\alpha_i)$ value distribution is $\alpha_i$, this establishes the lemma.
\end{proof}

\subsection{Completion of the Proof of \texorpdfstring{\cref{lem:key_lemma}}{Lemma 5.3}}
\label{ssec:appendix_check_schur_convexity_proof}
What is left to verify is that $f(\cdot, (y_j)_{j\notin I})$ is Schur-concave for each $(y_j)_{j\notin I}$ and $f((x_i)_{i\in I}, \cdot)$ is Schur-convex for each $(x_i)_{i\in I}$. Notice that both $f(\cdot, (y_j)_{j\notin I})$ and $f((x_i)_{i\in I}, \cdot)$ are symmetric functions. By the Schur–Ostrowski criterion, it suffices to show that
\begin{align*}
    (x_{i_1} - x_{i_2})\left(\frac{\partial f}{\partial x_{i_1}} - \frac{\partial f}{\partial x_{i_2}}\right)\leq 0
\end{align*}
and
\begin{align*}
    (y_{j_1} - y_{j_2})\left(\frac{\partial f}{\partial y_{j_1}}- \frac{\partial f}{\partial y_{j_2}}\right)\geq 0
\end{align*}
for all $((x_i)_{i\in I}, (y_j)_{j\notin I})\in K$.

First, let us prove the following lemma.
\begin{lemma}
	\label{lem:sumconvex}
	Let $\phi:[0,b]\to[0,\infty]$ be convex and nondecreasing with $\phi(0) = 0$.
	For any nonnegative $x_1, \dots, x_n$ with $\sum_{i=1}^n x_i\leq b$,
	\begin{equation*}
		\sum_{i=1}^n \phi(x_i) \leq \phi\left(\sum_{i=1}^nx_i\right).
	\end{equation*}
\end{lemma}
\begin{proof}
	Define $\bar x = \sum_{i=1}^n x_i$, and consider the optimization problem below.
	\begin{align*}
		\max_{(y_i)}     & \sum_{i=1}^n \phi(y_i) \nonumber                             \\
		\text{s.t.}\quad & \sum_{i=1}^n y_i \leq\bar x                                \\
		\quad\quad       & y_i\geq 0                        & \forall i\in[n] \nonumber
	\end{align*}
	Since $\phi$ is convex, an optimal solution $(y_i^\star)$ lies on an extreme point of the feasible region.
	Since the feasible region is a polytope defined by $n+1$ constraints, at least $n$ of them must be tight at an extreme point.
	Since $\phi$ is nondecreasing, this implies that there exists a unique $i^\star$ such that $y^\star_{i^\star} = \bar x$ and all other $y^\star_i=0$.
	Notice that $(x_i)$ is a feasible solution to the optimization problem. It follows that
	\begin{equation*}
		\sum_{i=1}^n \phi(x_i) \leq \sum_{i=1}^n \phi(y^\star_i) = \phi(y^\star_{i^\star}) + \sum_{i\neq i^\star}\phi(y^\star_i)= \phi(\bar x) = \phi\left(\sum_{i=1}^n x_i\right).
	\end{equation*}
\end{proof}

To verify that $f(\cdot, (y_j)_{j\notin I})$ is Schur-concave for each $(y_j)_{j\notin I}$, compute
\begin{equation*}
	\begin{split}
		(x_{i_1} - x_{i_2})\left(\frac{\partial f}{\partial x_{i_1}} - \frac{\partial f}{\partial x_{i_2}}\right) = \frac12(x_{i_1} - x_{i_2})^2\prod_{i\in I\setminus\{i_1, i_2\}}(1-x_i)\left(-2 + 2\prod_{j\notin I}(1-y_j)\right. \\
		\left. + \prod_{j\notin I}(1-y_j)\sum_{i\in I\setminus\{i_1, i_2\}}\frac{x_i}{1-x_i}(-1 + X)\right) \leq 0
	\end{split}
\end{equation*}
where the inequality comes from the inequalities $-2 + 2\prod_{j\notin I}(1-y_j) \leq 0$ and $-1 + X \leq 0$.

To verify that $f((x_i)_{i\in I}, \cdot)$ is Schur-convex for each $(x_i)_{i\in I}$, compute
\begin{equation}
	\label{eq:verify_schur_convexity}
	\begin{split}
		(y_{j_1} - y_{j_2})\left(\frac{\partial f}{\partial y_{j_1}}- \frac{\partial f}{\partial y_{j_2}}\right) = \frac12\prod_{i\in I}(1-x_i)\prod_{j\notin I\cup\{j_1, j_2\}}(1-y_j)\left(2X - \sum_{i\in I}\frac{x_i}{1-x_i}(1-X)\right).
	\end{split}
\end{equation}
By \cref{lem:sumconvex} applied to the function $x\mapsto \frac{x}{1-x}$ on $[0,1]$,
\begin{equation*}
	2X - \sum_{i\in I}\frac{x_i}{1-x_i}(1-X) \geq 2X - \frac{\sum_{i\in I}x_i}{1 - \sum_{i\in I}x_i}(1-X) = 2X - \frac{X}{1-X}(1-X) = X \geq 0.
\end{equation*}
Therefore, \eqref{eq:verify_schur_convexity} is nonnegative so $f((x_i)_{i\in I}, \cdot)$ is Schur-convex.

\subsection{Deferred Proofs of \texorpdfstring{\cref{prop:assymetric_non_robustness,lem:anticorrelatedbidding}}{Proposition 5.4 and Lemma 5.5}}
\label{ssec:appendix_upper_bound_proofs}
Notice that \cref{prop:assymetric_non_robustness} is a special case of \cref{lem:anticorrelatedbidding}, so we only need to prove \cref{lem:anticorrelatedbidding}.
\begin{proof}[Proof of \cref{lem:anticorrelatedbidding}]
	The other agents $j\neq i$ will use the following strategy.
	At each time $t$, if every $j\neq i$ has budget remaining, either no $j\neq i$ will bid or a single agent $j$ will bid, where an agent $j$ bids with probability $\alpha_j$ (and therefore no agent will bid with probability $1-\sum_{j\neq i}\alpha_j = \alpha_i$).
	Their strategy will be independent across times, but notice that the agents' bidding are very much not independent across agents.

	For any single agent $j$, their bids are i.i.d.
	$\mathrm{Bernoulli}(\alpha_j)$ across time conditioned on them having budget remaining.
	Let $E$ be the event that no agent $j\neq i$ runs out of budget.
	For the same reason as in \cref{lem:dontrunoutofmoney}, the probability that agent $j$ runs out of budget is at most $O(1/T^2)$, so $\Pr(E) \geq 1 - O(1/T^2)$.
	Using the idea and notation of the imaginary game where budgets are not enforced as introduced in \cref{ssec:appendix_equilibrium_proofs}, agent $i$'s expected utility can be bounded as
	\begin{equation*}
		\label{eq:robustnessimaginarygamedifference}
		\begin{split}
			\frac1T\sum_{t=1}^T \mathbb E[U_i[t]] \leq \frac1T\sum_{t=1}^T \mathbb E[\tilde U_i[t]] + \frac1T\sum_{t=1}^T \mathbb E[V_i[t]\pmb1_{E^c}] \leq \frac1T\sum_{t=1}^T \mathbb E[\tilde U_i[t]] + O\left(\frac{1}{T^2}\right).
		\end{split}
	\end{equation*}
	Thus, we just need to bound the imaginary utility $\frac1T\sum_{t=1}^T \mathbb E[\tilde U_i[t]]$.

	By the strategy of the other agents, (we use the notation for $\beta$-ideal utility as in \cref{def:ideal_utility})
	\begin{equation*}
		\begin{split}
			\frac1T\sum_{t=1}^T \mathbb E[\tilde U_i[t]] & = \frac1T\sum_{t=1}^T \mathbb E[V_i[t]b_i^t\pmb1\{i^t=i\}] = \frac1T\sum_{t=1}^T\left(\sum_{j\neq i}\alpha_jp_j^{\{i,j\}} + \alpha_i\right)\mathbb E[V_i[t]b_i^t] \\
			                                      & \leq \left(\sum_{j\neq i}\alpha_jp_j^{\{i,j\}} + \alpha_i\right)v_i^\star((1+\delta_i^T)\alpha_i)                                                    \\
			                                      & \leq \left(\sum_{j\neq i}\alpha_jp_j^{\{i,j\}} + \alpha_i\right)v_i^\star(\alpha_i) + O\left(\sqrt{\frac{\log T}{T}}\right),
		\end{split}
	\end{equation*}
	using \cref{lem:generalized_knapsack_problem_over_time} for the first inequality and \cref{lem:idealutilitylipschitz} for the second.
	The above and \eqref{eq:robustnessimaginarygamedifference} imply the lemma statement.
\end{proof}

\subsection{Deferred Proof of \texorpdfstring{\cref{thm:robustnesshardness}}{Theorem 5.6}}
\label{ssec:appendix_convex_combination_upper_bound_proof}
\begin{proof}[Proof of \cref{thm:robustnesshardness}]
	First, compute
	\begin{equation*}
		\begin{split}
			\sum_{i=1}^n\sum_{j\neq i}\alpha_i\alpha_jp_i^{\{i,j\}} =  \sum_{i=1}^n\sum_{j\neq i}\alpha_i\alpha_j(1 - p_j^{\{i,j\}}) = \sum_{i=1}^n \sum_{j\neq i}\alpha_i\alpha_j - \sum_{i=1}^n \sum_{j\neq i}\alpha_i\alpha_j p_j^{\{i,j\}} \\
			= 1- \sum_{i=1}^n \alpha_i^2 - \sum_{i=1}^n \sum_{j\neq i}\alpha_i\alpha_jp_i^{\{i,j\}}.
		\end{split}
	\end{equation*}
	Solving for $\sum_{i=1}^n\sum_{j\neq i}\alpha_i\alpha_jp_i^{\{i,j\}}$,
	\begin{equation*}
		\sum_{i=1}^n\sum_{j\neq i}\alpha_i\alpha_jp_i^{\{i,j\}} = \frac12-\frac12\sum_{i=1}^n\alpha_i^2.
	\end{equation*}
	Using the above and \cref{lem:anticorrelatedbidding},
	\begin{equation*}
		\begin{split}
			\lambda = \lambda\sum_{i=1}^n\alpha_i & \leq \sum_{i=1}^n \left(\alpha_i + \sum_{j\neq i}\alpha_j p_i^{\{i,j\}}\right)\alpha_i + O\left(\sqrt{\frac{\log T}{T}}\right) \\
			                                      & = \sum_{i=1}^n \alpha_i^2 + \left(\frac12 - \frac12\sum_{i=1}^n\alpha_i^2\right) + O\left(\sqrt{\frac{\log T}{T}}\right)       \\
			                                      & = \frac12 + \frac12\sum_{i=1}^n\alpha_i^2 + O\left(\sqrt{\frac{\log T}{T}}\right).
		\end{split}
	\end{equation*}

\end{proof}

\subsection{Deferred Proofs from \texorpdfstring{\cref{sec:computation}}{section 6}}
\label{ssec:appendix_computation_proofs}

\crefname{algorithm}{Algorithm}{Algorithms} 
\Crefname{algorithm}{Algorithm}{Algorithms} 

\floatname{algorithm}{Algorithm}

\begin{proof}[Proof of \cref{lem:computation_lemma}]
The runtime claim is easily checked.

Let $\mathcal S_i^{(k)} = \{S\in\mathcal S^{(k)}:i\in S\}$. By the Chernoff bound, for any iteration $k$, agent $i$, and $0 < \epsilon < 1$,
\begin{equation}
\label{eq:sufficient_samples_probability}
	\Pr_{\mathcal S^{(k)}\sim\mu^C}\left(|\mathcal S_i^{(k)}| \leq (1-\epsilon)\alpha_i C\right)\leq \exp\left(-\frac{\epsilon^2C\alpha_i}{2}\right).
\end{equation}
By the Hoeffding bound, for any fixed allocation probabilities $p_i^S$ and a fixed iteration $k$,
\begin{equation*}
\begin{split}
	\Pr_{\mathcal S^{(k)}\sim \mu^C} & \left(\left|\frac{1}{|\mathcal S_i^{(k)}|}\sum_{S\in |\mathcal S_i^{(k)}|}p_i^S - \E_{S\sim \mu}[p_i^S\mid i\in S]\right|\geq \delta\,\middle|\,|\mathcal S_i^{(k)}|\right)\leq 2\exp\left(-2|\mathcal S_i^{(k)}|\delta^2\right).
\end{split}
\end{equation*}
Therefore,
\begin{equation*}
	\Pr_{\mathcal S^{(k)}\sim \mu^C} \left(\left|\frac{1}{|\mathcal S_i^{(k)}|}\sum_{S\in |\mathcal S_i^{(k)}|}p_i^S - \E_{S\sim \mu}[p_i^S\mid i\in S]\right|\geq \delta\right) \leq 2\exp(-2(1-\epsilon)\alpha_i C\delta^2) + \exp\left(-\frac{\epsilon^2C\alpha_i }{2}\right).
\end{equation*}

Let $E_1$ be the event where no LP at any iteration is infeasible. Using the allocation probabilities $p_i^S$ that satisfy \eqref{eq:equilibrium_utility_condition_repeated} and \eqref{eq:robustness_condition_repeated} in the above equation, and the union bound over iterations $k\in[K]$ and agents $i\in[n]$, the above implies that
\begin{equation}
\label{eq:lp_infeasibility_probability}
	\Pr(E_1^c) \leq \sum_{i=1}^n\left(2K\exp(-2(1-\epsilon)\alpha_i C\delta^2) + K\exp\left(-\frac{\epsilon^2C\alpha_i }{2}\right)\right).
\end{equation}

By assumption, each $(p_i^{S^*})^{(k)}$ for $k\in[K]$ satisfies $\mathcal F(S^*)$. Since $\mathcal F(S^*)$ just consists of linear constraints, the average $\frac1K\sum_{k=1}^K(p_i^{S^*})^{(k)}$ satisfies $\mathcal F(S^*)$. \cref{alg:compute_allocation_probabilities} sets $(p_i^{S^*}) = \frac1K\sum_{k=1}^K(p_i^{S^*})^{(k)}$ on $E_1$. We now show that this output approximately satisfies \eqref{eq:interim_allocation_constraint}.

On $E_1$, the multiplicative weights guarantee for solving linear programs, \cref{thm:multiplicative_weights_lp}, says for every $i$,
\begin{equation}
\label{eq:sampled_interim_allocation_probability}
	\left|\frac1K\sum_{k=1}^K\frac{1}{|\mathcal S_i^{(k)}|}\sum_{S\in\mathcal S_i^{(k)}} (p_i^S)^{(k)} - \left(1 - \prod_{j\in [n]}(1-\alpha_j)\right)\right| \leq \delta + \left(\eta + \frac{\ln 2n}{\eta K}\right).
\end{equation}

For a fixed iteration $k$ and a randomly chosen $S$, $(p_i^S)^{(k)}$ is a random variable whose law only depends on $y^{(k)}$ and $S$. Since $\mathcal S^{(k)}$ is independent of $y^{(k)}$, conditioning on $y^{(k)}$, the values of $(p_i^S)$ for $S\in\mathcal S^{(k)}$ are just i.i.d. samples from the distribution of the $(p_i^S)^{(k)}$ for $S\sim\mu$. This is a distribution over $[0,1]$ with mean $\E_{S\sim\mu}[(p_i^S)^{(k)}\mid y^{(k)}] = \E_{S\sim\mu}[(p_i^S)^{(k)}\mid (\mathcal S^{(k')})_{k'\in[K]}]$.

Therefore, by Hoeffding's inequality, for any $\delta'>0$ and fixed $i$,
\begin{equation*}
\begin{split}
	\Pr_{\mathcal S^{(k)}\sim \mu^C} \left(\left|\frac{1}{|\mathcal S_i^{(k)}|}\sum_{S\in\mathcal S_i^{(k)}}(p_i^S)^{(k)} - \E_{S\sim \mu}\left[(p_i^S)^{(k)}\mid i\in S, (\mathcal S^{(k')})_{k'\in[K]}, |\mathcal S_i^{(k)}|\right]\right| \geq \delta'\,\middle|\,|\mathcal S_i^{(k)}|, y^{(k)}\right)\\
    \leq 2\exp(-2|\mathcal S_i^{(k)}|\delta'^2).
\end{split}
\end{equation*}
By \eqref{eq:sufficient_samples_probability}, this implies, for fixed $i$,
\begin{equation*}
\begin{split}
	\Pr_{\mathcal S^{(k)}\sim \mu^C}\left(\left|\frac{1}{|\mathcal S_i^{(k)}|}\sum_{S\in\mathcal S_i^{(k)}}(p_i^S)^{(k)} - \E_{S\sim \mu}\left[(p_i^S)^{(k)}\mid i\in S, (\mathcal S^{(k')})_{k'\in[K]}\right]\right| \geq \delta'\,\middle|\,(\mathcal S^{(k')})_{k'\in[K]}\right)\\
	\leq2\exp(-2(1-\epsilon)\alpha_i C\delta'^2) + \exp\left(-\frac{\epsilon^2C\alpha_i}{2}\right).
\end{split}
\end{equation*}
Let $E_2$ be the event that the above does not happen for any iteration $k$ and any agent $i$, so by the union bound over iterations $k\in[K]$ and agents $i\in[n]$ and the above,
\begin{equation*}
	\Pr(E_2^c)\leq \sum_{i=1}^n\left(2K\exp(-2(1-\epsilon)\alpha_i C\delta'^2) + K\exp\left(-\frac{\epsilon^2C\alpha_i}{2}\right)\right)
\end{equation*}

Let $E = E_1\cap E_2$. On $E$, by \eqref{eq:sampled_interim_allocation_probability}, the definition of $E_2$, and the triangle inequality,
\begin{equation*}
\begin{split}
	\left|\E_{\substack{S\sim\mu}}\left[\frac1K\sum_{k=1}^K(p_i^S)^{(k)}\,\middle|\, i\in S, (\mathcal S^{(k')})_{k'\in[K]}\right] - \left(1 - \prod_{j\in [n]}(1-\alpha_j)\right)\right| \leq \delta' + \delta + \left(\eta + \frac{\ln 2n}{\eta K}\right).
\end{split}
\end{equation*}
By the union bound,
\begin{equation*}
	\Pr(E^c) \leq \sum_{i=1}^n\left(2K\exp(-2(1-\epsilon)\alpha_i C\delta^2) + 2K\exp(-2(1-\epsilon)\alpha_i C\delta'^2) + 2K\exp\left(-\frac{\epsilon^2C\alpha_i}{2}\right)\right).
\end{equation*}

Substituting the choice of parameters in \eqref{eq:computation_lemma_parameter_choice} along with $\epsilon = \min\left\{\frac12,\sqrt{\frac{6\ln nK}{\underline \alpha C}}\right\}$ and $\delta' = \delta = \sqrt{\frac{3\ln nK}{\underline \alpha C}}$ in the above two displays, we obtain the lemma statement.
\end{proof}

\begin{proof}[Proof of \cref{thm:compute_dual_vectors_once}]
Conditioned on the initially computed samples $\mathcal S^{(k)}$ and on the sets $S^t$ bidding at each time, at any time $t$, if $S^t=S$, agent $i$ wins the item with probability
\begin{equation}
	p_i^S := \E_{k\sim [K]}\left[(p_i^S)^{(k)}\mid (\mathcal S^{(k')})_{k'\in[K]}\right] = \frac1K\sum_{k=1}^K (p_i^S)^{(k)}.
\end{equation}
Thus, running \cref{alg:compute_BRB} is equivalent to running \BRB with preset allocation probabilities $p_i^S$. Since $\frac1K\sum_{k=1}^K (p_i^S)^{(k)}$ is exactly the output of \cref{alg:compute_allocation_probabilities} when run with bidding agents $S$, we can use \cref{lem:computation_lemma} to see that
\begin{equation*}
	\left|\E_{\substack{S\sim\mu}}\left[p_i^{S}\,\middle|\, i\in S, (\mathcal S^{(k)})_{k\in [K]}\right] - \left(1 - \prod_{j\in [n]}(1-\alpha_j)\right)\right|\leq O\left(\sqrt{\frac{\log n}{K}}\right)\leq O\left(\sqrt{\frac{\log T}{T}}\poly(n, (1/\alpha_j)_{j\in[n]})\right)
\end{equation*}
with probability at least $1 - O(1/n^2T^2) \geq 1-O(\poly(n, (1/\alpha_j)_{j\in[n]})/T^2)$. We obtain the high probability equilibrium utility guarantee by applying \cref{thm:appendix_nash}.

The robustness guarantee follows from the fact that the each $(p_i^S)^{(k)}$ satisfies $\mathcal F(S)$ and that $\mathcal F(S)$ is convex, so the expectation of the $(p_i^S)^{(k)}$ satisfies $\mathcal F(S)$, which suffices to get the robustness guarantee from \cref{lem:bangforbuck}.

\end{proof}

\crefname{algorithm}{Mechanism}{Mechanisms} 
\Crefname{algorithm}{Mechanism}{Mechanisms} 

\bibliographystyle{siamplain}
\bibliography{bibliography}
\end{document}